\documentclass[runningheads]{llncs}
\usepackage[T1]{fontenc}
\usepackage{graphicx}
\usepackage{amsmath,amssymb,amsfonts}
\usepackage{proof}
\usepackage{tikz-cd}
\usepackage{multicol}
\usepackage{multirow}
\usepackage{enumitem}
\usepackage{subcaption}
\usepackage[normalem]{ulem}
\usepackage{mathrsfs}
\usepackage{url}
\usepackage{color}
\usepackage[linktoc=all,hidelinks,colorlinks=true,linkcolor=blue,citecolor=blue]{hyperref}
\usepackage{cleveref}
\usetikzlibrary{positioning}
\usepackage{todonotes}
\usepackage{orcidlink}
\usepackage{bussproofs}
\usepackage{mathtools}
\usepackage{cmll}
\usepackage{comment}
\usepackage{array}
\usepackage{booktabs}
\usepackage{appendix}

\newcommand{\JK}{\sf{JK}}
\newcommand{\WJK}{\sf{WJK}}

\newcommand{\SJK}{\sf{SJK}}
\newcommand{\E}{\mathcal{E}}

\usepackage{macros}

\makeatletter %
\def\@acknow{}%
\long\def\EarlyAcknow#1 \par{%
\def\@acknow{\abstractfont\abstracthead*{Acknowledgments}%
#1\par}}%

\def\printabstract{\ifx\@acknow\empty\else\@acknow\fi\par%
    \ifx\@abstract\empty\else\@abstract\fi\par}
\makeatother

\begin{document}

\title{A generalised and modular proof-theoretic semantics for sequent systems}

\author{Victor Barroso-Nascimento\orcidlink{0000-0002-3990-5996} 
\and Arghya Chakrabarty\orcidlink{0009-0007-4529-4419} 
\and 
Elaine Pimentel\orcidlink{0000-0002-7113-0801}}

\authorrunning{V. Barroso-Nascimento et al.}

\institute{Department of Computer Science, University College London, UK\\ \email{\{v.nascimento,arghya.chakrabarty.20,e.pimentel\}@ucl.ac.uk}
}

\maketitle

\begin{abstract}

Base-extension semantics explains logical consequence via derivability in atomic systems. Extending this approach across logics requires careful control of structural principles and contexts. We present a modular sequent-based framework in which support is indexed by multisets of atomic resources and the semantic clauses for connectives are guided by their left introduction rules. The framework separates additive from multiplicative behaviour and makes explicit the contributions of structural rules and restrictions on sequents. We establish soundness and completeness for classical and intuitionistic propositional logic and for multiplicative-additive fragments of classical and intuitionistic linear logic. 
We then investigate axiomatic extensions of intuitionistic logic through Jankov logic. Two classes of atomic bases yield sound and complete semantics. Completeness is proven using two novel and modular methods using simulation bases, one of which is constructive. These results provide a common semantic treatment of several logical systems while identifying the additional conditions needed to extend the framework beyond structural variations.
\end{abstract}

\section{Introduction}\label{sec:intro}
In model-theoretic semantics, mathematical structures help in supporting the notion of validity, which is based on a notion of {\em truth}. In the case of intuitionistic logic, for example, one could use Kripke structures, where the validation of atomic propositions using the classical notion of truth (\eg, via truth tables) is enough for describing the meaning of a logical connective. %

\noindent
Proof-theoretic semantics~\cite{Prawitz-2006,pts-91,schroeder2006validity,sep-proof-theoretic-semantics} ($\Pts$), on the other hand, takes {\em inference} as its starting point: logical connectives acquire their meaning through rules governing how statements containing them are established and used. 
The significance of $\Pts$ lies in this connection between meaning and justification: it seeks to explain why the inferences we make are warranted by the meaning of the expressions involved. This perspective is particularly well-suited for {\em natural deduction} and {\em intuitionistic logic}, where proofs express the constructive content of propositions.

\noindent
However, $\Pts$ does not readily extend in a modular way to other proof systems or logical traditions.
 In fact, although several extensions and adaptations have been proposed for other logical systems, there is still no modular methodology capable of systematically assigning proof-theoretic semantics to a broad family of logics. Each adaptation typically requires ad hoc constructions reflecting the specific proof theory of the logic under consideration~\cite{Sandqvist2015IL,Sandqvist,DBLP:journals/sLogica/GheorghiuGP26,stafford2025logics,Daragona-Prawitz,DBLP:journals/rsl/Ayhan26}.

\noindent
{\bf Shift to sequent systems.} A first step towards such a unification was taken in~\cite{DBLP:conf/tableaux/BarrosoNascimentoPP25}, where a proof-theoretic semantics for {\em classical logic} was proposed based on sequent systems rather than natural deduction. The resulting framework fits classical reasoning remarkably well. Indeed, one may argue that sequent calculi play, for classical logic, the same foundational role that natural deduction plays for intuitionistic logic \cite{sep-proof-theoretic-semantics}. Rather than viewing sequents merely as technical devices for proving meta-theoretical properties such as cut elimination, this approach treats them as the primary carriers of meaning, allowing logical constants to be interpreted directly through their sequent rules.

\noindent
The present paper develops this idea substantially. We argue that the sequent-based perspective provides a uniform and modular framework for proof-theoretic semantics that extends naturally beyond classical logic. Our framework encompasses not only classical and intuitionistic logic, but also substructural logics and, more interestingly, intermediate logics, showing that the same semantic machinery can be instantiated over different structural bases. In this sense, the framework isolates the contribution of the logical rules from that of the structural rules, making explicit how different notions of logical consequence arise from different choices of base systems.

\noindent
{\bf Base-extension semantics.} To achieve this, we make use of a strand of $\Pts$ called base-extension semantics~\cite{Sandqvist2015IL} ($\Bes$), where proof-theoretic validity is defined relative to a given collection of inference rules regarding basic formulas of the language. More specifically, in $\Bes$ the characterisation of consequence is given by an inductively defined semantic judgment whose base case is given by provability in an {\em atomic system} (or a {\em base}), which is a collection of rules involving only atomic formulas. 

\noindent
As a simple example, fix an integer $n$ and, for $k \in \{2,3,6\}$, let $p_k$ represent the atomic sentence ``$n$ is divisible by $k$''. Consider, informally, the rules 
\[ \frac{p_6}{p_2} \qquad \frac{p_6}{p_3} \qquad \frac{p_2 \qquad p_3}{p_6} \] 
These rules specify both consequences of asserting $p_6$ and sufficient grounds for asserting it: $p_6$ licenses each of $p_2$ and $p_3$, while $p_2$ and $p_3$ jointly license $p_6$. The internal arithmetic structure of the sentences is not represented in the formal language: only the inferential relationships among the atoms are specified. 

\noindent
After defining validity for atoms, one can define validity for logical connectives via semantic clauses that express proof conditions, which results in a framework that evaluates propositions exclusively in terms of proofs of its constituents.

\noindent
{\bf Generalising sequent $\Bes$.} Starting from the classical sequent-based semantics developed in~\cite{DBLP:conf/tableaux/BarrosoNascimentoPP25}, we transition to substructural logics by presenting separate additive and multiplicative semantic clauses and controlling the use and multiplicity of assumptions through the atomic rules. Substructurality is thus reflected directly in the conditions imposed on bases and their admissible extensions.

\noindent
The same framework also accommodates intuitionistic logic through the familiar restriction that right-hand contexts contain at most one formula. At the atomic level, this requires every premise and conclusion of every rule to have at most one atom in its succedent, with extensions preserving this condition. Retaining the appropriate weakening and contraction principles, and choosing the additive interpretation of disjunction and falsity, yields intuitionistic semantics. Combining restrictions on sequents with restrictions on structural rules similarly leads to substructural semantics. In particular, taking structural rules out of classical bases yields a semantics for the multiplicative-additive fragment of linear logic~\cite{girard87tcs} ($\MALL$), and taking them out of intuitionistic bases yields one for intuitionistic $\MALL$  ($\iMALL$). The same ideas also work, to some extent, for intermediate logics. These variations arise within a common semantic framework, making explicit the respective contributions of the logical clauses, the structural rules, and the shape of contexts.

\noindent
{\bf Modularity.} A remarkable property of our proof strategies is that they do not need to be changed in order to work for different logics. Validity of weakening is tied to closure under atomic weakening, the same holding for contraction, and by taking structural rules from bases the classical and intuitionistic proofs immediately become proofs for $\MALL$ and $\iMALL$. Moreover, the two proofs of completeness presented for the intermediate Jankov logic~\cite{Yankov1968,CitkinVandoulakis2022,Indrzejczak2022} can be adapted to {\em any intermediate logic}, obtained by adding axioms to $\IPL$, although soundness is still ad hoc.

\noindent
Another interesting consequence is that intuitionistic logic may be obtained by restricting the semantics of classical logic, but restricting the intuitionistic semantics also yields a different classical semantics. The same holds for $\MALL$ and $\iMALL$. This provides insights on the reason why intuitionistic logic may be visualised either as an extension or as a restriction of classical logic \cite{LegrisMolina2001}.

\noindent
{\bf The relevance.}
While it may be tempting to expect that the familiar syntactic restrictions ensuring nice proof-theoretic properties such as cut elimination immediately yield soundness and completeness with respect to the corresponding semantics, this turns out not to be the case. The interaction between bases and their extensions introduce subtleties that do not arise in the purely classical setting. For instance, despite being sequent-based, the semantics in \cite{DBLP:conf/tableaux/BarrosoNascimentoPP25} {\em does not produce} intuitionistic logic when restricted to single succedents. As showcased by our proofs for Jankov logic, other subtle features, such as whether rules can be used with arbitrary contexts or not, also have important semantic bearing.

\noindent
Beyond providing a correct and uniform semantic account, our framework also sheds new light on several foundational questions in proof theory. In particular, it offers a new perspective on the distinction between intuitionistic and classical reasoning~\cite{Barroso-Nascimento2025}. Since both emerge as different instantiations of the same semantic construction, the framework separates the role of structural principles from the behaviour of the logical connectives themselves. This perspective resonates with recent discussions on the proof-theoretic characterisation of intuitionistic consequence, including the claim that multiple-conclusion formulations of intuitionistic logic may obscure its genuine proof-theoretic nature \cite{Steinberger2011}.

\section{Base-extension semantics}\label{sec:bes}
$\Bes$ is founded on an inductively defined judgment called {\em support}, which mirrors the syntactic structure of formulas. The inductive definition begins with a base case: the support of atomic propositions is determined by derivability in a given {\em base} -- a specified collection of inference rules that govern atomic propositions. Sandqvist~\cite{Sandqvist2015IL} introduced a sound and complete formulation of $\Bes$ for intuitionistic propositional logic ($\IPL$).
In this work, we adopt Sandqvist's~\cite{Sandqvist2015IL} terminology, adapting it to the sequent setting as presented in~\cite{DBLP:conf/tableaux/BarrosoNascimentoPP25}. 

\subsection{Atomic Derivability}
Let $\At$ denote the set of atoms, and let $\Gamma_{\At},\Delta_{\At},\Theta_{\At},\Sigma_{\At},\ldots$ range over finite multisets of atoms. Commas denote multiset union, with multiplicities added. In particular, $\Gamma_{\At},p$ denotes $\Gamma_{\At}\uplus{p}$. An \emph{atomic sequent} is an expression $\Gamma_{\At}\Rightarrow\Delta_{\At}$, where either multiset may be empty.

\begin{definition}[Base]
 An {\em atomic system} (a.k.a. a {\em base}) $\B$ is a set of atomic rules of the form
\begin{prooftree}
\AxiomC{$\Gamma_{\At}^1 \Rightarrow \Delta^{1}_{\At}$}
\AxiomC{$\ldots$}
\AxiomC{$\Gamma_{\At}^n \Rightarrow \Delta^{n}_{\At}$}
\TrinaryInfC{$\Theta_{\At} \Rightarrow \Sigma_{\At}$}
\end{prooftree}
The sequence of sequents $\Gamma^{m}_{\At} \Rightarrow \Delta^{m}_{\At}$  %
might be empty, in which case we call the rule a {\em zero-level rule}. %
\end{definition}
Notice that the multisets of atoms occurring in the sequents of the rule might not be related to each other. We can have rules in which some of them are the same (\eg, a rule with $\Gamma^{1}_{\At} = \Theta_{\At}$), but no identity between the multisets is built into the definition. This definition is quite general and quite useful in handling, for instance, additive contexts, see \eg~\cite{DBLP:conf/tableaux/BarrosoNascimentoPP25}.

When defining bases, we distinguish between \emph{individual rules} and \emph{rule schemata}. In standard presentations of sequent calculi, rules are given by schemata whose metavariables range over formulas and contexts, with instantiation governed by the stated side conditions. In our setting, each individual rule in a base has {\em fixed atomic sequents} as its premises and conclusion. A displayed structural rule schema describes a family of such rules, obtained by varying the atoms and atomic contexts subject to the specified restrictions.

Applying an individual rule requires derivations of its specified premise sequents, with contexts matching exactly, including the multiplicities of atoms. The metavariables in a schema therefore describe different members of a family of rules. They do not permit the premises or conclusion of a rule already fixed in a base to vary during its application.

\begin{definition}[Deducibility] For every base $\B$, the relation $\vdash_{\B}$ is defined as the least relation generated by the following clauses:

\begin{enumerate}
    \item If $\Gamma_{\At} \Rightarrow \Delta_{\At}$ is the conclusion of a zero-level rule in $\B$, then $\vdash_\B \Gamma_{\At} \Rightarrow \Delta_{\At}$ holds;
    \item Assume $\B$ contains a rule with the following shape
\begin{prooftree}
\AxiomC{$\Gamma_{\At}^1 \Rightarrow \Delta^{1}_{\At}$}
\AxiomC{$\ldots$}
\AxiomC{$\Gamma_{\At}^n \Rightarrow \Delta_{\At}^{n}$}
\TrinaryInfC{$\Theta_{\At} \Rightarrow \Sigma_{\At} $}
\end{prooftree}
Then, if $ \vdash_{\B} \Gamma_{\At}^{m} \Rightarrow \Delta^{m}_{\At}$ holds for all $1 \leq m \leq n$, so does $\vdash_{\B} \Theta_{\At} \Rightarrow \Sigma_{\At}$.
\end{enumerate}
\end{definition}
Derivations are hence finite trees built by applying rules of the base, as illustrated next.
\begin{example}
Consider a base $\B$ consisting exactly of the following rules
\begin{prooftree}
\AxiomC{}
\RightLabel{\tiny\tiny{$R1$}}
\UnaryInfC{$p \Rightarrow q$}
\quad
\DisplayProof
\quad
    \AxiomC{$p \Rightarrow q$}
    \RightLabel{\tiny\tiny{$R2$}}
    \UnaryInfC{$r \Rightarrow s$}
    \quad
    \DisplayProof
    \quad
        \AxiomC{$p, j \Rightarrow q$}
    \RightLabel{\tiny\tiny{$R3$}}
    \UnaryInfC{$r \Rightarrow s$}
    \quad
    \DisplayProof
    \quad
       \AxiomC{$r \Rightarrow s$}
    \RightLabel{\tiny\tiny{$R4$}}
    \UnaryInfC{$j \Rightarrow k$}
\end{prooftree}

The following is a deduction in $\B$ that establishes $\vdash_{\B} j \Rightarrow k$
\begin{prooftree}
    \AxiomC{}
\RightLabel{\tiny\tiny{$R1$}}
\UnaryInfC{$p \Rightarrow q$}
\RightLabel{\tiny\tiny{$R2$}}
\UnaryInfC{$r \Rightarrow s$}
\RightLabel{\tiny\tiny{$R4$}}
\UnaryInfC{$j \Rightarrow k$}
\end{prooftree}
Observe that the application of the rule $R2$ {\em cannot be replaced by} $R3$ since $\nvdash_{\B} p, j \Rightarrow q$. 
\end{example}
Unlike in traditional sequent calculi, contexts are not arbitrary (or cumulative). 
The definition of derivability does not implicitly allow contexts to be enlarged or repeated occurrences of atoms to be identified. This makes it possible to control weakening and contraction through the choice of atomic rules.

\begin{definition}[Classical, intuitionistic and linear bases]
We say that a base is \emph{closed under a set of rules} if it contains every instance of the corresponding rule schemata.

Consider the following set of atomic rules
\begin{prooftree}
\AxiomC{}
\noLine
\UnaryInfC{}
\RightLabel{\tiny\init}
\UnaryInfC{$p \Rightarrow p$}
\quad
\DisplayProof
\quad 
   \AxiomC{}
    \noLine
    \UnaryInfC{$\Gamma_{\At} \Rightarrow \Delta_{\At}$}
    \RightLabel{\tiny\LW}
    \UnaryInfC{$\Gamma_{\At}, p \Rightarrow \Delta_{\At}$}
    \quad
    \DisplayProof
    \quad
     \AxiomC{}
     \noLine
    \UnaryInfC{$\Gamma_{\At} \Rightarrow \Delta_{\At}$}
    \RightLabel{\tiny\RW}
    \UnaryInfC{$\Gamma_{\At} \Rightarrow p, \Delta_{\At}$}
\end{prooftree}
 \begin{prooftree}
       \AxiomC{}
    \noLine
    \UnaryInfC{$\Gamma_{\At}, p, p\Rightarrow \Delta_{\At}$}
    \RightLabel{\tiny\LC}
    \UnaryInfC{$\Gamma_{\At}, p \Rightarrow \Delta_{\At}$}
    \quad
    \DisplayProof
    \quad
     \AxiomC{}
     \noLine
    \UnaryInfC{$\Gamma_{\At} \Rightarrow p, p,  \Delta_{\At}$}
    \RightLabel{\tiny\RC}
    \UnaryInfC{$\Gamma_{\At} \Rightarrow p, \Delta_{\At}$}
\end{prooftree}
\begin{prooftree}
\AxiomC{$\Gamma_{\At} \Rightarrow \Sigma_{\At}, p$}
\AxiomC{$\Theta_{\At}, p \Rightarrow \Delta_{\At}$}
\RightLabel{\tiny\Cutc}
\BinaryInfC{$\Gamma_{\At}, \Theta_{\At} \Rightarrow \Sigma_{\At}, \Delta_{\At}$}
\quad
\DisplayProof
\quad
\AxiomC{$\Gamma_{\At} \Rightarrow p$}
\AxiomC{$\Theta_{\At}, p \Rightarrow \Delta_{\At}$}
\RightLabel{\tiny\Cuti}
\BinaryInfC{$\Gamma_{\At}, \Theta_{\At} \Rightarrow \Delta_{\At}$}
\end{prooftree}
A base is said to be
\begin{itemize}
\item {\em Classical} if it is closed under rules \init, \LW, \RW, \LC, \RC\ and \Cutc.

\item {\em Intuitionistic} if it is closed under rules \init, \LW, \RW, \LC, \Cuti. Moreover, every premise and conclusion of every rule in the base has at most one atom in its succedent;
\item {\em Linear} if it is closed under rules \init, \Cutc;
\item {\em Intuitionistic linear} if it is closed under rules \init, \Cuti,  and all sequents appearing in rules contain at most one atom in their succedents.
\end{itemize}
\end{definition}

\begin{definition}[Extensions]
An atomic system $\C$ is an {\em extension} of an atomic system $\B$ (written $\C\supseteq\B$), if $\C$ results from adding a (possibly empty) set of atomic rules to $\B$, satisfying the conditions defining the chosen class of bases.
\end{definition}

\subsection{Support}\label{sec:suppot}
We now extend atomic derivability to an inductively defined notion of semantic support ($\Vdash$). Following the perspective of Gentzen~\cite{gentzen1969} and Prawitz~\cite{Prawitz-2006}, according to which introduction rules define the meaning of logical constants and elimination rules are justified by these definitions, the semantics in~\cite{DBLP:conf/tableaux/BarrosoNascimentoPP25} formulated the clauses for the connectives in terms of their right introduction rules in the sequent calculus. 
For example, the support in a base $\B$ for the conjunction is defined in~\cite{DBLP:conf/tableaux/BarrosoNascimentoPP25} as 
\[
\Vdash_{\B} \Delta, A \wedge B
\quad\text{iff}\quad
\Vdash_{\B} \Delta, A
\;\text{ and }\;
\Vdash_{\B} \Delta, B.
\]
However, directly translating introduction rules into semantic clauses does not provide a uniform approach across logics. Indeed, for intuitionistic logic, the following definition yields an incomplete base-extension semantics, as shown in~\cite{piecha2015failure}.
\[
\Vdash_{\B} A \vee B
\quad\text{iff}\quad
\Vdash_{\B} A
\;\text{ or }\;
\Vdash_{\B} B
\]
A related issue arises on the proof-theoretic definition for substructural logics. For example, one could attempt to define support for the multiplicative conjunction $\otimes$ as
\[
\Vdash_{\B} \Delta_1,\Delta_2, A \otimes B
\quad\text{iff}\quad
\Vdash_{\B} \Delta_1, A
\;\text{ and }\;
\Vdash_{\B} \Delta_2, B.
\]
Obviously, this would not be sound, as derivability of $\Delta \vdash A \otimes B$ does not generally imply that $\Delta$ can be partitioned into multisets $\Delta_1,\Delta_2$ such that $\Delta_1 \vdash A$ and $\Delta_2 \vdash B$. 
Thus, the context splitting displayed by the right introduction rule cannot simply be read as a necessary condition for derivability. 

To obtain a uniform treatment on semantics, we therefore formulate the semantic clauses in terms of the connectives' left introduction rules, characterising their meaning through the consequences licensed by their use\footnote{We observe that, in~\cite{Sandqvist2015IL}, the semantic clause for disjunction was defined as
\[
\Vdash_{\B}A\vee B
\mbox{ iff }
\forall \mathcal{C}\supseteq\B, p\in\At, \mbox{ if }
A\Vdash_{\C}p \mbox{ and }
B\Vdash_{\C} p, \mbox{ then }
\Vdash_{\C} p.
\]
thus simulating the elimination rule for disjunction in the natural deduction approach.}. 

Compared with the support clauses in~\cite{DBLP:journals/corr/abs-2504-08349}, these clauses also incorporate the notion of \emph{support using resources} introduced in~\cite{DBLP:journals/sLogica/GheorghiuGP26}.
The central idea is to index support by a multiset of atomic resources ($\Vdash_{\B}^{\Gamma_{\At}}$): the base $\B$ specifies the available inference rules, while the resource multiset $\Gamma_{\At}$ records the atomic assumptions used to establish support, including their multiplicities. Semantic consequence must then account for how these resources pass from hypotheses to conclusions. 

More precisely, in this formulation, for a non-empty atomic context $\Theta_{\At}$, $\Theta_{\At} \Vdash_{\B}^{\Gamma_{\At}} A$ means that, in every extension of $\B$, whenever $\Theta_{\At}$ is supported using resources $\Sigma_{\At}$, the conclusion $A$ is supported using the combined multiset $\Gamma_{\At},\Sigma_{\At}$. Here, $\Gamma_{\At}$ records the additional resources needed to pass from $\Theta_{\At}$ to $A$. Unlike extending the base, adding resources does not need to preserve support: without weakening, an extra assumption cannot simply be discarded. Likewise, without contraction, one occurrence cannot freely replace two. This explicit accounting makes support sensitive to restrictions on structural rules and is therefore essential for extending the semantics to substructural logics.

To improve readability and clearly distinguish between multiplicative and additive definitions, we adopt the following notation:
$\with,\oplus,\to, \top$ and $0$ represent additive conjunction, disjunction, implication, true and false, respectively, whereas $\otimes,\parr,\multimap, 1$ and $\bot$ represent their multiplicative counterparts. Most clauses are similar to those presented in \cite{buzoku2026ILL}.

\begin{definition}[Support and logical validity]\label{def:addandmultsupport}
{\em Additive and multiplicative support} in $\B$ using resources $\Gamma_{\At}$, represented as $\Vdash^{\Gamma_{\At}}_{\B}$, is defined in Figure~\ref{fig:sup}. An inference from $\Gamma$ to $\Delta$ is {\em valid}, written as $\Gamma \Vdash \Delta$, if $\Gamma \Vdash^{\varnothing}_{\B} \Delta$ for all bases $\B$. 
\end{definition}

\renewcommand{\arraystretch}{1.3}
\begin{figure}[t!]
 \begin{tabular}{
    @{}
    >{\centering\arraybackslash}p{0.09\textwidth}
    >{\raggedright\arraybackslash}p{0.26\textwidth}
    >{\raggedright\arraybackslash}p{0.04\textwidth}
    >{\raggedright\arraybackslash}p{0.64\textwidth}
    @{}
}
\toprule
\multicolumn{4}{l}{\textbf{Atomic Support and Inference}} \\
\midrule

(At)
& $\Vdash^{\Gamma_{\At}}_{\B}\Delta_{\At}$
& iff
& $\vdash_{\B}\Gamma_{\At}\Rightarrow\Delta_{\At}$ \\

(Inf)
& $\{A^1,\ldots,A^n\}\Vdash^{\Gamma_{\At}}_{\B}\Delta$
& iff
& For all $\C\supseteq\B$, all $\Sigma^i_{\At}$ and all
$\Pi^i_{\At}$, where $1\leq i\leq n$, if
$\Vdash^{\Sigma^i_{\At}}_{\C}A^i,\Pi^i_{\At}$, then
$\Vdash^{\Gamma_{\At}\cup\Sigma^1_{\At}\cup\cdots\cup\Sigma^n_{\At}}_{\C}
\Delta,\Pi^1_{\At},\ldots,\Pi^n_{\At}$.
\\
\midrule
\multicolumn{4}{l}{\textbf{Additive Connectives}} \\
\midrule

($\&$)
& $\Vdash^{\Gamma_{\At}}_{\B}\Delta,A\&B$
& iff
& For all $\mathcal{C}\supseteq\B$, all $\Theta_{\At}$ and all
$\Sigma_{\At}$, if either
$A\Vdash^{\Theta_{\At}}_{\C}\Sigma_{\At}$ or
$B\Vdash^{\Theta_{\At}}_{\C}\Sigma_{\At}$, then
$\Vdash^{\Gamma_{\At}\cup\Theta_{\At}}_{\C}\Delta,\Sigma_{\At}$. \\

($\oplus$)
& $\Vdash^{\Gamma_{\At}}_{\B}\Delta,A\oplus B$
& iff
& For all $\mathcal{C}\supseteq\B$, all $\Theta_{\At}$ and all
$\Sigma_{\At}$, if both
$A\Vdash^{\Theta_{\At}}_{\C}\Sigma_{\At}$ and
$B\Vdash^{\Theta_{\At}}_{\C}\Sigma_{\At}$, then
$\Vdash^{\Gamma_{\At}\cup\Theta_{\At}}_{\C}\Delta,\Sigma_{\At}$. \\

($\to$)
& $\Vdash^{\Gamma_{\At}}_{\B}\Delta,A\to B$
& iff
& For all $\mathcal{C}\supseteq\B$, all $\Theta_{\At}$ and all
$\Sigma_{\At}$, if both
$\Vdash^{\Theta_{\At}}_{\C}A,\Sigma_{\At}$ and
$B\Vdash^{\Theta_{\At}}_{\C}\Sigma_{\At}$, then
$\Vdash^{\Gamma_{\At}\cup\Theta_{\At}}_{\C}\Delta,\Sigma_{\At}$. \\

($0$)
& $\Vdash^{\Gamma_{\At}}_{\B}\Delta, 0$
& iff
& for all $\Theta_{\At}$ and $\Sigma_{\At}$, $\Vdash^{\Gamma_{\At} \cup \Theta_{\At}}_{\B}\Delta, \Sigma_{\At}$. \\

($\top$)
&$\supp{B}{\Gamma\Atomic}\Delta,\top$
& 
& always. \\

\midrule
\multicolumn{4}{l}{\textbf{Multiplicative Connectives}} \\
\midrule

($\otimes$)
& $\Vdash^{\Gamma_{\At}}_{\B}\Delta,A\otimes B$
& iff
& For all $\mathcal{C}\supseteq\B$, all $\Theta_{\At}$ and all
$\Sigma_{\At}$, if
$A,B\Vdash^{\Theta_{\At}}_{\C}\Sigma_{\At}$, then 
$\Vdash^{\Gamma_{\At}\cup\Theta_{\At}}_{\C}\Delta,\Sigma_{\At}$. \\

($\parr$)
& $\Vdash^{\Gamma_{\At}}_{\B}\Delta,A\parr B$
& iff
& For all $\mathcal{C}\supseteq\B$, all $\Theta_{\At}$, $\Pi_{\At}$,
$\Sigma_{\At}$ and $\Omega_{\At}$, if both
$A\Vdash^{\Theta_{\At}}_{\C}\Sigma_{\At}$ and
$B\Vdash^{\Pi_{\At}}_{\C}\Omega_{\At}$, then
$\Vdash^{\Gamma_{\At}\cup\Theta_{\At}\cup\Pi_{\At}}_{\C}
\Delta,\Sigma_{\At},\Omega_{\At}$. \\

($\multimap$)
& $\Vdash^{\Gamma_{\At}}_{\B}\Delta,A\multimap B$
& iff
& For all $\mathcal{C}\supseteq\B$, all $\Theta_{\At}$, $\Pi_{\At}$,
$\Sigma_{\At}$ and $\Omega_{\At}$, if both
$\Vdash^{\Theta_{\At}}_{\C}A,\Sigma_{\At}$ and
$B\Vdash^{\Pi_{\At}}_{\C}\Omega_{\At}$, then
$\Vdash^{\Gamma_{\At}\cup\Theta_{\At}\cup\Pi_{\At}}_{\C}
\Delta,\Sigma_{\At},\Omega_{\At}$. \\

($\bot$)
& $\Vdash^{\Gamma_{\At}}_{\B}\Delta, \bot$
& iff
& $\Vdash^{\Gamma_{\At}}_{\B}\Delta$. \\

($1$)
&
$\supp{B}{\Gamma\Atomic}\Delta, 1$
& iff
& For all $\base{C}\supseteq\base{B}$, all $\Theta\Atomic$, and all $\Sigma\Atomic$, if $\supp{C}{\Theta\Atomic}\Sigma\Atomic$ then $\supp{C}{\Gamma\Atomic\cup\Theta\Atomic}\Delta,\Sigma\Atomic$. \\

\bottomrule
\end{tabular}
\caption{Support in a base $\B$ using resources $\Gamma_{\At}$.}\label{fig:sup}
\end{figure}

In what follows, we establish soundness and completeness for the classical, intuitionistic, and substructural logics considered here with respect to the semantics determined by the corresponding classes of bases. 
Although these results may seem natural, they do not follow automatically from the design of the semantic clauses. Their importance lies in the {\em uniformity} and {\em modularity} of the construction: to the best of our knowledge, this is the first treatment to provide proof-theoretic semantics for these logics within a common framework, by systematically varying conditions on bases and the choice of semantic clauses.

In Section~\ref{sec:modular}, we engage the modularity of our framework and results. In particular, the lack of use of structural rules in our proofs, which are only dealt with in separate theorems, means that eliding structural rules from bases immediately yields semantics for $\MALL$ and $\iMALL$. Moreover, a new degree of modularity is achieved if we combine our results with those of \cite{DBLP:journals/corr/abs-2504-08349}, as then intuitionistic logic and $\iMALL$ can naturally be visualised as both extensions and restrictions of classical logic and $\MALL$.

In Section~\ref{sec:jankov}, we turn to the more delicate case of intermediate logics, for which extending this approach requires a finer analysis of how conditions on bases can capture the additional principles that distinguish each logic from intuitionistic logic. We provide two new strategies for proving completeness which work for any logic obtainable by adding axioms to intuitionistic logic. One such strategy combines use of simulation bases with composition of sequent deductions, and the other combines them with the traditional construct of intuitionistic prime theories. The first is completely constructive, whereas the second is non-constructive but even more general.

\section{Soundness and completeness}\label{sec:soundcomp}
This section is devoted to showing that our definition of support yields a sound and complete semantics for {\em classical logic}, before extending the framework to substructural and intuitionistic logics. We consider both the additive system $\ACLp$, with the logical rules in Figure~\ref{fig:ACLp}, and the multiplicative system $\MCLp$, with the logical rules in Figure~\ref{fig:MCLp}. Both systems also include the structural rules in Figure~\ref{fig:struc}. 

Although these additive and multiplicative sequent system presentations are equivalent in the presence of the structural rules, their corresponding semantic clauses differ substantially. Treating them separately also prepares the ground for a semantics of $\MALL$.

\begin{figure*}[t!]
    \centering
    \setlength{\tabcolsep}{3pt}
    \renewcommand{\arraystretch}{0.8}

    \begin{subfigure}[t]{1\textwidth}
        \centering
        \vspace{-0.8em}
\begin{prooftree}
    \AxiomC{$A,\Gamma \Rightarrow \Delta$}
    \RightLabel{\tiny\scriptsize{$L \&_1$}}
    \UnaryInfC{$A \& B, \Gamma \Rightarrow \Delta$}
 \quad
    \DisplayProof
    \quad
    \AxiomC{$B,\Gamma \Rightarrow \Delta$}
    \RightLabel{\tiny\scriptsize{$L \&_2$}}
    \UnaryInfC{$A \& B, \Gamma \Rightarrow \Delta$}
     \quad
    \DisplayProof
    \quad
    \AxiomC{$\Gamma \Rightarrow  \Delta, A$}
    \AxiomC{$\Gamma \Rightarrow \Delta, B$}
    \RightLabel{\tiny\scriptsize{$R \&$}}
    \BinaryInfC{$\Gamma \Rightarrow \Delta, A \& B$}
\end{prooftree}
\begin{prooftree}
    \AxiomC{$A, \Gamma \Rightarrow  \Delta$}
    \AxiomC{$B, \Gamma \Rightarrow \Delta$}
    \RightLabel{\tiny\scriptsize{$L \oplus$}}
    \BinaryInfC{$A \oplus B, \Gamma \Rightarrow \Delta$}
 \;
    \DisplayProof
    \;
    \AxiomC{$\Gamma \Rightarrow \Delta, A$}
    \RightLabel{\tiny\scriptsize{$R \oplus_{1}$}}
    \UnaryInfC{$\Gamma \Rightarrow \Delta, A \oplus B$}
     \;
    \DisplayProof
    \;
    \AxiomC{$\Gamma \Rightarrow \Delta, B$}
    \RightLabel{\tiny\scriptsize{$R \oplus_{2}$}}
    \UnaryInfC{$\Gamma \Rightarrow \Delta, A \oplus B$}
\end{prooftree}
\begin{prooftree}
    \AxiomC{$\Gamma \Rightarrow  \Delta, A$}
    \AxiomC{$B, \Gamma \Rightarrow \Delta$}
    \RightLabel{\tiny\scriptsize{$L \to$}}
    \BinaryInfC{$A \to B, \Gamma \Rightarrow \Delta$}
    \DisplayProof
    \AxiomC{$A, \Gamma \Rightarrow \Delta$}
    \RightLabel{\tiny\scriptsize{$R \to_{1}$}}
    \UnaryInfC{$\Gamma \Rightarrow \Delta, A \to B$}
   \,
    \DisplayProof
    \,
    \AxiomC{$\Gamma \Rightarrow \Delta, B$}
    \RightLabel{\tiny\scriptsize{$R \to_{2}$}}
    \UnaryInfC{$\Gamma \Rightarrow \Delta, A \to B$}
\end{prooftree}
\begin{prooftree}
        \AxiomC{}
  \RightLabel{\tiny\scriptsize{$L 0$}}
        \UnaryInfC{$\Gamma, 0 \Rightarrow \Delta$}
    \;
    \DisplayProof
    \;
        \AxiomC{}
        \RightLabel{\tiny\scriptsize{$R\top$}}
        \UnaryInfC{$\Gamma\Rightarrow\Delta,\top$}
    \end{prooftree}
\vspace{-0.3cm}
\caption{Additive inference rules of the sequent system $\ACLp$.}\label{fig:ACLp}
    \end{subfigure}
     \hfill
    \vspace{0.1cm}
    \begin{subfigure}[t]{1.1\textwidth}
         \centering
 \begin{prooftree}
    \AxiomC{$A, B, \Gamma \Rightarrow \Delta$}
     \RightLabel{\tiny\scriptsize{$L \otimes$}}
     \UnaryInfC{$A \otimes B, \Gamma \Rightarrow \Delta$}
     \DisplayProof
     \AxiomC{$\Gamma \Rightarrow  \Delta, A$}
     \AxiomC{$\Gamma' \Rightarrow \Delta', B$}
     \RightLabel{\tiny\scriptsize{$R \otimes$}}
     \BinaryInfC{$\Gamma, \Gamma' \Rightarrow \Delta, \Delta', A \otimes B$}
     \DisplayProof
     \AxiomC{$A, \Gamma \Rightarrow  \Delta$}
     \AxiomC{$B, \Gamma' \Rightarrow \Delta'$}
     \RightLabel{\tiny\scriptsize{$L \parr$}}
     \BinaryInfC{$A \parr B, \Gamma, \Gamma' \Rightarrow \Delta, \Delta'$}
\end{prooftree}
\begin{prooftree}
     \AxiomC{$\Gamma \Rightarrow \Delta, A, B$}
     \RightLabel{\tiny\scriptsize{$R \parr$}}
     \UnaryInfC{$\Gamma \Rightarrow \Delta, A \parr B$}
     \,
     \DisplayProof
     \,
     \AxiomC{$\Gamma \Rightarrow  \Delta, A$}
     \AxiomC{$B, \Gamma' \Rightarrow \Delta'$}
     \RightLabel{\tiny\scriptsize{$L \multimap$}}
    \BinaryInfC{$A \multimap B, \Gamma, \Gamma' \Rightarrow \Delta, \Delta'$}
  \,
     \DisplayProof
     \,
     \AxiomC{$A, \Gamma \Rightarrow \Delta, B$}
     \RightLabel{\tiny\scriptsize{$R \multimap$}}
     \UnaryInfC{$\Gamma \Rightarrow \Delta, A \multimap B$}
 \end{prooftree}
        \begin{prooftree}
   \AxiomC{}
     \RightLabel{\tiny\scriptsize{$L \bot$}}
     \UnaryInfC{$ \bot \Rightarrow \varnothing$}
 \;
    \DisplayProof
     \;
     \AxiomC{$\Gamma \Rightarrow \Delta$}
     \RightLabel{\tiny\scriptsize{$R \bot$}}
     \UnaryInfC{$\Gamma \Rightarrow \Delta, \bot$}
     \;
     \DisplayProof
     \;
    \AxiomC{$\Gamma \Rightarrow \Delta$}
     \RightLabel{\tiny\scriptsize{$L 1$}}
     \UnaryInfC{$1,\Gamma\Rightarrow \Delta$}
    \;
    \DisplayProof
     \;
     \AxiomC{}
     \RightLabel{\tiny\scriptsize{$R 1$}}
     \UnaryInfC{$ \varnothing \Rightarrow 1$}
     \end{prooftree}
 \vspace{-0.3cm}
 \caption{Multiplicative inference rules of the sequent system $\MCLp$.}\label{fig:MCLp}
     \end{subfigure}
    \hfill
    \vspace{0.1cm}
   \begin{subfigure}[t]{1\textwidth}
        \centering
  \begin{prooftree}
    \AxiomC{$\Gamma \Rightarrow \Delta$}
    \RightLabel{$LW$}
    \UnaryInfC{$A, \Gamma \Rightarrow \Delta$}
 \;
    \DisplayProof
    \;
    \AxiomC{$\Gamma \Rightarrow \Delta$}
    \RightLabel{$RW$}
    \UnaryInfC{$ \Gamma \Rightarrow \Delta, A$}
 \;
    \DisplayProof
    \;
    \AxiomC{$A, A, \Gamma \Rightarrow \Delta$}
    \RightLabel{$LC$}
    \UnaryInfC{$A, \Gamma \Rightarrow \Delta$}
 \;
    \DisplayProof
    \;
    \AxiomC{$\Gamma \Rightarrow \Delta, A, A$}
    \RightLabel{$RC$}
    \UnaryInfC{$ \Gamma \Rightarrow \Delta, A$}
\end{prooftree}
\begin{prooftree}
    \AxiomC{}
    \RightLabel{$\init$}
    \UnaryInfC{$ A \Rightarrow A$}
\end{prooftree}
\vspace{-0.3cm}
\caption{Structural rules}\label{fig:struc}
    \end{subfigure}
\caption{Additive and multiplicative systems for propositional classical logic.}\label{fig:CLp}
\end{figure*}

We start by exploring some structural properties of bases and state some useful lemmas -- you may find the proofs in Appendix~\ref{sec:soundcomp}. In particular, we study the relationship between bases closed under structural rules and satisfaction of structural principles at the semantic level.

The next two lemmata show how to extend (Inf) to general formulas and to arbitrary contexts, respectively.
\begin{lemma}[Non-atomic inference]\label{lemma:non-atomicinference}
    $\{ A^{1}, \ldots, A^{n}\} \Vdash^{\Gamma_{\At}}_\base{B} \Delta$ iff, for all $\C \supseteq \B$, all atomic $\Sigma^{i}_{\At}$ and all arbitrary $\Pi^{i}$ (with $1 \leq i \leq n$), if for all $i$ $\Vdash^{\Sigma^{i}_{\At}}_{\C} A^{i}, \Pi^{i}$, then $\Vdash^{\Gamma_{\At} \cup \Sigma^{1}_{\At} \cup \ldots \cup \Sigma^{n}_{\At}}_{\C} \Delta, \Pi^{1}, \ldots, \Pi^{n}$
\end{lemma}

\begin{lemma} \label{lemma:partialInf}
  Let $\Gamma = \{A^{1}, \ldots, A^{n}\}$.  If $\Gamma, \Delta \Vdash^{\Theta_{\At}}_{\B} \Sigma$ and $\Vdash^{\Omega^{m}_{\At}}_{\B} \Pi^{m}_{\At}, A^{m}$ for all $1 \leq m \leq n$, then $\Delta \Vdash^{\Theta_{\At}, \Omega^{1}_{\At}, \ldots, \Omega^{n}_{\At}}_{\B} \Sigma, \Pi^{1}_{\At}, \ldots, \Pi^{n}_{\At}$. 
\end{lemma}

Monotonicity, provable through standard methods \cite{DBLP:journals/sLogica/GheorghiuGP26,buzoku2026ILL,DBLP:conf/tableaux/BarrosoNascimentoPP25}, also comes in handy. The proof is by  straightforward induction and it is omitted. %

\begin{lemma}
    If $\Vdash^{\Gamma_{\At}}_{\B} \Delta$ and $\C \supseteq \B$ then $\Vdash^{\Gamma_{\At}}_{\C} \Delta$.
\end{lemma}

The next result is a {\em substitution lemma}, also referred to as a {\em cut} result. It will play a  key role in the proof of soundness presented in the next section.
\begin{lemma}\label{lemma:basecut}
    If $\Gamma \Vdash^{\Theta_{\At}}_{\B} A, \Delta$ and $A, \Sigma \Vdash^{\Omega_{\At}}_{\B} \Pi$ then $\Gamma, \Sigma \Vdash^{\Theta_{\At}, \Omega_{\At}}_{\B} \Delta, \Pi$.
\end{lemma}
We note, however, that this result can be extended to the {\em resources} if and only if we assume that bases are closed under \init\ and \Cutc, as the next result shows.

\begin{proposition}\label{prop:subs}
For all $\Gamma_{\At}, \Delta_{\At}$ and $\mathcal{B}$ 
\[\Gamma_{\At} \Vdash^{\Theta_{\At}}_{\B} \Delta_{\At} \mbox{ iff }\vdash_{\B}, \Theta_{\At} \Gamma_{\At} \Rightarrow \Delta_{\At}\] 
holds if and only if $\B$ is closed under the atomic rules \init\ and \Cutc. 
\end{proposition}
We now show that structural principles of semantic satisfaction are closely tied to the corresponding atomic rules in the base. This correspondence is not automatic: in sequent calculi, contraction cannot, in general, be restricted to atomic formulas without loss of provability~\cite{DBLP:journals/ndjfl/Brunnler06}. Our results show that, in $\Bes$, atomic contraction rules in the base suffice to ensure contraction at the semantic level for arbitrary formulas. The same kind of results hold for weakening.
 
\begin{theorem}[Contraction]\label{thm:contractionleftright}
  For all $\B$,  $((\Gamma, A, A \Vdash^{\Theta_{\At}}_{\C} \Delta$ implies $\Gamma, A \Vdash^{\Theta_{\At}}_{\C} \Delta)$ and $(\Gamma \Vdash^{\Theta_{\At}}_{\C} \Delta, A, A$ implies $\Gamma \Vdash^{\Theta_{\At}}_{\C} \Delta, A))$ holds for all $\C \supseteq \B$, $\Gamma, \Delta, \Theta_{\At}, A$,  if, and only if, $((\Vdash^{p, p, \Gamma_{\At}}_{\C} \Delta_{\At}$ implies  $\Vdash^{p, \Gamma_{\At}}_{\C} \Delta_{\At})$ and $(\Vdash^{ \Gamma_{\At}}_{\C} \Delta_{\At}, p, p$ implies  $\Vdash^{ \Gamma_{\At}}_{\C} \Delta_{\At}, p)$) holds for all $\C \supseteq \B$, $\Gamma_{\At}, \Delta_{\At}$, $p$.
\end{theorem}

\begin{corollary}\label{cor:leftrightcontraction}
    If $\B$ is closed under left and right contraction, then $((\Gamma, A, A \Vdash^{\Theta_{\At}}_{\B} \Delta$ implies $\Gamma, A \Vdash^{\Theta_{\At}}_{\B} \Delta)$ and $(\Gamma \Vdash^{\Theta_{\At}}_{\B} \Delta, A, A$ implies $\Gamma \Vdash^{\Theta_{\At}}_{\B} \Delta, A))$ holds for all $\Gamma, \Delta, \Theta_{\At}, A$.
\end{corollary}

\begin{theorem}[Weakening]\label{thm:weakeningleftright}
  For all $\B$, if  $((\Gamma \Vdash^{\Theta_{\At}}_{\C} \Delta$ implies $\Gamma, A \Vdash^{\Theta_{\At}}_{\C} \Delta)$ and $(\Gamma \Vdash^{\Theta_{\At}}_{\C} \Delta$ implies $\Gamma \Vdash^{\Theta_{\At}}_{\C} \Delta, A))$ holds for all $\C \supseteq \B$, $\Gamma, \Delta, \Theta_{\At}, A$,  if, and only if, $((\Vdash^{ \Gamma_{\At}}_{\C} \Delta_{\At}$ implies  $\Vdash^{p, \Gamma_{\At}}_{\C} \Delta_{\At})$ and $(\Vdash^{ \Gamma_{\At}}_{\C} \Delta_{\At}$ implies  $\Vdash^{ \Gamma_{\At}}_{\C} \Delta_{\At}, p)$) holds for all $\C \supseteq \B$, $\Gamma_{\At}, \Delta_{\At}$, $p$.
\end{theorem}

\begin{corollary}
    If $\B$ is closed under left and right weakening, then $((\Gamma \Vdash^{\Theta_{\At}}_{\B} \Delta$ implies $\Gamma, A \Vdash^{\Theta_{\At}}_{\B} \Delta)$ and $(\Gamma \Vdash^{\Theta_{\At}}_{\B} \Delta$ implies $\Gamma \Vdash^{\Theta_{\At}}_{\B} \Delta, A))$ holds for all $\Gamma, \Delta, \Theta_{\At}, A$.
\end{corollary}

We move on with the proof of soundness. Interestingly enough, the argument for each logical connective requires neither contraction nor weakening. This makes explicit the separation between the semantic clauses for the connectives and the structural principles of the logic. The proof is therefore modular: the arguments for the logical rules carry over unchanged to the corresponding substructural systems.

\begin{theorem}[Soundness]\label{lemma:classicalsoundnessadd}
    If all bases are classical, the semantics is sound with respect to $\ACLp$ and $\MCLp$. In other words, if there is a proof of $\Gamma\Rightarrow\Delta$ in $\ACLp$ or $\MCLp$, then $\Gamma\supp{}{}\Delta$.
\end{theorem}
\begin{proof}
The proof follows the techniques developed in~\cite{DBLP:conf/tableaux/BarrosoNascimentoPP25} and proceeds by induction on the derivation of $\Gamma\Rightarrow\Delta$, analysing the last rule applied. The full proof is given in Appendices~\ref{app:add} for the additive case and~\ref{app:mult} for the multiplicative one. We observe that  Lemma~\ref{lemma:basecut} is key, showing that, in fact, ``cut'' plays a central role in establishing soundness of the logical rules.
\end{proof}

To prove completeness, we follow the script given in~\cite{Sandqvist2015IL}, 
by encoding the relevant sequent calculus within an atomic base. More concretely, given a target sequent $\Gamma\Rightarrow\Delta$, we assign a distinct atom $p^A$ to each of its subformulas $A$, leaving atoms unchanged. We then construct a simulation base $\mathcal U$ containing atomic counterparts of the logical and structural rules. For example, the right introduction rule for additive conjunction is represented by
\[\infer[R \&]{\Gamma_{\At} \Rightarrow \Delta_{\At}, p^{A \& B}}{\Gamma_{\At} \Rightarrow  \Delta_{\At}, p^{A}& \Gamma_{\At} \Rightarrow \Delta_{\At}, p^{B}}\]
Although $p^{A\&B}$ is syntactically an atom, this rule reproduces the inferential relationship between a conjunction and its conjuncts. The corresponding left rules are encoded similarly. Thus, the internal structure of formulas is represented through the rules governing their atomic names.

The key step is to show that these names faithfully simulate the semantic behaviour of the formulas they represent. An induction on formulas establishes that, for every extension $\mathcal B\supseteq\mathcal U$,
\[
\Vdash_{\mathcal B}^{\Omega_{\mathsf{At}}} A,\Theta_{\mathsf{At}}
\quad\text{iff}\quad
\vdash_{\mathcal B}
\Omega_{\mathsf{At}}\Rightarrow p^A,\Theta_{\mathsf{At}},
\]
and, more generally, the analogous equivalence holds for contexts of formulas. Now, if $\Gamma\Vdash\Delta$, this consequence holds in particular over $\mathcal U$. The atomic identities $p^A\Rightarrow p^A$, together with the simulation equivalence, supply support for the assumptions in $\Gamma$. Semantic inference then yields an atomic derivation of $\Gamma_{\mathsf{At}}\Rightarrow\Delta_{\mathsf{At}}$, where these contexts contain the names of the formulas in $\Gamma$ and $\Delta$. Finally, replacing each atomic name by its associated formula transforms this derivation into a sequent calculus derivation of $\Gamma\Rightarrow\Delta$, with any resulting cuts removed by cut admissibility.

A key feature of the present formulation, and completely different from~\cite{Sandqvist2015IL,DBLP:journals/corr/abs-2504-08349}, is that support is indexed by finite multisets of atomic resources, with atomic support defined directly through derivability
\[
\Vdash^{\Gamma_{\At}}_{\B} \Delta_{\At}
\quad\text{iff}\quad
\vdash_{\B} \Gamma_{\At} \Rightarrow \Delta_{\At}.
\]
This direct connection makes resource usage explicit and allows structural behaviour to be controlled through conditions on bases and their admissible extensions, providing the basis for a modular treatment of substructural logics.

\begin{definition}[Atomic mapping] Let $\Sigma$ be a set of formulas. Let $\Sigma_{S}$ be the set of all subformulas of formulas in $\Sigma$. We say that a function $\alpha: \Sigma_{S} \rightarrow \At$ is an \emph{atomic mapping} for $\Sigma$ if (1) $\alpha$ is injective, (2) $\alpha(p) = p$ for $p \in \At$. If $A\notin\At$, we denote $\alpha(A) \eqcolon p^A$.
\end{definition}

\begin{definition}[Simulation base]
Let $\Sigma$ be a set of formulas, and let $\alpha$ be an atomic mapping on $\Sigma$. An {\em additive (resp. multiplicative) simulation base} $\U$, for $\Sigma$ and $\alpha$, is a base that is closed under \init, \Cutc, weakening, and contraction, instantiated for all $p \in \At$, in addition to the simulated sequent calculus rules in Figure~\ref{fig:comp} (resp. Figure~\ref{fig:compmult}), instantiated for all $A, B \in \Sigma$.
\end{definition}

We are now ready to present the completeness proof. The proof of the following result is detailed in Appendix~\ref{app:add}.

\begin{lemma}\label{mainlemmacompletenessadd}
   Let $\Theta_{\At}$ and $\Omega_{\At}$ be any (possibly empty) set of atoms, $\Sigma$ any (non-empty) set containing only subformulas of $\Gamma \cup \Delta$ and $\Sigma_{\At} = \{p^{B} | B \in \Sigma\}$.  Then, for all $\B \supseteq \U$, $\Vdash^{\Omega_{\At}}_{\B} \Sigma, \Theta_{\At}$ if and only if $\vdash_{\B} \ \Omega_{\At}\Rightarrow \Sigma_{\At}, \Theta_{\At}$.
\end{lemma}

\begin{theorem}[Completeness]\label{thm:additivecompleteness}
If $\Gamma \Vdash \Delta$ then the sequent $\Gamma \Rightarrow \Delta$ is provable in $\ACLp$ and $\MCLp$.
\end{theorem}

\begin{proof}
    Fix a mapping $\alpha$ and a simulation base $\U$. Assume $\Gamma \Vdash \Delta$. Then $\Gamma \Vdash^{\varnothing}_{\U} \Delta$. Let $\Gamma = \{A^{1}, \ldots, A^{n}\}$. For every $1 \leq m \leq n$ we have $\vdash_{\U} p^{A^{m}} \Rightarrow p^{A^{m}}$ by \init. Lemma \ref{mainlemmacompletenessadd} yields $\Vdash^{p^{A^{m}}}_{\U} {A^{m}}$ for all such $m$, whence since $\{ A^{1}, \ldots, A^{n}\} \Vdash_{\U} \Delta$ we conclude $\Vdash^{\Gamma_{\At}}_{\U} \Delta$. Lemma \ref{mainlemmacompletenessadd} then yields $\vdash_{\U} \Gamma_{\At} \Rightarrow \Delta_{\At}$. Given the structure of the simulation base, by picking the deduction with conclusion $\Gamma_{\At} \Rightarrow \Delta_{\At}$ on $\U$ and replacing every atoms with shape $p^{C}$ by $C$ we obtain a deduction showing $\Gamma \Rightarrow \Delta$ in our sequent calculus.
    
\end{proof}

\section{Modularity}\label{sec:modular}
A central feature of the semantics is its modular treatment of substructurality. The logical parts of the soundness and completeness arguments do not rely on weakening or contraction. Instead, Theorems~\ref{thm:contractionleftright} and~\ref{thm:weakeningleftright} establish how these principles at the level of semantic entailment are governed by their atomic counterparts. We can therefore vary the structural requirements on bases and their admissible extensions while retaining the semantic clauses for the connectives. For the linear instances, we adopt the following terminology.

We do not need to change anything in the proofs. Importantly, modularity only holds when we are proving soundness and completeness for sequent systems whose rules match their semantic clauses. Soundness and completeness is still provable for classical logic if we use a sequent calculus with rules for $(\oplus)$ and a semantics with a clause for $(\&)$, but the proof requires structural reasoning. 

Modularity also applies to other types of structural restrictions. In particular, if we combine intuitionistic linear support with intuitionistic bases we obtain a semantics for intuitionistic logic, and if we combine it with intuitionistically linear bases we obtain a semantics for intuitionistic linear logic. This is an immediate consequence of the proofs of soundness, completeness, and the following result, proven in Appendix~\ref{app:single}.

\begin{lemma}\label{lemma:intmod}

If all bases only contain rules with at most one succedent, then for all $\B$, $\Sigma$, $\Gamma_{\At}$ and $\Delta$, if  $\Sigma \Vdash^{\Gamma_{\At}}_{\B} \Delta$ then there is at most one formula in $\Delta$, provided the semantics does not contain $(\bot)$, $(\top)$, $(0)$ and $(1)$.
\end{lemma}

The clauses for units therefore have to be manually restricted, although for the restrictions on bases spread to all others.

\begin{definition}
Linear support is defined using all clauses in Figure~\ref{fig:sup} except $(\to)$. Intuitionistic support is defined using the same clauses, excluding $(\to)$, $(\parr)$ and requiring the $\Delta$ in $(\bot)$, $(\top)$, $(0)$ and $(1)$ to be empty. Intuitionistic linear support uses the clauses of intuitionistic support minus $(\bot)$. 
\end{definition}
\begin{corollary}
If bases are 
\begin{itemize}
\item classically linear, then linear support is sound and complete with respect to $\MALL$;
\item intuitionistically linear, then intuitionistic linear support is sound and complete with respect to $\iMALL$;
\item intuitionistic, the intuitionistic linear support is sound and complete with respect to $\IPL$.
\end{itemize}
\end{corollary}

The results we present for classical and intuitionistic logic can also be combined in an interesting way with those in \cite{DBLP:journals/corr/abs-2504-08349}. Here we show that, by defining semantic clauses as in Figure \ref{fig:sup} and quantifying over arbitrary sets of atoms $\Sigma_{\At}$, we obtain either classical or linear logic, depending on whether bases contain structural rules or not. If we restrict bases and all semantic clauses by requiring $\Sigma_{\At}$ to contain at most one formula --- therefore quantifying only over singletons and the empty set --- we obtain either intuitionistic linear logic or intuitionistic logic. A direct consequence of the results in \cite{DBLP:journals/corr/abs-2504-08349} and our clause for $\bot$ is that, if we further restrict $\Sigma_{\At}$ and require it to be empty, the semantics becomes classical or linear again, provided the atomic support clause is in elimination form when defined for non-empty sets.

By elimination form atomic support clause we mean the following:

\medskip

\noindent $(\mathsf{EAt})$ $\Vdash^{\Gamma_{\At}}_{\B} p$ iff, for all $\C \supseteq \B$ and all $\Sigma_{\At}, \Theta_{\At}$, if $\vdash_{\C} p, \Theta_{\At} \Rightarrow \Sigma_{\At}$ then $\vdash_{\C} \Gamma_{\At}, \Theta_{\At} \Rightarrow \Sigma_{\At}$.

\medskip

This clause has been studied in detail in \cite[pg. 7]{DBLP:journals/corr/abs-2504-08349}, which uses (sequent-style) natural deduction. Their Proposition 3.8 shows that, in the presence of \Cuti, $(\At)$ and $(\mathsf{EAt})$ are equivalent, but in its absence they are not. Unlike there, our clause allows $\Sigma_{\At}$ to be empty, but this does not impact the proof.

The clause encodes the cut principle in the definition of atomic support, and it cannot be defined for $\{\varnothing\}$ because the empty set contains no formula we could cut. Therefore, $\Vdash^{\Gamma_{\At}}_{\B} \varnothing$ would still have to be defined with $(\At)$. This also explains why no obvious elimination-based clause can be provided for $\bot$. Support for the empty set cannot be defined in elimination format, and it is trivial to prove $\Vdash^{\Gamma_{\At}}_{\B} \varnothing$ iff $\Vdash^{\Gamma_{\At}}_{\B} \bot$ iff $\vdash_{\B} \Gamma_{\At} \Rightarrow \varnothing$. On the other hand, support for $\{\varnothing \}$ cannot be conflated with support for $\bot$ because the former, being a set of atoms, is quantified over in semantic clauses, but not the latter.

Now suppose that support is defined through $(\mathsf{EAt})$ for non-empty sets and $(\At)$ for empty sets. The equivalency of those clauses would make it so that this semantics is equivalent to the ones we present for classical, linear, intuitionistic, and intuitionistic linear logic. Therefore, in those semantics, by restricting quantification over arbitrary sets $\Gamma_{\At}$ in the semantic clauses to sets with at most one formula, the semantics becomes one for either intuitionistic or intuitionistic linear logic. If we go further and restrict quantification to the empty set (therefore it is no longer a quantification at all, just like in \cite{DBLP:journals/corr/abs-2504-08349}), due to the shape of the results in \cite{DBLP:journals/corr/abs-2504-08349}, the semantics becomes either classical or linear again, as support of $\varnothing$ is equivalent to support of $\bot$. Therefore, by restricting quantification on the semantics of classical linear logic we obtain intuitionistic linear logic, and by restricting it again we obtain a different semantics for classical linear logic. Likewise, a restriction of the classical semantics yields intuitionistic semantics, and a restriction of the intuitionistic semantics yield a different classical semantics.

This provides interesting insights on the topic of whether intuitionistic logic is a restriction or an extension of classical logic  \cite{LegrisMolina2001}. The unorthodox view that intuitionistic logic contains classical logic as a proper subpart was advocated by \L ukasiewicz \cite{lukasiewicz1952intuitionistic}, but the standard view is that classical logic contains intuitionistic logic as a proper subpart \cite{sep-logic-intuitionistic}. Our semantics provides substance to the claim that both viewpoints are correct and provides insights on why this is so: the empty set behaves like $\bot$, singleton sets behave like intuitionistic formulas and formulas on bigger sets behave like classical formulas. Hence the intuitionistic semantics stands at the midpoint between two distinct classical semantics.

\section{Intermediate logics and modular completeness proofs}\label{sec:jankov}

Modularity becomes more nuanced when we consider intermediate logics. This is not surprising, as most intermediate logics do not have cut-free sequent calculi and cannot be obtained by including only atomic instances of axioms~\cite{DBLP:journals/aml/DyckhoffN12,Ciabattoni2013}. Even if inclusion of atomic rules ends up producing the desired semantic behaviour, there is no guarantee that soundness of the added axioms can be proved inductively. We will show that, surprisingly, completeness, on the other hand, can be proved modularly in at least two ways, provided we find a base with suitable semantic behaviour.
All proofs of this section are contained in Appendix~\ref{app:jankov}.

As a case study for soundness and (modular) completeness, we analyse Jankov logic, which extends intuitionistic logic with weak excluded middle, $\neg A \oplus \neg\neg A$. Its atomic instances suffice to derive the full axiom schema~\cite{Corsi2008,Hosoi1986}, so we use the following presentation.

\begin{definition}
   Let $\Int$ be the sequent system with succedents restricted to at most one formula, containing the structural rules LW and CW, \Cuti, and the rules for $\multimap$, $\bot$, $\oplus$, and $\otimes$ of Figure~\ref{fig:CLp}, in addition to a version of \init\ allowing arbitrary contexts on the left. Then the system $\JK$ is obtained by adding to $\Int$ all instances of the following rule 

\begin{prooftree}
    \AxiomC{}
    \RightLabel{$\sf{AxJK}$}
    \UnaryInfC{$\Rightarrow \neg p \oplus \neg \neg p$}
\end{prooftree}
\end{definition}

For the completeness argument using prime theories, we also extend syntactic consequence to possibly infinite sets of premises.

\begin{definition}
  For possibly infinite $\Gamma$,  $\Gamma \vdash_{\Int} \Delta$ holds iff there is a finite $\Theta \subseteq \Gamma$ such that $\Theta \Rightarrow \Delta$ is derivable in $\Int$. Likewise, $\Gamma \vdash_{\JK} \Delta$ holds iff there is a finite $\Theta \subseteq \Gamma$ such that $\Theta \Rightarrow \Delta$ is derivable in $\JK$.
\end{definition}

The key to modular completeness is that derivations in an axiomatic extension can be represented by intuitionistic derivations with finitely many additional assumptions. For Jankov logic, these assumptions are instances of weak excluded middle.

\begin{definition}
    A \emph{Jankov set} is a set of formulas of shape $\neg A \oplus \neg \neg A$. An \emph{atomic Jankov set} is a set of formulas of shape $\neg p  \oplus \neg \neg p$, for $p \in \At$.
\end{definition}

\begin{lemma}\label{lemma:derivJankov}
    The sequent $\Gamma \Rightarrow \Delta$ is derivable in $\JK$ iff there is a finite atomic Jankov set $\Theta$ such that $\Theta, \Gamma \Rightarrow \Delta$ is derivable in $\Int$.
\end{lemma}

The reduction of derivability in an axiomatic extension to intuitionistic derivability with finitely many additional axiom assumptions is general. For Jankov logic, atomic instances suffice~\cite[Lemma 4]{Hosoi1986}, yielding the following equivalent formulation and its corollary.

\begin{lemma}\label{lemma:finite}
    $\Gamma \Rightarrow \Delta$ is derivable in $\JK$ iff $\Theta, \Gamma \Rightarrow \Delta$ is derivable in $\Int$ for a finite Jankov set $\Theta$.
\end{lemma}

\begin{corollary}\label{cor:Jankovcor}
    There is a finite Jankov set $\Theta$ such that $\Theta, \Gamma \Rightarrow \Delta$ is derivable in $\Int$ iff there is a finite Jankov atomic set $\Sigma$ such that $\Sigma, \Gamma \Rightarrow \Delta$ is derivable in $\Int$.
\end{corollary}

These results allow us to establish soundness without a separate induction on $\JK$ derivations. Although atomic instances suffice to axiomatise the logic, the completeness constructions below use Jankov formulas of arbitrary complexity.

We now define two classes of bases for Jankov logic, retaining the intuitionistic clauses for support. We first introduce notation for atomic rules and their closure under additional left contexts.

\begin{definition}
Let
\begin{itemize}
\item     $[\Gamma^{1}_{\At} \Rightarrow \Delta^{1}_{\At}, \ldots , \Gamma^{n}_{\At} \Rightarrow \Delta^{n}_{\At} / \Theta_{\At} \Rightarrow \Sigma_{\At}]$ be a rule with premises $\Gamma^{1}_{\At} \Rightarrow \Delta^{1}_{\At}, \ldots, \Gamma^{n}_{\At} \Rightarrow \Delta^{n}_{\At}$ and conclusion $\Theta_{\At} \Rightarrow \Sigma_{\At}$.
\item     $\{[\mathfrak{C}^{1}, \Gamma^{1}_{\At} \Rightarrow \Delta^{1}_{\At}, \ldots , \mathfrak{C}^{n}_{\At}, \Gamma^{n}_{\At} \Rightarrow \Delta^{n}_{\At} / \mathfrak{C}^{1}, \ldots, \mathfrak{C}^{n}, \Theta_{\At} \Rightarrow \Sigma_{\At}]\}$ be the set of all rules  $\{[\Omega^{1}_{\At}, \Gamma^{1}_{\At} \Rightarrow \Delta^{1}_{\At}, \ldots , \Omega^{n}_{\At}, \Gamma^{n}_{\At} \Rightarrow \Delta^{n}_{\At} / \Omega^{1}_{\At}, \ldots, \Omega^{n}_{\At}, \Theta_{\At} \Rightarrow \Sigma_{\At}]\}$, where $\Omega_{\At}$ is an arbitrary set of atoms.
\end{itemize}
Then
   $\B$ is said to be \emph{cumulative} if  $[\Gamma^{1}_{\At} \Rightarrow \Delta^{1}_{\At}, \ldots , \Gamma^{n}_{\At} \Rightarrow \Delta^{n}_{\At} / \Theta_{\At} \Rightarrow \Sigma_{\At}] \in \B$ implies  $\{[\mathfrak{C}^{1}, \Gamma^{1}_{\At} \Rightarrow \Delta^{1}_{\At}, \ldots , \mathfrak{C}^{n}_{\At}, \Gamma^{n}_{\At} \Rightarrow \Delta^{n}_{\At} / \mathfrak{C}^{1}, \ldots, \mathfrak{C}^{n}, \Theta_{\At} \Rightarrow \Sigma_{\At}]\} \subseteq \B$.
\end{definition}

Cumulativity allows additional atomic assumptions to be carried through rule applications. It makes atomic derivability behave as in~\cite[Definition 3]{DBLP:conf/tableaux/BarrosoNascimentoPP25} and supports a corresponding composition argument (cf.~\cite[Lemma 5]{DBLP:conf/tableaux/BarrosoNascimentoPP25}).

\begin{definition}
Let   $\B^{\neg} = (\B \cup \{ [/p \Rightarrow \varnothing] \}) $ and $\B^{\neg \neg} = (\B \cup \{ [\mathfrak{C}, p \Rightarrow \varnothing / \mathfrak
   {C}\Rightarrow \varnothing ]\})$ for every $p\in\At$.

    An intuitionistic base $\B$ is 
    \begin{itemize}
    \item \emph{weakly Jankov} if it is cumulative and, for all $\C \supseteq \B$, all $\Gamma_{\At}$, and all $\Delta_{\At}$ with at most one formula, if $\vdash_{\C^{\neg}} \Gamma _{\At}\Rightarrow \Delta_{\At}$ and $\vdash_{\C^{\neg \neg} } \Gamma_{\At} \Rightarrow\Delta_{\At}$ then $\vdash_{\C} \Gamma_{\At} \Rightarrow \Delta_{\At}$.
\item \emph{strongly Jankov} if either $\vdash_{\B} p \Rightarrow \varnothing$ or, for all $\C \supseteq \B$ and all $\Gamma_{\At}$, if $\vdash_{\C} \Gamma_{\At}, p \Rightarrow \varnothing$ then $\vdash_{\C} \Gamma_{\At} \Rightarrow \varnothing$.
\end{itemize}

$\Gamma \Vdash_{\WJK} \Delta$ and $\Gamma \Vdash_{\SJK} \Delta$ iff $\Gamma \Vdash^{\varnothing}_{\B} \Delta$ holds in all weakly and strongly Jankov bases, respectively.
\end{definition}

Both conditions ensure that every atomic instance of weak excluded middle is supported. The weak condition expresses closure under reasoning by cases through the two extensions $\mathcal B^{\neg}$ and $\mathcal B^{\neg\neg}$. The strong condition requires one of the two alternatives to hold at the base. The weak condition could instead be stated directly in terms of atomic rules, but the resulting formulation would be considerably less concise.

The following result shows that the weak excluded middle axiom is sound.

\begin{lemma}\label{lemma:jankovsound}
  If $\B$ is weakly or strongly Jankov then  $\Vdash^{\varnothing}_{\B} \neg p \oplus \neg \neg p$  for all $p \in \At$.
\end{lemma}

Together with intuitionistic soundness and the reduction of $\JK$ derivations to intuitionistic derivations with Jankov assumptions, this establishes soundness for both classes of bases. 

\begin{theorem}[Soundness]
$\Gamma \vdash_{\JK} A$ implies $\Gamma \Vdash_{\WJK} A$ and $\Gamma \Vdash_{\SJK} A$. 
\end{theorem}

We now turn to completeness, defining some new notions to adapt the traditional simulation base proof.

\begin{definition}
   For finite $\Gamma$, we define $\Gamma^{S}_{\At}$ as the set of all atomic subformulas of formulas in $\Gamma$.

Moreover, $\mathbb{F}(\Gamma^{S}_{\At})$ is the set of all formulas such that all their atomic subformulas are in $\Gamma^{S}_{\At}$. 
In particular, %
$\bot$ and formulas with it as subformula are vacuously included.
\end{definition}

\begin{definition}
Given any $\mathbb{F}(\Gamma^{S}_{\At})$ we define the {\em simulation mapping} $\alpha$ as a bijection $\alpha: \mathbb{F}(\Gamma^{S}_{\At}) \mapsto \At$ such that, for all $p \in \Gamma^{S}_{\At}$, $\alpha(p) = p$.
\end{definition}

Both completeness arguments (for weak and strong Jankov bases) use a simulation mapping for the formulas over the finite set of atoms occurring in the sequent under consideration. The mapping determines an intuitionistic simulation base $\mathcal U$, which is then extended to obtain a suitable Jankov base. The constructions differ at this point: the weak case gives a constructive proof, whereas the strong case proceeds by contraposition using a prime theory.

For weakly Jankov bases, let $\Theta$ contain all Jankov formulas in this language. We extend $\mathcal U$ with rules supporting the formulas in $\Gamma \cup \Theta$ and show that the resulting base is weakly Jankov. Semantic entailment then yields an intuitionistic derivation of $\Omega,\Gamma \Rightarrow \Delta$ for some finite $\Omega \subseteq \Theta$. Since $\Omega$ is a Jankov set, Corollary~\ref{cor:Jankovcor} and Lemma~\ref{lemma:derivJankov} give a $\JK$ derivation of $\Gamma \Rightarrow \Delta$.

\begin{lemma}\label{lemma:weak}
Let $\U$ be a simulation base built with $\alpha$ defined for $\mathbb{F}(\Gamma^{S}_{\At})$, where rules are allowed to have arbitrary extra left contexts. Let $\Theta$ be the set of all Jankov formulas in $\mathbb{F}(\Gamma^{S}_{\At})$ and $\E$ the extension of $\U$ obtained by adding $[/\Delta_{\At} \Rightarrow p^{A}]$ for all $A \in \Theta$ and all $\Delta_{\At}$. Then $\E$ is a weakly Jankov base.
\end{lemma}

\begin{theorem}[Weak completeness]\label{thm:weakcompleteness}
If $\Gamma \Vdash_{\WJK} \Delta$ then $\Gamma \vdash_{\JK} \Delta$.
\end{theorem}

For strongly Jankov bases, we prove completeness by contraposition. We write the succedent as a single formula $A$: an empty succedent can be represented by $\bot$. The construction uses the following standard prime-extension result~\cite[pg. 170]{Dalen}.

\begin{definition}
A set $\Gamma$ is a \emph{prime theory} if $\Gamma \vdash_{\Int} A$ implies $A \in \Gamma$ and $A \oplus B \in \Gamma$ implies that either $A \in \Gamma$ or $B \in \Gamma$.
\end{definition}

\begin{lemma}\label{lemma:prime}
    If $\Gamma \nvdash_{\Int} A$ then there is a prime theory $\Gamma'$ such that $\Gamma \subseteq \Gamma'$ and $\Gamma' \nvdash_{\Int} A$.
\end{lemma}

Suppose that $\Gamma \nvdash_{\JK} A$, and let $\Theta$ be the set of all Jankov formulas in $\mathbb{F}((\Gamma \cup \{ A \})^{S}_{\At})$. By Lemma~\ref{lemma:finite}, $\Gamma \cup \Theta \nvdash_{\Int} A$. Lemma~\ref{lemma:prime} therefore yields a prime theory $\Gamma'$ containing $\Gamma \cup \Theta$ such that $\Gamma' \nvdash_{\mathrm{Int}} A$. Let $\mathcal U$ be the simulation base for this language, and extend it to a base $\E$ by adding the zero-premise rule $\Rightarrow p^B$ for every $B \in \Gamma'$. The next two lemmas show that support in $\mathcal E$ coincides with membership in $\Gamma'$ and that $\mathcal E$ is strongly Jankov. Thus $\mathcal E$ supports every formula in $\Gamma$ but does not support $A$, providing the required counterexample to semantic entailment.

\begin{lemma}\label{mainlemmaJankovstrong}
  For all $C$ in $\mathbb{F}((\Gamma \cup \{ A\})^{S}_{\At})$, we have $\Vdash^{\varnothing}_{\E} C$ iff $C \in \Gamma'$. 
\end{lemma}

\begin{lemma}\label{lemma:isstronglyJankov}
    $\E$ is a strongly Jankov base.
\end{lemma}

\begin{theorem}[Strong completeness]\label{thm:strongcomplete}
$\Gamma \Vdash_{\SJK} A$ implies $\Gamma \vdash_{\JK} A$. %
\end{theorem}

Jankov logic thus admits two distinct base-extension semantics, with completeness established by two adaptations of the intuitionistic argument. The modularity lies in separating the general completeness construction from the conditions on bases needed to validate the additional axioms. Beyond the verification of those conditions, the constructions {\em do not depend} on the particular form of weak excluded middle. They therefore provide two routes for proving completeness for {\em any} axiomatic extension of intuitionistic logic.

\section{Concluding remarks}
We have developed a sequent-based formulation of base-extension semantics that provides a common treatment of classical, intuitionistic, and substructural reasoning. Its central ingredients are atomic derivability, support indexed by multisets of resources, and semantic clauses guided by the consequences licensed by the left introduction rules. Together, these ingredients make explicit how the interpretation of connectives interacts with structural principles and restrictions on contexts.

The soundness and completeness results show how this construction can be instantiated for classical and intuitionistic propositional logic and their multiplicative-additive linear fragments. A significant feature is that atomic weakening and contraction suffice to establish the corresponding semantic principles for arbitrary formulas. In particular, the treatment of contraction reveals a distinction between syntactic and semantic accounts: although contraction cannot generally be restricted to atoms in a sequent calculus, its semantic counterpart can be secured through conditions on atomic bases. The framework thus provides both adequacy results and an account of how structural behaviour propagates through the interpretation of logical connectives.

The treatment of Jankov logic extends this investigation beyond structural variations. Weakly and strongly Jankov bases impose different requirements on the support of weak excluded middle, yet both recover Jankov consequence. Their completeness proofs also exhibit two approaches: a constructive extension of the simulation-base method and a countermodel construction using prime theories and simulation bases. This case study shows how the common framework can accommodate an intermediate logic while making visible the additional requirements imposed by its characteristic axiom.

A natural next step is to identify general conditions under which axiomatic extensions of intuitionistic logic admit analogous treatments. This includes determining when constructive simulation-base arguments suffice and when additional structure on bases is needed. Extensions to modal logics and cyclic proof systems offer further opportunities to investigate how the choice of atomic rules and admissible extensions shapes proof-theoretic validity.

\newpage
\bibliographystyle{splncs03}
\bibliography{references}

@inproceedings{Ciabattoni2013,
  author    = {Ciabattoni, Agata and Maffezioli, Paolo and Spendier, Lara},
  title     = {Hypersequent and Labelled Calculi for Intermediate Logics},
  booktitle = {Automated Reasoning with Analytic Tableaux and Related Systems (TABLEAUX 2013)},
  editor    = {Galmiche, Didier and Larchey-Wendling, Dominique},
  series    = {Lecture Notes in Computer Science},
  volume    = {8123},
  pages     = {81--96},
  year      = {2013},
  publisher = {Springer},
  doi       = {10.1007/978-3-642-40537-2_9}
}

@article{Barroso-Nascimento2025,
  author    = {Barroso-Nascimento, Victor and Pereira, Luiz Carlos and Pimentel, Elaine},
  title     = {An ecumenical view of proof-theoretic semantics},
  journal   = {Synthese},
  year      = {2025},
  volume    = {206},
  number    = {4},
  pages     = {197},
  issn      = {1573-0964},
  doi       = {10.1007/s11229-025-05269-z},
  url       = {https://doi.org/10.1007/s11229-025-05269-z}
}

@Inbook{Indrzejczak2022,
author="Indrzejczak, Andrzej",
editor="Citkin, Alex
and Vandoulakis, Ioannis M.",
title="Dialogues and Proofs; Yankov's Contribution to Proof Theory",
bookTitle="V.A. Yankov on Non-Classical Logics, History and Philosophy of Mathematics",
year="2022",
publisher="Springer International Publishing",
address="Cham",
pages="53--70",
isbn="978-3-031-06843-0",
doi="10.1007/978-3-031-06843-0_3",
url="https://doi.org/10.1007/978-3-031-06843-0_3"
}

@article{Yankov1968,
  author  = {Vadim A. Yankov},
  title   = {The calculus of the weak ``Law of excluded middle''},
  journal = {Mathematics of the USSR-Izvestiya},
  volume  = {2},
  number  = {5},
  pages   = {997--1004},
  year    = {1968},
  doi     = {10.1070/IM1968v002n05ABEH000690}
}

@book{CitkinVandoulakis2022,
  editor    = {Alex Citkin and Ioannis M. Vandoulakis},
  title     = {V.A. Yankov on Non-Classical Logics, History and Philosophy of Mathematics},
  series    = {Outstanding Contributions to Logic},
  volume    = {24},
  publisher = {Springer Nature Switzerland},
  year      = {2022},
  doi       = {10.1007/978-3-031-06843-0},
  isbn      = {978-3-031-06842-3}
}

@article{DBLP:journals/ndjfl/Brunnler06,
  author       = {Kai Br{\"{u}}nnler},
  title        = {Locality for Classical Logic},
  journal      = {Notre Dame J. Formal Log.},
  volume       = {47},
  number       = {4},
  pages        = {557--580},
  year         = {2006},
  url          = {https://doi.org/10.1305/ndjfl/1168352668},
  doi          = {10.1305/NDJFL/1168352668},
  bibsource    = {dblp computer science bibliography, https://dblp.org}
}

@inproceedings{DBLP:journals/corr/abs-2504-08349,
  author       = {Victor Barroso{-}Nascimento and
                  Ekaterina Piotrovskaya and
                  Elaine Pimentel},
  editor       = {Clemens Kupke and
                  Stefan Milius},
  title        = {A Proof-Theoretic Approach to the Semantics of Classical Linear Logic},
  booktitle    = {Proceedings of the 41st Conference on the Mathematical Foundations
                  of Programming Semantics, {MFPS} XLI, University of Strathclyde, Glasgow,
                  UK, June 16-21, 2025},
  series       = {Electronic Notes in Theoretical Informatics and Computer Science},
  volume       = {5},
  publisher    = {EpiSciences},
  year         = {2025},
  month        = apr,
  url          = {https://doi.org/10.46298/entics.16888},
  doi          = {10.46298/ENTICS.16888},
  bibsource    = {dblp computer science bibliography, https://dblp.org}
}

@Inbook{Corsi2008,
author="Corsi, Giovanna",
editor="Lupacchini, Rossella
and Corsi, Giovanna",
title="The Logic of the Weak Excluded Middle: A Case Study of Proof-Search",
bookTitle="Deduction, Computation, Experiment: Exploring the Effectiveness of Proof",
year="2008",
publisher="Springer Milan",
address="Milano",
pages="95--116",
isbn="978-88-470-0784-0",
doi="10.1007/978-88-470-0784-0_6",
url="https://doi.org/10.1007/978-88-470-0784-0_6"
}

@article{Steinberger2011,
  author  = {Steinberger, Florian},
  title   = {Why Conclusions Should Remain Single},
  journal = {Journal of Philosophical Logic},
  volume  = {40},
  number  = {3},
  pages   = {333--355},
  year    = {2011},
  doi     = {10.1007/s10992-010-9153-3}
}

@InCollection{sep-logic-intuitionistic,
	author       =	{Moschovakis, Joan},
	title        =	{{Intuitionistic Logic}},
	booktitle    =	{The {Stanford} Encyclopedia of Philosophy},
	editor       =	{Edward N. Zalta and Uri Nodelman},
	howpublished =	{\url{https://plato.stanford.edu/archives/sum2024/entries/logic-intuitionistic/}},
	year         =	{2024},
	edition      =	{{S}ummer 2024},
	publisher    =	{Metaphysics Research Lab, Stanford University}
}

@article{lukasiewicz1952intuitionistic,
  author  = {Lukasiewicz, Jan},
  title   = {On the intuitionistic theory of deduction},
  journal = {Indagationes Mathematicae},
  volume  = {14},
  pages   = {202--212},
  year    = {1952}
}

@article{LegrisMolina2001,
  author    = {Legris, Javier and Molina, Jorge A.},
  title     = {Intuitionistic Logic Considered As An Extension of Classical Logic: Some Critical Remarks},
  journal   = {Philosophia Scientiae},
  volume    = {5},
  number    = {2},
  pages     = {27--50},
  year      = {2001},
  url       = {}
}

@article{Hosoi1986,
  author    = {Tsutomu Hosoi},
  title     = {Pseudo two-valued evaluation method for intermediate logics},
  journal   = {Studia Logica},
  volume    = {45},
  number    = {1},
  pages     = {3--8},
  year      = {1986},
  doi       = {10.1007/BF01881544},
  url       = {https://link.springer.com/article/10.1007/BF01881544}
}

@misc{buzoku2026ILL,
      title={A Proof-theoretic Semantics for Intuitionistic Linear Logic}, 
      author={Yll Buzoku},
      year={2026},
      eprint={2402.01982},
      archivePrefix={arXiv},
      primaryClass={cs.LO},
      url={https://arxiv.org/abs/2402.01982}, 
}

@article{DBLP:journals/sLogica/GheorghiuGP26,
  author       = {Alexander V. Gheorghiu and
                  Tao Gu and
                  David J. Pym},
  title        = {Proof-theoretic Semantics for Intuitionistic Multiplicative Linear
                  Logic},
  journal      = {Studia Logica},
  volume       = {114},
  number       = {2},
  pages        = {451--511},
  year         = {2026},
  url          = {https://doi.org/10.1007/s11225-024-10158-6},
  doi          = {10.1007/S11225-024-10158-6},
  bibsource    = {dblp computer science bibliography, https://dblp.org}
}

@article{stafford2025logics,
author     = {Stafford, Will and Piecha, Thomas and Schroeder-Heister, Peter},
title      = {Logics of Proof-Theoretic Validity},
journal   = {Topoi},
year      = {2025},
volume    = {45},
number    = {2},
pages     = {643--661},
doi       = {10.1007/s11245-025-10248-7}
}

@article{DBLP:journals/rsl/Ayhan26,
  author       = {Sara Ayhan},
  title        = {Contradictions without Negation and a Proof-Theoretic, Bilateralist
                  Account of Connexive Logics},
  journal      = {Rev. Symb. Log.},
  volume       = {19},
  number       = {2},
  pages        = {375--402},
  year         = {2026},
  url          = {https://doi.org/10.1017/s1755020326101063},
  doi          = {10.1017/S1755020326101063},
  bibsource    = {dblp computer science bibliography, https://dblp.org}
}

@article{Daragona-Prawitz,
author  = {Antonio Piccolomini d'Aragona and Dag Prawitz},
year    = {2026},
title   = {Some Variants of Proof-theoretic Semantics and Their Relations with Intuitionistic Logic},
journal = {Topoi},
pages   = {663--674},
volume  = {45},
issue = {2},
doi     = {10.1007/s11245-026-10408-3},
}

@inproceedings{DBLP:conf/tableaux/BarrosoNascimentoPP25,
  author       = {Victor Barroso{-}Nascimento and
                  Ekaterina Piotrovskaya and
                  Elaine Pimentel},
  editor       = {Gian Luca Pozzato and
                  Tarmo Uustalu},
  title        = {A Sequent Calculus Perspective on Base-Extension Semantics},
  booktitle    = {Automated Reasoning with Analytic Tableaux and Related Methods - 34th
                  International Conference, {TABLEAUX} 2025, Reykjavik, Iceland, September
                  27-29, 2025, Proceedings},
  series       = {Lecture Notes in Computer Science},
  volume       = {15980},
  pages        = {3--21},
  publisher    = {Springer},
  year         = {2025},
  url          = {https://doi.org/10.1007/978-3-032-06085-3\_1},
  doi          = {10.1007/978-3-032-06085-3\_1},
  bibsource    = {dblp computer science bibliography, https://dblp.org}
}

@article{pts-91,
  author    = {Peter Schroeder-Heister},
  title = {Uniform Proof-Theoretic Semantics for Logical Constants (Abstract)},
  journal = {Journal of Symbolic Logic},
  volume = {56},
  pages = {1142},
  year = {1991}
}

@article{schroeder2006validity,
 author       = {Peter Schroeder{-}Heister},
  title        = {Validity Concepts in Proof-theoretic Semantics},
  journal      = {Synth.},
  volume       = {148},
  number       = {3},
  pages        = {525--571},
  year         = {2006},
  url          = {https://doi.org/10.1007/s11229-004-6296-1},
  doi          = {10.1007/S11229-004-6296-1},
  bibsource    = {dblp computer science bibliography, https://dblp.org}
}

@incollection{Gentzen1969,
	publisher = {North-Holland Publishing Company},
	year = {1969},
	author = {Gerhard Gentzen},
	editor = {Manfred E. Szabo},
	booktitle = {{The Collected Papers of Gerhard Gentzen}},
	title = {{Investigations into Logical Deduction}}
}

@article{Sandqvist2015IL,
	year = {2015},
	doi = {10.1093/jigpal/jzv021},
	pages = {719--731},
	number = {5},
	title = {Base-Extension Semantics for Intuitionistic Sentential Logic},
	volume = {23},
	journal = {Logic J. of  the IGPL},
	author = {Tor Sandqvist}
}

@InCollection{sep-proof-theoretic-semantics,
	author       =	{Schroeder-Heister, Peter},
	title        =	{{Proof-Theoretic Semantics}},
	address = {Palo Alto},
	booktitle    =	{The {Stanford} Encyclopedia of Philosophy},
	editor       =	{Edward N. Zalta and Uri Nodelman},
	howpublished =	{\url{https://plato.stanford.edu/archives/sum2024/entries/proof-theoretic-semantics/}},
	year         =	{2024},
	edition      =	{{S}ummer 2024},
	publisher    =	{Metaphysics Research Lab, Stanford University}
}

@article{piecha2015failure,
  author = {Thomas Piecha and Wagner de Campos Sanz and Peter Schroeder{-}Heister},
	doi = {10.1007/s10992-014-9322-x},
	journal = {Journal of Philosophical Logic},
	number = {3},
	pages = {321--335},
	publisher = {Springer Netherlands},
	title = {Failure of Completeness in Proof-Theoretic Semantics},
	volume = {44},
	year = {2015}
}

@Article{girard87tcs,
  author =       "Jean-Yves Girard",
  title =        "Linear Logic",
  journal =      "Theoretical Computer Science",
  volume =       "50",
  year =         "1987",
  pages =        "1--102",
  	doi = {10.1016/0304-3975(87)90045-4},
}

@article{DBLP:journals/aml/DyckhoffN12,
  author    = {Roy Dyckhoff and
               Sara Negri},
  title     = {Proof analysis in intermediate logics},
  journal   = {Arch. Math. Log.},
  volume    = {51},
  number    = {1-2},
  pages     = {71--92},
  year      = {2012},
  url       = {http://dx.doi.org/10.1007/s00153-011-0254-7},
  doi       = {10.1007/s00153-011-0254-7},
  bibsource = {dblp computer science bibliography, http://dblp.org}
}

@article{Prawitz-2006,
 doi       = {10.1007/s11229-004-6295-2},
 title     = {Meaning Approached Via Proofs},
 author    = {Dag Prawitz},
 publisher = {Springer Netherlands},
 journal   = {Synthese},
 issn      = {0039-7857,1573-0964},
 year      = {2006},
 volume    = {148},
 issue     = {3},
 pages     = {507--524},
 url       = {}
}

@book{Dalen,
   title =     {Logic and Structure},
   author =    {Dirk van Dalen},
   publisher = {Springer},
   isbn =      {},
   year =      {2013},
   series =    {Universitext},
   edition =   {5th},
   volume =    {},
   url =       {}
}

@article{Sandqvist,
    author = {Sandqvist, Tor},
    title = "{Classical logic without bivalence}",
    journal = {Analysis},
    volume = {69},
    number = {2},
    pages = {211-218},
    year = {2009},
    month = {04},
    issn = {0003-2638},
    doi = {10.1093/analys/anp003},
    url = {https://doi.org/10.1093/analys/anp003},
    eprint = {https://academic.oup.com/analysis/article-pdf/69/2/211/296090/anp003.pdf},
}

\appendix

\section{Additive and multiplicative simulation bases}\label{app:sim-base}
Figures~\ref{fig:comp} and~\ref{fig:compmult} present the additive and multiplicative simulation bases, respectively, used in the completeness proof presented in Section~\ref{sec:soundcomp}.
\begin{figure}[h]
\[\infer[L \&]{p^{A \& B}, \Gamma_{\At} \Rightarrow \Delta_{\At}}{p^{A},\Gamma_{\At} \Rightarrow \Delta_{\At}}
\quad
\infer[L \&]{p^{A \& B}, \Gamma_{\At} \Rightarrow \Delta_{\At}}{p^{B},\Gamma_{\At} \Rightarrow \Delta_{\At}}
\quad 
\infer[R \&]{\Gamma_{\At} \Rightarrow \Delta_{\At}, p^{A \& B}}{\Gamma_{\At} \Rightarrow  \Delta_{\At}, p^{A}& \Gamma_{\At} \Rightarrow \Delta_{\At}, p^{B}}
\]
\smallskip
\[
\infer[L\oplus]{p^{A \oplus B}, \Gamma_{\At}, \Rightarrow \Delta_{\At}}{p^{A}, \Gamma_{\At} \Rightarrow  \Delta_{\At} & p^{B}, \Gamma_{\At} \Rightarrow  \Delta_{\At}}\quad
\infer[R \oplus_{1}]{\Gamma_{\At} \Rightarrow \Delta_{\At}, p^{A \oplus B}}{\Gamma_{\At} \Rightarrow \Delta_{\At}, p^{A}}
\quad
\infer[R \oplus_{2}]{\Gamma_{\At} \Rightarrow \Delta_{\At}, p^{A \oplus B}}{\Gamma_{\At} \Rightarrow \Delta_{\At}, p^{B}}
\]
\smallskip
\[
\infer[L \to]{p^{A \to B}, \Gamma_{\At} \Rightarrow \Delta_{\At}}{\Gamma_{\At} \Rightarrow  \Delta_{\At}, p^{A} & p^{B}, \Gamma_{\At} \Rightarrow \Delta_{\At}}
\quad 
\infer[R \to_{1}]{\Gamma_{\At} \Rightarrow \Delta_{\At}, p^{A \to B}}{p^{A}, \Gamma_{\At} \Rightarrow \Delta_{\At}}
\quad 
\infer[R \to_{2}]{\Gamma_{\At} \Rightarrow \Delta_{\At}, p^{A \to B}}{\Gamma_{\At} \Rightarrow \Delta_{\At}, p^{B}}
\]
\smallskip
\[
\infer[L0]{\Gamma_{\At}, p^{0} \Rightarrow \Delta_{\At}}{}
\quad
\infer[R\top]{\Gamma_{\At}\Rightarrow \Delta_{\At},p^\top}{}
\]

\caption{Rules of the additive simulation base $\U$.}\label{fig:comp}
\end{figure}

\begin{figure}[h!]
\[\infer[L \otimes]{p^{A \otimes B}, \Gamma_{\At} \Rightarrow \Delta_{\At}}{p^{A}, p^{B}, \Gamma_{\At} \Rightarrow \Delta_{\At}}
\quad
\infer[R \otimes]{\Gamma_{\At}, \Gamma'_{\At} \Rightarrow \Delta_{\At}, \Delta'_{\At} p^{A \otimes B}}{\Gamma_{\At} \Rightarrow  \Delta_{\At}, p^{A}& \Gamma'_{\At} \Rightarrow \Delta'_{\At}, p^{B}}
\]
\smallskip
\[
\infer[L\parr]{p^{A \parr B}, \Gamma_{\At}, \Gamma'_{\At}, \Rightarrow \Delta_{\At}, \Delta'_{\At}}{p^{A}, \Gamma_{\At} \Rightarrow  \Delta_{\At} & p^{B}, \Gamma'_{\At} \Rightarrow  \Delta'_{\At}}\quad
\infer[R \parr_{1}]{\Gamma_{\At} \Rightarrow \Delta_{\At}, p^{A \parr B}}{\Gamma_{\At} \Rightarrow \Delta_{\At}, p^{A}, p^{B}}
\]
\smallskip
\[
\infer[L \multimap]{p^{A \multimap B}, \Gamma_{\At}, \Gamma_{\At}' \Rightarrow \Delta_{\At}, \Delta_{\At}'}{\Gamma_{\At} \Rightarrow  \Delta_{\At}, p^{A} & p^{B}, \Gamma_{\At}' \Rightarrow \Delta_{\At}'}
\quad 
\infer[R \multimap]{\Gamma_{\At} \Rightarrow \Delta_{\At}, p^{A \multimap B}}{p^{A}, \Gamma_{\At} \Rightarrow \Delta_{\At}, p^{B}}
\]
\smallskip
\[
\infer[L \bot]{ p^{\bot} \Rightarrow \varnothing}{}
\quad
\infer[R \bot]{\Gamma_{\At} \Rightarrow \Delta_{\At}, p^{\bot}}{\Gamma_{\At} \Rightarrow \Delta_{\At}}
\quad 
\infer[L 1]{\Gamma_{\At}, p^{1} \Rightarrow \Delta_{\At}}{\Gamma_{\At} \Rightarrow \Delta_{\At}}
\quad 
\infer[R 1]{\varnothing \Rightarrow p^1}{}
\]
\caption{Rules of the multiplicative simulation base $\U$.}\label{fig:compmult}
\end{figure}

\newpage
\section{Proofs of structural properties -- Section~\ref{sec:soundcomp}}\label{app:struc}
\begin{definition}
    The {\em degree of a formula $A$} is the number of logical operators occurring in it. The {\em degree of a set $\Delta$} is the sum of the degree of all formulas in $\Delta$. Given a collection of $n$ sets $\Pi^{i}$ ($1 \leq i \leq n$), $\{\Pi^i\}_1^n$ is defined as the sum of degrees of all $\Pi^{i}$.
\end{definition}

\begin{proof}[of Lemma~\ref{lemma:non-atomicinference}]
    ($\Leftarrow$): This direction follows immediately from (Inf), as the statement holds in particular for atomic sets $\Pi^i\Atomic$.

    ($\Rightarrow$):  For the base case, if $\mathrm{deg}(\{\Pi^i\}_1^n)=0$, it follows that for all $i$, $\mathrm{deg}(\Pi^i)=0$. In other words, all sets only contain atoms, hence the result follows by (Inf). For our induction hypothesis, assume the statement holds for $\mathrm{deg}(\{\Pi^i\}_1^n)\leq k$. 

    Suppose $\mathrm{deg}(\{\Pi^i\}_1^n)=k+1$. Then, some set contains a complex formula $B$. Without loss of generality, say $\Pi^1 = \{\Omega, B\}$. We shall proceed by a structural induction on $B$, where in all cases we want to show that:
    \begin{align*}
        \text{if }&\supp{C}{\Sigma^1\Atomic}A^1, \Omega, B,\,\supp{C}{\Sigma^2\Atomic}A^2,\Pi^2,\dots,\,\supp{C}{\Sigma^n\Atomic}A^n,\Pi^n\\
        \text{ then }&\supp{C}{\Gamma\Atomic\cup(\bigcup\Sigma^i\Atomic)}\Delta,\Omega,B,\Pi^2,\dots,\Pi^n \qquad (\star).
    \end{align*}
    We shall demonstrate the case for one of the connectives here. The other cases are analogous.
    
    \begin{itemize}
        \item[$B = C_{1}\& C_{2}$:] We have that $\supp{C}{\Sigma^1\Atomic}A^1, \Omega, C_{1}\& C_{2}$. By ($\&$), and for $i$ equal to either 1 or 2, for all $\base{D}\supseteq\base{C}$, and for any atomic $\Theta\Atomic,\Lambda\Atomic$;
        \[
        (\star) \iff\text{if } C_i\supp{D}{\Theta\Atomic}\Lambda\Atomic \text{, then }\supp{D}{\Gamma\Atomic\cup\bigcup\Sigma^i\Atomic\cup\Theta\Atomic}\Delta,\Omega,\Lambda\Atomic,\Pi^2,\dots,\Pi^n
        \]
        Let $\base{D}\supseteq\base{C}$ such that $C_i\supp{D}{\Theta\Atomic}\Lambda\Atomic$. As $\supp{C}{\Sigma^1\Atomic}A^1, \Omega, C_{1}\& C_{2}$, we have by ($\&$) that $\supp{D}{\Sigma^1\Atomic\cup\Theta\Atomic}\Omega,\Lambda\Atomic,A^1$.
        
        Note that $\mathrm{deg}(\Lambda\Atomic)=0$ and $\mathrm{deg}(\varphi)\geq 1$, as $\varphi$ is a complex formula. Hence, $\mathrm{deg}(\{\Omega,\Lambda\Atomic,\Pi^2,\dots,\Pi^n\})$ is strictly less than $\mathrm{deg}(\{\Omega,\varphi,\Pi^2,\dots,\Pi^n\})$. As a result, we can apply the induction hypothesis: as $\supp{D}{\Sigma^1\Atomic\cup\Theta\Atomic}\Omega,\Lambda\Atomic,A^1$, and by monotonicity, $\supp{D}{\Sigma^i\Atomic}A^i,\Pi^i$ for $i=2,\dots,n$, we thus have that $\supp{D}{\Gamma\Atomic\cup\bigcup\Sigma\Atomic^i\cup\Theta\Atomic}\Delta,\Omega,\Lambda\Atomic,\Pi^2,\dots,\Pi^n$, as desired.
\end{itemize}             
\end{proof}

\begin{proof}[of Lemma~\ref{lemma:partialInf}]
  Without loss of generality, let $\Delta = \{ B_{1}, \ldots, B_{k}\}$.  Assume $\Gamma, \Delta \Vdash^{\Theta_{\At}}_{\B} \Sigma$ and $\Vdash^{\Omega^{m}_{\At}}_{\B} \Pi^{m}_{\At}, A_{m}$ for all $1 \leq m \leq k$. Now take any $\C \supseteq \B$ such that $\Vdash^{\Phi^{i}_{\At}}_{\C} B_{i}, \Psi^{i}_{\At}$ for all $1 \leq i \leq k$. By monotonicity we have $\Vdash^{\Omega^{m}_{\At}}_{\C} \Pi^{m}_{\At}, A_{m}$ for all such $m$, hence since $\Gamma, \Delta \Vdash^{\Theta_{\At}}_{\B} \Sigma$ we conclude $ \Vdash^{\Theta_{\At},  \Omega^{1}_{\At}, \ldots, \Omega^{n}_{\At}, \Phi^{1}_{\At}, \ldots, \Phi^{n}_{\At}}_{\B} \Sigma, \Pi^{1}_{\At}, \ldots, \Pi^{n}_{\At}, \Psi^{1}_{\At}, \ldots, \Psi^{n}_{\At}$. Since $\C$, $\Phi^{i}_{\At}, \Psi^{i}_{\At}$ were all arbitrary we conclude $\Delta \Vdash^{\Theta_{\At}, \Omega^{1}_{\At}, \ldots, \Omega^{n}_{\At}}_{\B} \Sigma, \Pi^{1}_{\At}, \ldots, \Pi^{n}_{\At}$.
\end{proof}

\begin{proof}[of Lemma~\ref{lemma:basecut}]
 Assume $\Gamma \Vdash^{\Theta_{\At}}_{\B} A, \Delta$ and $A, \Sigma \Vdash^{\Omega_{\At}}_{\B} \Pi$. If $\Gamma = \Sigma = \{\varnothing\}$ the result follows immediately by Lemma \ref{lemma:non-atomicinference}.  If neither set is empty, let $\Gamma = \{B_{1}, \ldots, B_{n}\}$, $\Sigma = \{C_{1}, \ldots, C_{k}\}$. Now pick any $\C \supseteq \B$ such that $\Vdash^{\Phi^{m}_{\At}}_{\C} \Psi^{m}_{\At}, B_{m}$ for $1 \leq m \leq n$ and $\Vdash^{\Upsilon^{i}_{\At}}_{\C} \Lambda^{i}_{\At}, C_{i}$ for $1 \leq i \leq k$. Since $\Gamma \Vdash^{\Theta_{\At}}_{\B} A, \Delta$ we have $ \Vdash^{\Theta_{\At}, \Phi^{1}_{\At}, \ldots, \Phi^{n}_{\At}}_{\C} \Psi^{1}_{\At}, \ldots, \Psi^{n}_{\At}, A, \Delta,$. Since $A, \Sigma \Vdash^{\Omega_{\At}}_{\B} \Pi$ by Lemma \ref{lemma:partialInf} we have $A \Vdash^{\Omega_{\At}, \Upsilon^{1}_{\At}, \ldots, \Upsilon^{k}_{\At}}_{\B} \Lambda^{1}_{\At}, \ldots, \Lambda^{k}_{\At}, \Pi$, so we can apply Lemma \ref{lemma:non-atomicinference} to conclude $\Vdash^{\Theta_{\At}, \Omega_{\At}, \Upsilon^{1}_{\At}, \ldots, \Upsilon^{k}_{\At}}_{\B} \Lambda^{1}_{\At}, \ldots, \Lambda^{k}_{\At}, \Pi, \Delta$. Since the $\Phi^{m}_{\At}, \Psi^{m}_{\At}, \Upsilon^{i}_{\At}, \Lambda^{i}_{\At}$ were all arbitrary we conclude $\Gamma, \Sigma \Vdash^{\Theta_{\At}, \Omega_{\At}}_{\B} \Delta, \Pi$. The proofs are analogous if only one set is empty, the difference being that if $\Gamma = \{\varnothing\}$ we conclude $\Vdash^{\Theta_{\At}}_{\C} A, \Delta$ by Monotonicity before applying Lemma \ref{lemma:non-atomicinference}, and if $\Sigma = \{\varnothing\}$ we omit the steps of assuming $\Vdash^{\Upsilon^{i}_{\At}}_{\C} \Lambda^{i}_{\At}, C_{i}$ for $1 \leq i \leq k$ and applying Lemma \ref{lemma:partialInf}.
\end{proof}

\begin{proof}[of Proposition~\ref{prop:subs}]
For a counterexample, let $\Gamma_{\At} = \{p_{1}, \ldots, p_{n}\}$. Define $\B$ as the base containing all rules of the following shape for all $\Phi^{1}_{\At}, \ldots, \Phi^{n}_{\At}, \Psi^{1}_{\At}, \ldots, \Psi^{n}_{\At}$, and no other rules
\begin{prooftree}
    \AxiomC{$\Phi^{1}_{\At} \Rightarrow p_{1}, \Psi^{1}_{\At}$}
    \AxiomC{$\ldots$}
     \AxiomC{$\Phi^{n}_{\At} \Rightarrow p_{n}, \Psi^{n}_{\At}$}
     \TrinaryInfC{$\Gamma_{\At}, \Phi^{1}_{\At}, \ldots, \Phi^{n}_{\At} \Rightarrow \Psi^{1}_{\At}, \ldots, \Psi^{n}_{\At}, \Delta_{\At}$}
\end{prooftree}
 Take any $\C \supseteq \B$ such that $\Vdash^{\Omega^{m}_{\At}}_{\C} \Pi^{m}_{\At}, p_{m}$ for all $1 \leq m \leq n$. Then $\vdash_{\C} \Omega^{m}_{\At} \Rightarrow \Pi^{m}_{\At}, p_{m}$ for all $m$, so by applying the appropriate rule of $\B$ we get $\vdash_{\C} \Gamma_{\At}, \Omega^{1}_{\At}, \ldots, \Omega^{n}_{\At} \Rightarrow \Pi^{1}_{\At}, \ldots , \Pi^{n}_{\At}, \Delta_{\At}$ and then $\Vdash^{\Gamma_{\At},  \Omega^{1}_{\At}, \ldots, \Omega^{n}_{\At}}_{\C} \Pi^{1}_{\At}, \ldots , \Pi^{n}_{\At}, \Delta_{\At}$. Since the $\C$, $\Omega^{m}_{\At}$ and $\Pi^{m}_{\At}$ were all arbitrary, we conclude $\Gamma_{\At} \Vdash^{\Theta_{\At}}_{\B} \Delta_{\At}$. But since there are no rules with empty premises in $\B$ we have $\nvdash_{\B}, \Upsilon_{\At} \Rightarrow \Lambda_{\At}$ for all $\Upsilon_{\At}$ and $\Lambda_{\At}$, so in particular $\nvdash_{\B}, \Theta_{\At} \Gamma_{\At} \Rightarrow \Delta_{\At}$.

Now suppose that $\B$ is closed under the atomic rules \init\ and \Cutc. If $\Gamma_{\At}$ is empty the result is immediate. Otherwise, without loss of generality, let $\Gamma_{\At} = \{p_{1}, \ldots , p_{n}\}$. Assume $\Gamma_{\At} \Vdash^{\Theta_{\At}}_{\B} \Delta_{\At}$. By \init\ we conclude $\vdash_{\B} p_{m} \Rightarrow p_{m}$ and thus  $\Vdash^{p_{m}}_{\B} p_{m}$ for all $1 \leq m \leq n$. By (Inf) we conclude $\Vdash^{\Gamma_{\At}, \Theta_{\At}}_{\B} \Delta_{\At}$, so $\vdash_{\B} \Theta_{\At}, \Gamma_{\At} \Rightarrow \Delta_{\At} $. For the converse, assume $\vdash_{\B} \Theta_{\At}, \Gamma_{\At} \Rightarrow \Delta_{\At} $. Pick any $\C \supseteq \B$ such that $\Vdash^{\Omega^{m}_{\At}}_{\C} p_{m}, \Pi^{m}_{\At}$ for all $1 \leq m \leq n$. Then $\vdash_{\C} \Omega^{m}_{\At}  \Rightarrow p_{m}, \Pi^{m}_{\At}$ for all such $m$, and since $\vdash_{\B} \Theta_{\At}, \Gamma_{\At} \Rightarrow \Delta_{\At} $ we obtain $\vdash_{\C} \Theta_{\At},  \Omega^{1}_{\At} , \ldots, \Omega^{n}_{\At} \Rightarrow   \Pi^{1}_{\At} , \ldots, \Pi^{n}_{\At}   \Delta_{\At}$ by applying Cut $n$ times, so $\Vdash^{\Theta_{\At},  \Omega^{1}_{\At} , \ldots, \Omega^{n}_{\At}}_{\C}   \Pi^{1}_{\At} , \ldots, \Pi^{n}_{\At}   \Delta_{\At}$. Since $\C$ and the $\Omega^{m}_{\At}, \Pi^{m}_{\At}$ were all arbitrary we apply (Inf) to conclude $\Gamma_{\At} \Vdash^{\Theta_{\At}}_{\B} \Delta_{\At}$.
\end{proof}

\begin{lemma}\label{lemma:auxiliarycontractionleft} 

      For any $\B$, if $(\Vdash^{p, p, \Gamma_{\At}}_{\mathcal{C}} \Delta_{\At}$ implies  $\Vdash^{p, \Gamma_{\At}}_{\C} \Delta_{\At})$ holds for all $\C \supseteq \B$, $\Gamma_{\At}, \Delta_{\At}$, $p$ then $(\Vdash^{p, p, \Gamma_{\At}}_{\C} \Delta$ implies  $\Vdash^{p, \Gamma_{\At}}_{\C} \Delta)$ holds for all $\C \supseteq \B$, $\Gamma_{\At}, \Delta$, $p$.
    
\end{lemma}
\begin{proof}
  Assume that $(\Vdash^{p, p, \Gamma_{\At}}_{\C} \Delta_{\At}$ implies  $\Vdash^{p, \Gamma_{\At}}_{\C} \Delta_{\At})$ holds for all $\C \supseteq \B, \Gamma_{\At}, \Delta_{\At}$, $p$.  Now pick an arbitrary $\C \supseteq \B$ and assume $\Vdash^{p, p, \Gamma_{\At}}_{\C} \Delta$. If $\Delta = \Delta_{\At}$ the assumptions already yield $\Vdash^{p, \Gamma_{\At}}_{\C} \Delta$. If $\Delta$ is not a set of atoms, there is some non-atomic $B$ with $B \in \Delta$. The proof for $\bot$ and $0$ is trivial. The proof is analogous for all remaining connectives, so we only show the case were $B = C \parr D$ and $\Delta = \Sigma \cup \{ C \parr D\}$. By assumption, $\Vdash^{p, p, \Gamma_{\At}}_{\C} \Sigma, C \parr D$. Now pick any $\D \supseteq \C$ and arbitrary $\Omega_{\At}, \Pi_{\At}, \Phi_{\At}, \Psi_{\At}$ such that $C \Vdash^{\Omega_{\At}}_{\D} \Pi_{\At}$ and $D\Vdash^{\Phi_{\At}}_{\D} \Psi_{\At}$. Then $\Vdash^{p, p, \Gamma_{\At}, \Omega_{\At}, \Psi_{\At}}_{\D} \Sigma, \Pi_{\At}, \Psi_{\At}$. Since the degree of $\Sigma$ is lesser than that of $\Delta$ the induction hypothesis yields $\Vdash^{p, \Gamma_{\At}, \Omega_{\At}, \Psi_{\At}}_{\D} \Sigma, \Pi_{\At}, \Psi_{\At}$, so by the semantic clause for disjunction we conclude $\Vdash^{p,  \Gamma_{\At}}_{\C} \Sigma, C \parr D$.
\end{proof}

\begin{lemma}\label{lemma:auxiliarycontractionright} 

      For any $\B$, if $(\Vdash^{\Gamma_{\At}}_{\C} \Delta_{\At}, p, p$ implies  $\Vdash^{ \Gamma_{\At}}_{\C} \Delta_{\At}, p)$ holds for all $\C \supseteq \B$, $\Gamma_{\At}, \Delta_{\At}$, $p$ then $(\Vdash^{ \Gamma_{\At}}_{\C} \Delta, p, p$ implies  $\Vdash^{\Gamma_{\At}}_{\C} \Delta, p)$ holds for all $\C \supseteq \B$, $\Gamma_{\At}, \Delta$, $p$.
    
\end{lemma}
\begin{proof}
Assume that $(\Vdash^{\Gamma_{\At}}_{\C} \Delta_{\At}, p, p$ implies  $\Vdash^{ \Gamma_{\At}}_{\C} \Delta_{\At}, p)$ holds for all $\C \supseteq \B$, $\Gamma_{\At}, \Delta_{\At}$, $p$. Now take an arbitrary $\C \supseteq \B$ and assume $\Vdash^{ \Gamma_{\At}}_{\C} \Delta, p, p$.  If $\Delta = \Delta_{\At}$ the assumptions already yield $\Vdash^{ \Gamma_{\At}}_{\C} \Delta,p$. If $\Delta$ is not a set of atoms, there is some non-atomic $B$ with $B \in \Delta$. Once again, the proof for $\bot$ and $0$ is trivial and the cases are analogous for all other connectives, so we only show the case where $B = C \parr D$ and $\Delta = \Sigma \cup \{ C \parr D\} \cup \{p\} \cup \{p\}$. By assumption, $\Vdash^{ \Gamma_{\At}}_{\C} \Sigma, C \parr D, p, p$. Now take any $\D \supseteq \C$ and arbitrary $\Omega_{\At}, \Pi_{\At}, \Phi_{\At}, \Psi_{\At}$ such that $C \Vdash^{\Omega_{\At}}_{\D} \Pi_{\At}$ and $D\Vdash^{\Phi_{\At}}_{\D} \Psi_{\At}$. Then $\Vdash^{ \Gamma_{\At}, \Omega_{\At}, \Psi_{\At}}_{\D} \Sigma, \Pi_{\At}, \Psi_{\At}, p, p$. Since the degree of $\Sigma$ is lower than that of $\Delta$ the induction hypothesis yields $\Vdash^{ \Gamma_{\At}, \Omega_{\At}, \Psi_{\At}}_{\D} \Sigma, \Pi_{\At}, \Psi_{\At}, p$, so by the semantic clause for disjunction we conclude $\Vdash^{ \Gamma_{\At}}_{\C} \Sigma, C \parr D, p$.
\end{proof}

\begin{proof}[of Theorem~\ref{thm:contractionleftright}]
($\Rightarrow$): For this direction, we actually prove something stronger: if $(\Gamma, A, A \Vdash^{\Theta_{\At}}_{\C} \Delta$ implies $\Gamma, A \Vdash^{\Theta_{\At}}_{\C} \Delta)$  for all $\C \supseteq \B$, $\Gamma_{\At}, \Delta$, $p$ then $(\Vdash^{p, p, \Gamma_{\At}}_{\C} \Delta_{\At}$ implies  $\Vdash^{p, \Gamma_{\At}}_{\C} \Delta_{\At})$ for all  $\C \supseteq \B$, $\Gamma_{\At}, \Delta_{\At}$, $p$, and if $(\Gamma \Vdash^{\Theta_{\At}}_{\C} \Delta, A, A$ implies $\Gamma \Vdash^{\Theta_{\At}}_{\C} \Delta, A)$ for all $\C \supseteq \B$, $\Gamma_{\At}, \Delta$, $p$ then $(\Vdash^{ \Gamma_{\At}}_{\C} \Delta_{\At}, p, p$ implies  $\Vdash^{ \Gamma_{\At}}_{\C} \Delta_{\At}, p)$ for all $\C \supseteq \B$, $\Gamma_{\At}, \Delta_{\At}$, $p$.

Assume that $(\Gamma, A, A \Vdash^{\Theta_{\At}}_{\C} \Delta$ implies $\Gamma, A \Vdash^{\Theta_{\At}}_{\C} \Delta$) holds for all $A$ and all $\C \supseteq \B$, $\Theta_{\At}, \Gamma, \Delta$. Now pick an arbitrary $\C \supseteq \B$ and assume $\Vdash^{p, p, \Gamma_{\At}}_{\C} \Delta_{\At}$. Then $\vdash_{\C} p, p, \Gamma_{\At} \Rightarrow \Delta_{\At}$.  Pick any $\D \supseteq \C$ such that $\Vdash^{\Theta_{\At}}_{\D} \Omega_{\At}, p$ and $\Vdash^{\Phi_{\At}}_{\D} \Psi_{\At}, p$. Then $\vdash_{\D} \Theta_{\At} \Rightarrow \Omega_{\At}, p$ and $\vdash_{\D} \Phi_{\At} \Rightarrow \Psi_{\At}, p$. Since $\vdash_{\C} \Gamma_{\At}, p, p \Rightarrow \Delta_{\At}$ we also have $\vdash_{\D} \Gamma_{\At}, p, p \Rightarrow \Delta_{\At}$, so by applying Cut twice we get $\vdash_{\D} \Gamma_{\At}, \Theta_{\At}, \Phi_{\At} \Rightarrow \Delta_{\At}, \Omega_{\At}, \Psi_{\At}$, hence $\Vdash^{\Gamma_{\At}, \Theta_{\At}, \Phi_{\At}}_{\D} \Omega_{\At}, \Psi_{\At}, \Delta_{\At}$. Since $\Theta_{\At}, \Omega_{\At}, \Phi_{\At}, \Psi_{\At}$ and $\D$ were arbitrary we conclude $p, p \Vdash^{\Gamma^{\At}}_{\C} \Delta_{\At}$. Then we put $\Gamma= \{\varnothing \}$, $\Delta = \Delta_{\At}$,  $\Theta_{\At} = \Gamma_{\At}$ and $A = p$ to conclude $p \Vdash^{\Gamma^{\At}}_{\C} \Delta_{\At}$. By Axiom, $\vdash_{\C} p \Rightarrow p$ and then $\Vdash^{p}_{\C}p$, so by (Inf) we conclude $\Vdash^{p, \Gamma_{\At}}_{\C} \Delta_{\At}$, as desired.

Assume that $(\Gamma \Vdash^{\Theta_{\At}}_{\C} \Delta, A, A$ implies $\Gamma \Vdash^{\Theta_{\At}}_{\C} \Delta, A)$ holds for all $\C \supseteq \B$, $A, \Theta_{\At}, \Gamma, \Delta$. Now pick an arbitrary $\C \supseteq \B$ and assume $\Vdash^{ \Gamma_{\At}}_{\C} \Delta_{\At}, p, p$. After putting $\Gamma= \{\varnothing \}$, $\Delta = \Delta_{\At}$,  $\Theta_{\At} = \Gamma_{\At}$ and $A = p$ we can use the original assumption to conclude $\Vdash^{ \Gamma_{\At}}_{\C} \Delta_{\At}, p$, so the result follows from arbitrariness of $\C \supseteq \B$, $\Gamma_{\At}$, $\Delta_{\At}$, $p$.

($\Leftarrow$): Assume that $((\Vdash^{p, p, \Gamma_{\At}}_{\C} \Delta_{\At}$ implies  $\Vdash^{p, \Gamma_{\At}}_{\C} \Delta_{\At})$ and $(\Vdash^{ \Gamma_{\At}}_{\C} \Delta_{\At}, p, p$ implies  $\Vdash^{ \Gamma_{\At}}_{\C} \Delta_{\At}, p)$) holds for all $\C \supseteq \B$, $\Gamma_{\At}, \Delta_{\At}$,$p$. 

Pick an arbitrary $\C \supseteq \B$ and assume $\Gamma, A, A \Vdash^{\Theta_{\At}}_{\C} \Delta$. Let $\Gamma = \{B_{1}, \ldots, B_{n}\}$. Pick any  $\D \supseteq \C$ such that $\Vdash^{\Phi^{m}_{\At}}_{\D} B_{m}, \Psi^{m}_{\At}$ for all $1 \leq m \leq n$ and $\Vdash^{\Upsilon_{\At}}_{\D} \Lambda_{\At}, A$. By applying (Inf) and using $\Vdash^{\Upsilon_{\At}}_{\D} \Lambda_{\At}, A$ twice we conclude $\Vdash^{\Theta_{\At}, \Phi^{1}_{\At}, \ldots, \Phi^{n}_{\At}, \Upsilon_{\At}, \Upsilon_{\At}}_{\D} \Delta,  \Psi^{1}_{\At}, \ldots, \Psi^{n}_{\At}, \Lambda_{\At}, \Lambda_{\At}$. By applying Lemmas \ref{lemma:auxiliarycontractionleft} and \ref{lemma:auxiliarycontractionright} a sufficient amount of times we conclude $\Vdash^{\Theta_{\At}, \Phi^{1}_{\At}, \ldots, \Phi^{n}_{\At}, \Upsilon_{\At}}_{\D} \Delta,  \Psi^{1}_{\At}, \ldots, \Psi^{n}_{\At}, \Lambda_{\At}$, so by (Inf) we conclude $\Gamma, A, A \Vdash^{\Theta_{\At}}_{\C} \Delta$. The proof for $\Gamma = \{ \varnothing \}$ is analogous.

Now pick any $\C \supseteq \B$ and assume $\Gamma \Vdash^{\Theta_{\At}}_{\C} \Delta, A, A$. If $\Gamma = \{\varnothing\}$ and $A = p$ if follows from Lemma \ref{lemma:auxiliarycontractionright} that $\Gamma \Vdash^{\Theta_{\At}}_{\C} \Delta, p$. The proofs for the cases in which $\Gamma = \{\varnothing\}$ and $A$ is non-atomic and in which $\Gamma \neq \{ \varnothing\}$ are analogous, so we only prove the former. Just as in the proof of Lemma \ref{lemma:auxiliarycontractionleft}, the proof is similar in the case of all connectives, so we only write the case where $A = \{C \parr D\}$. By assumption, $\Vdash^{ \Theta_{\At}}_{\C} \Delta, C \parr D, C \parr D$. Now pick any $\D \supseteq \C$ and arbitrary $\Omega_{\At}, \Pi_{\At}, \Phi_{\At}, \Psi_{\At}$ such that $C \Vdash^{\Omega_{\At}}_{\D} \Pi_{\At}$ and $D\Vdash^{\Phi_{\At}}_{\D} \Psi_{\At}$. Then by applying the semantic clause for disjunction we obtain $\Vdash^{ \Theta_{\At}, \Omega_{\At}, \Psi_{\At}}_{\D} \Delta, C \parr D, \Pi_{\At}, \Psi_{\At}$. By applying the semantic clause for disjunction again we obtain $\Vdash^{ \Theta_{\At}, \Omega_{\At}, \Psi_{\At}, \Omega_{\At}, \Psi_{\At}}_{\D} \Delta, \Pi_{\At}, \Psi_{\At} , \Pi_{\At}, \Psi_{\At}$. By applying Lemmas \ref{lemma:auxiliarycontractionleft} and \ref{lemma:auxiliarycontractionright} a sufficient number of times we obtain $\Vdash^{ \Theta_{\At}, \Omega_{\At}, \Psi_{\At}}_{\D} \Delta,  \Pi_{\At}, \Psi_{\At}$, so one last application of the semantic clause for disjunction yields $\Vdash^{\Theta_{\At}}_{\C} \Delta, C \parr D$. The only difference in the proof of the case where $\Gamma$ is non-empty is that we let $\Gamma = \{B_{1}, \ldots, B_{n}\}$ and the extension $\D \supseteq \C$ is also required to have $\Vdash^{\Upsilon^{m}_{\At}}_{\D} B_{m}, \Lambda^{m}_{\At}$ for all $1 \leq m \leq n$.
\end{proof}

\begin{lemma}\label{lemma:auxiliaryweakeningleft} 

      For any $\B$, if $(\Vdash^{ \Gamma_{\At}}_{\mathcal{C}} \Delta_{\At}$ implies  $\Vdash^{p, \Gamma_{\At}}_{\C} \Delta_{\At})$ holds for all $\C \supseteq \B$, $\Gamma_{\At}, \Delta_{\At}$, $p$ then $(\Vdash^{ \Gamma_{\At}}_{\C} \Delta$ implies  $\Vdash^{p, \Gamma_{\At}}_{\C} \Delta)$ holds for all $\C \supseteq \B$, $\Gamma_{\At}, \Delta$, $p$.
    
\end{lemma}
\begin{proof}
  Assume that $(\Vdash^{ \Gamma_{\At}}_{\C} \Delta_{\At}$ implies  $\Vdash^{p, \Gamma_{\At}}_{\C} \Delta_{\At})$ holds for all $\C \supseteq \B, \Gamma_{\At}, \Delta_{\At}$, $p$.  Now pick an arbitrary $\C \supseteq \B$ and assume $\Vdash^{\Gamma_{\At}}_{\C} \Delta$. If $\Delta = \Delta_{\At}$ the assumptions already yield $\Vdash^{p, \Gamma_{\At}}_{\C} \Delta$. If $\Delta$ is not a set of atoms, there is some non-atomic $B$ with $B \in \Delta$. The proof is analogous for all connectives, so we only show the case were $B = C \parr D$ and $\Delta = \Sigma \cup \{ C \parr D\}$. By assumption, $\Vdash^{ \Gamma_{\At}}_{\C} \Sigma, C \parr D$. Now pick any $\D \supseteq \C$ and arbitrary $\Omega_{\At}, \Pi_{\At}, \Phi_{\At}, \Psi_{\At}$ such that $C \Vdash^{\Omega_{\At}}_{\D} \Pi_{\At}$ and $D\Vdash^{\Phi_{\At}}_{\D} \Psi_{\At}$. Then $\Vdash^{ \Gamma_{\At}, \Omega_{\At}, \Psi_{\At}}_{\D} \Sigma, \Pi_{\At}, \Psi_{\At}$. Since the degree of $\Sigma$ is lesser than that of $\Delta$ the induction hypothesis yields $\Vdash^{p, \Gamma_{\At}, \Omega_{\At}, \Psi_{\At}}_{\D} \Sigma, \Pi_{\At}, \Psi_{\At}$, so by the semantic clause for disjunction we conclude $\Vdash^{p,  \Gamma_{\At}}_{\C} \Sigma, C \parr D$.
\end{proof}

\begin{lemma}\label{lemma:auxiliaryweakeningright} 

      For any $\B$, if $(\Vdash^{\Gamma_{\At}}_{\C} \Delta_{\At}$ implies  $\Vdash^{ \Gamma_{\At}}_{\C} \Delta_{\At}, p)$ holds for all $\C \supseteq \B$, $\Gamma_{\At}, \Delta_{\At}$, $p$ then $(\Vdash^{ \Gamma_{\At}}_{\C} \Delta$ implies  $\Vdash^{\Gamma_{\At}}_{\C} \Delta, p)$ holds for all $\C \supseteq \B$, $\Gamma_{\At}, \Delta$, $p$.
    
\end{lemma}
\begin{proof}
Assume that $(\Vdash^{\Gamma_{\At}}_{\C} \Delta_{\At}$ implies  $\Vdash^{ \Gamma_{\At}}_{\C} \Delta_{\At}, p)$ holds for all $\C \supseteq \B$, $\Gamma_{\At}, \Delta_{\At}$, $p$. Now pick an arbitrary $\C \supseteq \B$ and assume $\Vdash^{ \Gamma_{\At}}_{\C} \Delta$.  If $\Delta = \Delta_{\At}$ the assumptions already yield $\Vdash^{ \Gamma_{\At}}_{\C} \Delta,p$. If $\Delta$ is not a set of atoms, there is some non-atomic $B$ with $B \in \Delta$. Once again, the proof is trivial for $\bot$ and $0$ and analogous for all other connectives, so we only show the case were $B = C \parr D$ and $\Delta = \Sigma \cup \{ C \parr D\}$. By assumption, $\Vdash^{ \Gamma_{\At}}_{\C} \Sigma, C \parr D$. Now pick any $\D \supseteq \C$ and arbitrary $\Omega_{\At}, \Pi_{\At}, \Phi_{\At}, \Psi_{\At}$ such that $C \Vdash^{\Omega_{\At}}_{\D} \Pi_{\At}$ and $D\Vdash^{\Phi_{\At}}_{\D} \Psi_{\At}$. Then $\Vdash^{ \Gamma_{\At}, \Omega_{\At}, \Psi_{\At}}_{\D} \Sigma, \Pi_{\At}, \Psi_{\At}$. Since the degree of $\Sigma$ is lesser than that of $\Delta$ the induction hypothesis yields $\Vdash^{ \Gamma_{\At}, \Omega_{\At}, \Psi_{\At}}_{\D} \Sigma, \Pi_{\At}, \Psi_{\At}, p$, so by the semantic clause for disjunction we conclude $\Vdash^{ \Gamma_{\At}}_{\C} \Sigma, C \parr D, p$.

\end{proof}

\begin{proof}[of Theorem~\ref{thm:weakeningleftright}]
($\Rightarrow$): For this direction, we once again prove something stronger: if $(\Gamma\Vdash^{\Theta_{\At}}_{\C} \Delta$ implies $\Gamma, A \Vdash^{\Theta_{\At}}_{\C} \Delta)$  for all $\C \supseteq \B$, $\Gamma_{\At}, \Delta$, $p$ then $(\Vdash^{ \Gamma_{\At}}_{\C} \Delta_{\At}$ implies  $\Vdash^{p, \Gamma_{\At}}_{\C} \Delta_{\At})$ for all  $\C \supseteq \B$, $\Gamma_{\At}, \Delta_{\At}$, $p$, and if $(\Gamma \Vdash^{\Theta_{\At}}_{\C} \Delta$ implies $\Gamma \Vdash^{\Theta_{\At}}_{\C} \Delta, A)$ for all $\C \supseteq \B$, $\Gamma_{\At}, \Delta$, $p$ then $(\Vdash^{ \Gamma_{\At}}_{\C} \Delta_{\At}$ implies  $\Vdash^{ \Gamma_{\At}}_{\C} \Delta_{\At}, p)$ for all $\C \supseteq \B$, $\Gamma_{\At}, \Delta_{\At}$, $p$.

Assume that $(\Gamma \Vdash^{\Theta_{\At}}_{\C} \Delta$ implies $\Gamma, A \Vdash^{\Theta_{\At}}_{\C} \Delta$) holds for all $A$ and all $\C \supseteq \B$, $\Theta_{\At}, \Gamma, \Delta$. Now pick an arbitrary $\C \supseteq \B$ and assume $\Vdash^{\Gamma_{\At}}_{\C} \Delta_{\At}$. The assumption yields $p \Vdash^{\Gamma_{\At}}_{\C} \Delta_{\At}$. From \init\ we have $\vdash_{\C} p \Rightarrow p$ and then $\Vdash^{p}_{\C} p$, so by (Inf) we conclude $\Vdash^{\Gamma_{\At}, p}_{\C} \Delta_{\At}$

Assume that $(\Gamma \Vdash^{\Theta_{\At}}_{\C} \Delta$ implies $\Gamma \Vdash^{\Theta_{\At}}_{\C} \Delta, A)$ holds for all $\C \supseteq \B$, $A, \Theta_{\At}, \Gamma, \Delta$. Now pick an arbitrary $\C \supseteq \B$ and assume $\Vdash^{ \Gamma_{\At}}_{\C} \Delta_{\At}$. We can use the original assumption to conclude $\Vdash^{ \Gamma_{\At}}_{\C} \Delta_{\At}, p$, so the result follows from arbitrariness of $\C \supseteq \B$, $\Gamma_{\At}$, $\Delta_{\At}$, $p$.

($\Leftarrow$): Assume that $((\Vdash^{ \Gamma_{\At}}_{\C} \Delta_{\At}$ implies  $\Vdash^{p, \Gamma_{\At}}_{\C} \Delta_{\At})$ and $(\Vdash^{ \Gamma_{\At}}_{\C} \Delta_{\At}$ implies  $\Vdash^{ \Gamma_{\At}}_{\C} \Delta_{\At}, p)$) holds for all $\C \supseteq \B$, $\Gamma_{\At}, \Delta_{\At}$,$p$.  Pick an arbitrary $\C \supseteq \B$. If $\Gamma \neq \{\varnothing\}$, assume $\Gamma \Vdash^{\Theta_{\At}}_{\C} \Delta$. Let $\Gamma = \{B_{1}, \ldots, B_{n}\}$. Pick any  $\D \supseteq \C$ such that $\Vdash^{\Phi^{m}_{\At}}_{\D} B_{m}, \Psi^{m}_{\At}$ for all $1 \leq m \leq n$ and $\Vdash^{\Upsilon_{\At}}_{\D} A, \Lambda_{\At}$. Since $\Gamma \Vdash^{\Theta_{\At}}_{\C} \Delta$ holds, we can apply (Inf) to conclude $\Vdash^{\Theta_{\At}, \Phi^{1}_{\At}, \ldots, \Phi^{n}_{\At},}_{\D} \Delta,  \Psi^{1}_{\At}, \ldots, \Psi^{n}_{\At}$. By applying Lemmas \ref{lemma:auxiliaryweakeningleft} and \ref{lemma:auxiliaryweakeningright} a sufficient amount of times we conclude $\Vdash^{\Theta_{\At}, \Phi^{1}_{\At}, \ldots, \Phi^{n}_{\At}, \Upsilon_{\At}}_{\D} \Delta,  \Psi^{1}_{\At}, \ldots, \Psi^{n}_{\At}, \Lambda_{\At}$, so by (Inf) we conclude $\Gamma, A \Vdash^{\Theta_{\At}}_{\C} \Delta$. If $\Gamma = \{\varnothing\}$, we once again pick a $\D \supseteq \C$ such that $\Vdash^{\Upsilon_{\At}}_{\D} A, \Lambda_{\At}$, but then from the assumption $ \Vdash^{\Theta_{\At}}_{\C} \Delta$ we get $ \Vdash^{\Theta_{\At}}_{\D} \Delta$ by Monotonicity, to which we apply Lemmas \ref{lemma:auxiliaryweakeningleft} and \ref{lemma:auxiliaryweakeningright} in the same way.

Now pick any $\C \supseteq \B$ and assume $\Gamma \Vdash^{\Theta_{\At}}_{\C} \Delta$. If $\Gamma = \{\varnothing\}$ and $A$ is atomic ir follows from Lemma \ref{lemma:auxiliarycontractionright} that $\Gamma \Vdash^{\Theta_{\At}}_{\C} \Delta, p$. Once again, the proofs for the cases in which $\Gamma = \{\varnothing\}$ and $A$ is non-atomic and in which $\Gamma \neq \{ \varnothing\}$ are analogous, so we only prove the former. The proof is also similar in for all connectives, so we only write the case where $A = \{C \parr D\}$. By assumption, $\Vdash^{ \Theta_{\At}}_{\C} \Delta$.  Now pick any $\D \supseteq \C$ and arbitrary $\Omega_{\At}, \Pi_{\At}, \Phi_{\At}, \Psi_{\At}$ such that $C \Vdash^{\Omega_{\At}}_{\D} \Pi_{\At}$ and $D\Vdash^{\Phi_{\At}}_{\D} \Psi_{\At}$. By Monotonicity, $\Vdash^{ \Theta_{\At}}_{\D} \Delta$. Then we apply Lemmas \ref{lemma:auxiliaryweakeningleft} and \ref{lemma:auxiliaryweakeningright} a sufficient ammount of times to obtain $\Vdash^{ \Theta_{\At}, \Omega_{\At}, \Psi_{\At}}_{\D} \Delta, \Pi_{\At}, \Psi_{\At}$, so by applying the semantic clause for disjunction we conclude $\Vdash^{ \Theta_{\At}}_{\D} \Delta, A \parr B$. The only difference in the proof of the case where $\Gamma$ is non-empty is that we let $\Gamma = \{B_{1}, \ldots, B_{n}\}$ and the extension $\D \supseteq \C$ is also required to have $\Vdash^{\Upsilon^{m}_{\At}}_{\D} B_{m}, \Lambda^{m}_{\At}$ for all $1 \leq m \leq n$.

\end{proof}

\section{Proofs of soundness and completeness in Section~\ref{sec:soundcomp}}\label{app:sound-comp}
\subsection{The additive case}\label{app:add}
\begin{proof}[of Theorem~\ref{lemma:classicalsoundnessadd} -- The additive case]
We shall prove soundness with respect to the the additive sequent system $\ACLp$ for classical logic. In the following, 
$\Gamma = \{A^1,\dots,A^n\}$. We proceed by structural induction, based on the last rule applied in the derivation of $\Gamma\Rightarrow\Delta$.
    \begin{itemize}
        \item [(Ax):]
        The last rule applied in the derivation is (Ax), so in particular suppose our derivation concludes $A\Rightarrow A$. Pick a classical base $\B$. In any extension $\base{C}$, it holds trivially that if $\supp{C}{\varnothing}A$ then $\supp{C}{\varnothing}A$. By (Inf) we thus have that $A\supp{B}{\varnothing}A$, and as $\base{B}$ was arbitrary, we have that $A\supp{}{}A$ in general.
        
        \item [($L\&$):]
        If the last rule in the derivation is $(L\&)$, then the derivation concludes in one of the the following two ways:
        \begin{prooftree}
            \AxiomC{$A,\Gamma\Rightarrow\Delta$}
            \RightLabel{$L \&_{1}$}
            \UnaryInfC{$A\& B,\Gamma\Rightarrow\Delta$}
            \quad
            \DisplayProof
            \quad
             \AxiomC{$B,\Gamma\Rightarrow\Delta$}
             \RightLabel{$L \&_{2}$}
            \UnaryInfC{$A\& B,\Gamma\Rightarrow\Delta$}
        \end{prooftree}

        We only prove the case of $L \&_{1}$, as the other one is analogous. For the induction, we assume that $A, \Gamma\supp{}{}\Delta$. Let $\Gamma = A^1,\dots, A^n$ and let $\base{B}$ be an arbitrary classical base. We want to show that $A \& B, \Gamma\supp{B}{\varnothing}\Delta$. By (Inf), this is equivalent to showing that for all $\base{C}\supseteq\base{B}$, for all $i$ from 1 to $n$:
        \begin{align*}
        \text{if: }& \supp{C}{\Sigma^i_{\At}}A^i,\Pi^i_{\At},\text{ and }\supp{C}{\Upsilon\Atomic}A\& B, \Lambda\Atomic \\
        \text{then: }& \supp{C}{\left(\bigcup\Sigma\Atomic^i\right)\cup\Upsilon\Atomic}\Delta,\Lambda\Atomic,\Pi^1\Atomic,\dots,\Pi^n\Atomic \quad (\star)
        \end{align*}

        So, choose a $\base{C}\supseteq\base{B}$ such that the above conditions hold. By Definition~\ref{def:addandmultsupport}, we have $A,\Gamma\supp{C}{\varnothing}\Delta$. As $\supp{C}{\Sigma^i_{\At}}A^i,\Pi\Atomic^i$, we cut by all elements in $\Gamma$ to give $A\supp{C}{\bigcup\Sigma\Atomic^i}\Delta,\Pi\Atomic^1,\dots,\Pi\Atomic^n$.
        
        By assumption $\supp{C}{\Upsilon\Atomic}A\& B, \Lambda\Atomic$. Hence by the semantic clause for conjunction we have that $\supp{C}{\Upsilon\Atomic\cup\bigcup\Sigma\Atomic^i}\Delta,\Lambda\Atomic,\Pi\Atomic^1,\dots,\Pi\Atomic^n$, which is $(\star)$, as desired.
        
        Hence, we have shown that $A\& B,\Gamma\supp{B}{\varnothing}\Delta$, and as $\base{B}$ was an arbitrary classical base, the result holds without restriction on base or context. Thus, assuming $A,\Gamma\supp{}{}\Delta$ we have shown $A\& B,\Gamma\supp{}{}\Delta$.

        \item[$(R\&)$:]
        If the last rule in the derivation is $(R\&)$, then the derivation concludes in the following way:
        \begin{prooftree}
            \AxiomC{$\Gamma\Rightarrow\Delta,A$}
            \AxiomC{$\Gamma\Rightarrow\Delta,B$}
            \BinaryInfC{$\Gamma,\Rightarrow\Delta,A\& B$}
        \end{prooftree}
        For the induction, we assume that $\Gamma\supp{}{}\Delta,A$ and $\Gamma\supp{}{}\Delta,B$. Let $\Gamma = A^1,\dots, A^n$ and let $\base{B}$ be an arbitrary classical base. By Definition~\ref{def:addandmultsupport}, we have that $\Gamma\supp{B}{\varnothing}\Delta,A$ and $\Gamma\supp{B}{\varnothing}\Delta,B$. Pick an extension $\base{C}\supseteq\base{B}$ such that for all $i$ from 1 to $n$, $\supp{C}{\Sigma^i_{\At}}A^i,\Pi^i_{\At}$, where $\Pi^i\Atomic,\Sigma^i\Atomic$ represent arbitrary atomic contexts. We want to show that in this base:
        \begin{equation}
        \supp{C}{\left(\bigcup\Sigma^i\Atomic\right)}A\& B,\Delta,\Pi\Atomic^1,\dots,\Pi\Atomic^n.
        \end{equation}

        For clarity, let:
        \begin{align*}
            \hat\Gamma&\coloneq\left\{\bigcup^n_{i=1}\Sigma^i\Atomic\right\},\\
            \hat\Delta&\coloneq\{\Delta,\Pi\Atomic^1,\dots,\Pi\Atomic^n\}, \quad (\star)
        \end{align*}
        so that the consequence $(\star)$ above can be rewritten as $\supp{C}{\hat\Gamma}A\& B,\hat\Delta$. By the clause for conjunction, $(\star)$ is thus equivalent to: for all $\base{D}\supseteq\base{C}$, and for arbitrary $\Phi\Atomic,\Psi\Atomic$;
        \[
        \text{if either }A\supp{D}{\Phi\Atomic} \Psi\Atomic\text{ or } B\supp{D}{\Phi\Atomic}\Psi\Atomic\text{, then }\supp{D}{\hat\Gamma\cup\Phi\Atomic}\hat\Delta,\Psi\Atomic.
        \]

        Now, choose a $\base{D}\supseteq\base{C}$ such that $A\supp{D}{\Phi\Atomic}\Psi\Atomic$, for arbitrary atomic contexts $\Phi\Atomic,\Psi\Atomic$. By monotonicity, $\Gamma\supp{D}{\varnothing}\Delta,A$, therefore by Lemma~\ref{lemma:basecut}, $\Gamma\supp{D}{\Phi\Atomic}\Delta,\Psi\Atomic$. As $\supp{C}{\Sigma^i_{\At}}A^i,\Pi^i_{\At}$, then by (Inf) we have that $\supp{D}{\Phi\Atomic\cup(\bigcup\Sigma^i\Atomic)}\Delta,\Psi\Atomic,\Pi^1\Atomic,\dots,\Pi^n\Atomic$, which is $\supp{D}{\hat\Gamma\cup\Phi\Atomic}\hat\Delta,\Psi\Atomic$, as desired. The case when instead $B\supp{D}{\Phi\Atomic}\Psi\Atomic$ follows similarly.

        Hence, we have shown: if $\supp{C}{\Sigma^i_{\At}}A^i,\Pi^i_{\At}$, then $\supp{C}{\hat\Gamma}A\& B,\hat\Delta$. By (Inf), this is equivalent to having shown $\Gamma\supp{B}{\varnothing}A\& B,\Delta$. As $\base{B}$ was arbitrary, we have $\Gamma\supp{}{}A\& B,\Delta$ in general, assuming that $\Gamma\supp{}{}\Delta,A$ and $\Gamma\supp{}{}\Delta,B$. That concludes the induction for this case.

\item[$(L\oplus)$:]
        If the last rule in the derivation is $(L\oplus)$, then the derivation concludes in the following way:
        \begin{prooftree}
            \AxiomC{$A,\Gamma\Rightarrow\Delta$}
            \AxiomC{$B,\Gamma\Rightarrow\Delta$}
            \BinaryInfC{$A\oplus B,\Gamma\Rightarrow\Delta$}
        \end{prooftree}
        For the induction, we assume that $A,\Gamma\supp{}{}\Delta$ and $B,\Gamma\supp{}{}\Delta$. Let $\Gamma = A^1,\dots, A^n$ and let $\base{B}$ be an arbitrary classical base. We want to show that $A\oplus B,\Gamma\supp{B}{\varnothing}\Delta$. By (Inf), this is equivalent to showing that for all $\base{C}\supseteq\base{B}$, for all $i$ from 1 to $n$:
        \begin{align*}
        \text{if: }& \supp{C}{\Sigma^i_{\At}}A^i,\Pi^i_{\At},\text{ and }\supp{C}{\Upsilon\Atomic}A\oplus B, \Lambda\Atomic \\
        \text{then: }& \supp{C}{\left(\bigcup\Sigma\Atomic^i\right)\cup\Upsilon\Atomic}\Delta,\Lambda\Atomic,\Pi^1\Atomic,\dots,\Pi^n\Atomic \quad (\star)
        \end{align*}

        So, choose a $\base{C}\supseteq\base{B}$ such that the above conditions hold. By Definition~\ref{def:addandmultsupport}, we have $A,\Gamma\supp{C}{\varnothing}\Delta$ and $B,\Gamma\supp{C}{\varnothing}\Delta$. As $\supp{C}{\Sigma^i_{\At}}A^i,\Pi\Atomic^i$, we cut by all elements in $\Gamma$ to give $A\supp{C}{\bigcup\Sigma\Atomic^i}\Delta,\Pi\Atomic^1,\dots,\Pi\Atomic^n$. Similarly, we have $B\supp{C}{\bigcup\Sigma\Atomic^i}\Delta,\Pi\Atomic^1,\dots,\Pi\Atomic^n$.
        
        By assumption $\supp{C}{\Upsilon\Atomic}A\oplus B, \Lambda\Atomic$. Hence by the semantic clause for disjunction we have that $\supp{C}{\Upsilon\Atomic\cup\bigcup\Sigma\Atomic^i}\Delta,\Lambda\Atomic,\Pi\Atomic^1,\dots,\Pi\Atomic^n$, which is $(\star)$, as desired.
        
        Hence, we have shown that $A\oplus B,\Gamma\supp{B}{\varnothing}\Delta$, and as $\base{B}$ was an arbitrary classical base, the result holds without restriction on base or context. Thus, assuming $A,\Gamma\supp{}{}\Delta$ and $B,\Gamma\supp{}{}\Delta$, we have shown $A\oplus B,\Gamma\supp{}{}\Delta$.

        \item[$(R\oplus_1)$:]
        If the last rule in the derivation is $(R\oplus_1)$, then the derivation concludes in the following way:
        \begin{prooftree}
            \AxiomC{$\Gamma\Rightarrow\Delta,A$}
            \UnaryInfC{$\Gamma\Rightarrow\Delta,A\oplus B$}
        \end{prooftree}
        For the induction, we assume that $\Gamma\supp{}{}\Delta,A$. Let $\Gamma = A^1,\dots, A^n$ and let $\base{B}$ be an arbitrary classical base. We want to show that $\Gamma\supp{B}{\varnothing}\Delta,A\oplus B$. By (Inf), this is equivalent to showing that for all $\base{C}\supseteq\base{B}$, for all $i$ from 1 to $n$:
        \begin{align*}
        \text{if: }& \supp{C}{\Sigma^i_{\At}}A^i,\Pi^i_{\At},\\
        \text{then: }& \supp{C}{\bigcup\Sigma\Atomic^i}A\oplus B,\Delta,\Pi^1\Atomic,\dots,\Pi^n\Atomic \quad (\star)
        \end{align*}
        For clarity, let:
        \begin{align*}
            \hat\Gamma&\coloneq\left\{\bigcup^n_{i=1}\Sigma^i\Atomic\right\},\\
            \hat\Delta&\coloneq\{\Delta,\Pi\Atomic^1,\dots,\Pi\Atomic^n\},
        \end{align*}
        so that the consequence $(\star)$ above can be rewritten as $\supp{C}{\hat\Gamma}A\oplus B,\hat\Delta$. Let $\base{C}$ be a base with the above condition. By the semantic clause for disjunction, $(\star)$ is thus equivalent to: for all $\base{D}\supseteq\base{C}$, and for arbitrary $\Theta\Atomic,\Lambda\Atomic$;
        \[
        \text{if both }A\supp{D}{\Theta\Atomic}\Lambda\Atomic\text{ and } B\supp{D}{\Theta\Atomic}\Lambda\Atomic\text{, then}\supp{D}{\hat\Gamma\cup\Theta\Atomic}\hat\Delta,\Lambda\Atomic.
        \]
        
        Now, choose a $\base{D}\supseteq\base{C}$ such that $A\supp{D}{\Theta\Atomic}\Lambda\Atomic$ and  $B\supp{D}{\Theta\Atomic}\Lambda\Atomic$. By monotonicity, $\Gamma\supp{D}{\varnothing}\Delta,A$, therefore by Lemma~\ref{lemma:basecut}, $\Gamma\supp{D}{\Theta\Atomic}\Delta,\Lambda\Atomic$. As $\supp{C}{\Sigma^i_{\At}}A^i,\Pi^i_{\At}$, then by (Inf) we have that $\supp{D}{\Theta\Atomic\cup(\bigcup\Sigma^i\Atomic)}\Delta,\Lambda\Atomic,\Pi^1\Atomic,\dots,\Pi^n\Atomic$, which is $\supp{D}{\hat\Gamma\cup\Theta\Atomic}\hat\Delta,\Lambda\Atomic$, as desired.

        Hence, we have shown: if $\supp{C}{\Sigma^i_{\At}}A^i,\Pi^i_{\At}$, then $\supp{C}{\hat\Gamma}A\oplus B,\hat\Delta$. As $\base{B}$ was arbitrary, we have $\Gamma\supp{}{}\Delta,A\oplus B$ in general, assuming that $\Gamma\supp{}{}\Delta,A$. That concludes the induction for this case. The case $(R\oplus_2)$ is similar.

        \item[$(L\rightarrow)$:] 
        If the last rule in the derivation is $(L\rightarrow)$, then the derivation concludes in the following way:
        \begin{prooftree}
            \AxiomC{$\Gamma\Rightarrow\Delta,A$}
            \AxiomC{$B,\Gamma\Rightarrow\Delta$}
            \BinaryInfC{$A\rightarrow B,\Gamma\Rightarrow\Delta$}
        \end{prooftree}
        
        For the induction, we assume that $\Gamma\supp{}{}\Delta,A$ and $B,\Gamma\supp{}{}\Delta$. Let $\Gamma = A^1,\dots, A^n$ and let $\base{B}$ be an arbitrary classical base. We want to show that $A\to B,\Gamma\supp{B}{\varnothing}\Delta$. By (Inf), this is equivalent to showing that for all $\base{C}\supseteq\base{B}$, for all $i$ from 1 to $n$ and all $j$ from 1 to $m$:
        \begin{align*}
        \text{if: }& \supp{C}{\Sigma^i_{\At}}A^i,\Pi\Atomic^i,\text{ and }\supp{C}{\Upsilon\Atomic}A\rightarrow B, \Lambda\Atomic \\
        \text{then: }& \supp{C}{\left(\bigcup\Sigma\Atomic^i\right)\cup\Upsilon\Atomic}\Delta,\Lambda\Atomic,\Pi^1\Atomic,\dots,\Pi^n\Atomic \quad (\star)
        \end{align*}

        So, choose a $\base{C}\supseteq\base{B}$ such that the above conditions hold. By Definition~\ref{def:addandmultsupport}, we have $\Gamma\supp{C}{\varnothing}\Delta,A$ and $B,\Gamma\supp{C}{\varnothing}\Delta$. As $\supp{C}{\Sigma^i_{\At}}A^i,\Pi\Atomic^i$, we cut by all elements in $\Gamma$ to give $\supp{C}{\bigcup\Sigma\Atomic^i}A,\Delta,\Pi\Atomic^1,\dots,\Pi\Atomic^n$. Similarly, we have $B\supp{C}{\bigcup\Sigma\Atomic^i}\Delta,\Pi\Atomic^1,\dots,\Pi\Atomic^n$.
        
        By assumption $\supp{C}{\Upsilon\Atomic}A\rightarrow B, \Lambda\Atomic$. Hence by the semantic clause for implication we have that $\supp{C}{\Upsilon\Atomic\cup\bigcup\Sigma\Atomic^i}\Delta,\Lambda\Atomic,\Pi\Atomic^1,\dots,\Pi\Atomic^n$, which is $(\star)$, as desired.

        \item[$(R\rightarrow_1)$:]
        If the last rule in the derivation is $(R\rightarrow_1)$, then the derivation concludes in the following way:
        \begin{prooftree}
            \AxiomC{$A,\Gamma\Rightarrow\Delta$}
            \UnaryInfC{$\Gamma\Rightarrow\Delta,A\rightarrow B$}
        \end{prooftree}
        For the induction, we assume that $A,\Gamma\supp{}{}\Delta$. Let $\Gamma = A^1,\dots, A^n$ and let $\base{B}$ be an arbitrary classical base. We want to show that $\Gamma\supp{B}{\varnothing}\Delta,A\rightarrow B$. By (Inf), this is equivalent to showing that for all $\base{C}\supseteq\base{B}$, for all $i$ from 1 to $n$:
        \begin{align*}
        \text{if: }& \supp{C}{\Sigma^i_{\At}}A^i,\Pi^i_{\At},\\
        \text{then: }& \supp{C}{\left(\bigcup\Sigma\Atomic^i\right)}A\rightarrow B,\Delta,\Pi^1\Atomic,\dots,\Pi^n\Atomic
        \end{align*}
        For clarity, let:
        \begin{align*}
            \hat\Gamma&\coloneq\left\{\bigcup^n_{i=1}\Sigma^i\Atomic\right\},\\
            \hat\Delta&\coloneq\{\Delta,\Pi\Atomic^1,\dots,\Pi\Atomic^n\},
        \end{align*}
        so that the consequence above can be rewritten as $\supp{C}{\hat\Gamma}A\rightarrow B,\hat\Delta$. Let $\base{C}$ be a base with the above condition. By the semantic clause for implication, this is equivalent to: for all $\base{D}\supseteq\base{C}$, and for arbitrary $\Theta\Atomic,\Phi\Atomic,$;
        \[
        \text{if both }\supp{D}{\Theta\Atomic}A,\Lambda\Atomic\text{ and } B\supp{D}{\Theta\Atomic}\Lambda\Atomic\text{, then}\supp{D}{\hat\Gamma\cup\Theta\Atomic}\hat\Delta,\Lambda\Atomic \quad (\star)
        \]
        
        Now, choose a $\base{D}\supseteq\base{C}$ such that $\supp{D}{\Theta\Atomic}A,\Lambda\Atomic$ and  $B\supp{D}{\Theta\Atomic}\Lambda\Atomic$. By Definition~\ref{def:addandmultsupport}, $A,\Gamma\supp{D}{\varnothing}\Delta$, therefore by Lemma~\ref{lemma:basecut} of $A$, $\Gamma\supp{D}{\Theta\Atomic}\Delta,\Lambda\Atomic$. As $\supp{C}{\Sigma^i_{\At}}A^i,\Pi^i_{\At}$, then by (Inf) we have that $\supp{D}{\Theta\Atomic\cup(\bigcup\Sigma^i\Atomic)}\Delta,\Lambda\Atomic,\Pi^1\Atomic,\dots,\Pi^n\Atomic$, which is $(\star)$, as desired.

        Hence, we have shown: if $\supp{C}{\Sigma^i_{\At}}A^i,\Pi^i_{\At}$, then $\supp{C}{\hat\Gamma}A\rightarrow B,\hat\Delta$. As $\base{B}$ was arbitrary, we have $\Gamma\supp{}{}\Delta,A\rightarrow B$ in general, assuming that $A,\Gamma\supp{}{}\Delta$. That concludes the induction for this case.

        \item[$(R\rightarrow_2)$]
        If the last rule in the derivation is $(R\rightarrow_2)$, then the derivation concludes in the following way:
        \begin{prooftree}
            \AxiomC{$\Gamma\Rightarrow\Delta,B$}
            \UnaryInfC{$\Gamma\Rightarrow\Delta,A\rightarrow B$}
        \end{prooftree}
        For the induction, we assume that $\Gamma\supp{}{}\Delta,B$. Let $\Gamma = A^1,\dots, A^n$ and let $\base{B}$ be an arbitrary classical base. We want to show that $\Gamma\supp{B}{\varnothing}\Delta,A\rightarrow B$. By (Inf), this is equivalent to showing that for all $\base{C}\supseteq\base{B}$, for all $i$ from 1 to $n$:
        \begin{align*}
        \text{if: }& \supp{C}{\Sigma^i_{\At}}A^i,\Pi^i_{\At},\\
        \text{then: }& \supp{C}{\left(\bigcup\Sigma\Atomic^i\right)}A\rightarrow B,\Delta,\Pi^1\Atomic,\dots,\Pi^n\Atomic
        \end{align*}
        For clarity, let:
        \begin{align*}
            \hat\Gamma&\coloneq\left\{\bigcup^n_{i=1}\Sigma^i\Atomic\right\},\\
            \hat\Delta&\coloneq\{\Delta,\Pi\Atomic^1,\dots,\Pi\Atomic^n\},
        \end{align*}
        so that the consequence above can be rewritten as $\supp{C}{\hat\Gamma}A\rightarrow B,\hat\Delta$. Let $\base{C}$ be a base with the above condition. By the semantic clause for implication, this is equivalent to: for all $\base{D}\supseteq\base{C}$, and for arbitrary $\Theta\Atomic,\Phi\Atomic,$;
        \[
        \text{if both }\supp{D}{\Theta\Atomic}A,\Lambda\Atomic\text{ and } B\supp{D}{\Theta\Atomic}\Lambda\Atomic\text{, then}\supp{D}{\hat\Gamma\cup\Theta\Atomic}\hat\Delta,\Lambda\Atomic \quad (\star)
        \]
        
        Now, choose a $\base{D}\supseteq\base{C}$ such that $\supp{D}{\Theta\Atomic}A,\Lambda\Atomic$ and  $B\supp{D}{\Theta\Atomic}\Lambda\Atomic$. By Definition~\ref{def:addandmultsupport}, $\Gamma\supp{D}{\varnothing}\Delta,B$, therefore by Lemma~\ref{lemma:basecut} of $B$, $\Gamma\supp{D}{\Theta\Atomic}\Delta,\Lambda\Atomic$. As $\supp{C}{\Sigma^i_{\At}}A^i,\Pi^i_{\At}$, then by (Inf) we have that $\supp{D}{\Theta\Atomic\cup(\bigcup\Sigma^i\Atomic)}\Delta,\Lambda\Atomic,\Pi^1\Atomic,\dots,\Pi^n\Atomic$, which is $(\star)$, as desired.

        Hence, we have shown: if $\supp{C}{\Sigma^i_{\At}}A^i,\Pi^i_{\At}$, then $\supp{C}{\hat\Gamma}A\rightarrow B,\hat\Delta$. As $\base{B}$ was arbitrary, we have $\Gamma\supp{}{}\Delta,A\rightarrow B$ in general, assuming that $\Gamma\supp{}{}\Delta,B$. That concludes the induction for this case.

          \item[$(R 0)$] The derivation has the following shape:

            \begin{prooftree}
            \AxiomC{}
            \UnaryInfC{$\Gamma, 0 \Rightarrow \Delta$}
        \end{prooftree}

  Let $\Gamma = \{A^{1}, \ldots, A^{n}\}$.  By picking an arbitrary classical base $\B$ and an arbitrary $\C \supseteq \B$ with $\Vdash_{\C}^{\Sigma^{i}_{\At}} A^{i}, \Pi_{\At}$ for all $1 \leq i \leq m$ and $\Vdash_{\C}^{\Phi_{\At}} 0, \Lambda_{\At}$. By $0$ we obtain $\Vdash_{\C}^{\Phi_{\At} \cup \Theta_{\At}}  \Lambda_{\At}, \Omega_{\At}$ for all $\Theta_{\At}$ and $\Omega_{\At}$. Using the same reasoning employed to prove Lemmas \ref{lemma:auxiliaryweakeningleft} and \ref{lemma:auxiliaryweakeningright} we conclude $\Vdash^{\bigcup \Sigma^{i}_{\At} \cup \Phi_{\At}}_{\C} \Delta, \Pi^{1}_{\At}, \ldots \Pi_{\At} \ldots, \Lambda_{\At}$, hence by (Inf) and arbitrariness of $\B$ and $\C$ we conclude $\Gamma, 0 \Vdash \Delta$.
  
  We obtain $\Vdash_{\C}^{\bigcup \Sigma^{i}_{\At}} \Delta, \Pi^{1}_{\At}, \ldots, \Pi^{n}_{\At}$. Then we have  $\Vdash_{\C}^{\bigcup \Sigma^{i}_{\At}} \Delta, \bot, \Pi^{1}_{\At}, \ldots, \Pi^{n}_{\At}$ by $(\bot)$, hence by (Inf) and arbitrariness of $\Sigma^{i}_{\At}$, $\Pi^{i}_{\At}$, $\B$ and $\C$ we have $\Gamma \Vdash \Delta$.
  \end{itemize}
    \end{proof}

\begin{proof}[of Lemma~\ref{mainlemmacompletenessadd} -- The additive case]
If the complexity of of $\Sigma \cup \Theta_{\At}$ is $0$ then $\Sigma$ is a set of atoms, so the result holds by the semantic clause for atoms. If the complexity of $\Sigma \cup \Theta_{\At}$ is greater than $0$, then $\Sigma = \Pi \cup C$ for some $C$ with shape $A \& B$, $A \oplus B$ or $A \to B$. 
    \begin{itemize}
         \item[($C = A \& B$):] ($\Rightarrow$): Assume $\Vdash^{\Omega_{\At}}_{\B} A \& B, \Pi, \Theta_{\At}$. Pick any $\C \supseteq \B$ and $\Phi_{\At}, \Psi_{\At}$ such that $\Vdash^{\Phi^{\At}}_{\C} A, \Psi_{\At}$. The induction hypothesis yields $\vdash_{\C} \Phi_{\At} \Rightarrow p^{A}, \Psi_{\At}$ and then $\Vdash^{\Phi^{\At}}_{\C} p^{A}, \Psi_{\At}$, so we obtain $A \Vdash^{\varnothing}_{\B} p^{A}$ by (Inf). Analogous reasoning yields $B \Vdash^{\varnothing}_{\B} p^{B}$. Since we have $\Vdash^{\Omega_{\At}}_{\B} A \& B, \Pi, \Theta_{\At}$ and both $A \Vdash^{\varnothing}_{\B} p^{A}$ and $B \Vdash^{\varnothing}_{\B} p^{B}$  by the semantic clause for additive conjunction we have both $\Vdash^{\Omega_{\At}}_{\B} p^{A}, \Pi, \Theta_{\At}$ and $\Vdash^{\Omega_{\At}}_{\B} p^{B}, \Pi, \Theta_{\At}$. The induction hypothesis yields $\vdash_{\B} \Omega_{\At} \Rightarrow p^{A}, \Pi_{\At}, \Theta_{\At}$ and $\vdash_{\B} \Omega_{\At} \Rightarrow p^{B}, \Pi_{\At}, \Theta_{\At}$, so we can construct the following deduction:

\begin{prooftree}
    \AxiomC{\vdots}
    \noLine
    \UnaryInfC{$\Omega_{\At} \Rightarrow p^{A}, \Pi_{\At}, \Theta_{\At}$}
        \AxiomC{\vdots}
    \noLine
    \UnaryInfC{$\Omega_{\At} \Rightarrow p^{B}, \Pi_{\At}, \Theta_{\At}$}
    \BinaryInfC{$ \Omega_{\At} \Rightarrow p^{A \& B}, \Pi_{\At}, \Theta_{\At}$}
\end{prooftree}

Which already shows the desired $\vdash_{\B} \Omega_{\At} \Rightarrow p^{A \& B}, \Pi_{\At}, \Theta_{\At}$
  
  ($\Leftarrow$): Assume $\vdash_{\B} \Omega_{\At} \Rightarrow p^{A \& B},  \Pi_{\At}, \Theta_{\At}$. We start by constructing the following deduction:
      \begin{prooftree}
      \AxiomC{$\vdots$}
      \noLine
      \UnaryInfC{$\Omega_{\At} \Rightarrow p^{A \& B},  \Pi_{\At}, \Theta_{\At}$}
    \AxiomC{}
      \RightLabel{\tiny{\init}}
    \UnaryInfC{$p^{A} \Rightarrow p^{A}$}
         \RightLabel{\tiny{$L \&$}}
    \UnaryInfC{$p^{A \& B} \Rightarrow p^{A}$}
    \RightLabel{\tiny{\Cutc}}
    \BinaryInfC{$\Omega_{\At} \Rightarrow p^{A},  \Pi_{\At}, \Theta_{\At}$}
      \end{prooftree}
So we have $\vdash_{\B} \Omega_{\At} \Rightarrow p^{A},  \Pi_{\At}, \Theta_{\At}$, hence $\Vdash^{\Omega_{\At}}_{\B} A,  \Pi, \Theta_{\At}$ by the induction hypothesis. A similar argument shows $\Vdash^{\Omega_{\At}}_{\B} B,  \Pi, \Theta_{\At}$. Now pick any $\C \supseteq \B$ and $\Upsilon_{\At}, \Lambda_{\At}$ such that $A \Vdash^{\Upsilon_{\At}}_{C} \Lambda_{\At}$. Since $\Vdash^{\Omega_{\At}}_{\B} A,  \Pi, \Theta_{\At}$, by monotonicity we have $\Vdash^{\Omega_{\At}}_{\C} A,  \Pi, \Theta_{\At}$, so by (Inf) we get  $\Vdash^{\Omega_{\At}, \Upsilon_{\At}}_{\C}  \Pi, \Theta_{\At}, \Lambda_{\At}$. A similar argument shows that, for any $\C \supseteq \B$, if $B \Vdash^{\Upsilon_{\At}}_{C} \Lambda_{\At}$ then $\Vdash^{\Omega_{\At}, \Upsilon_{\At}}_{\C}  \Pi, \Theta_{\At}, \Lambda_{\At}$, so by the semantic clause for additive conjunction we conclude $\Vdash^{\Omega_{\At}}_{\B} A \& B, \Pi, \Theta_{\At}$.

  \item[($C = A \oplus B$):] ($\Rightarrow$): Assume $\Vdash^{\Omega_{\At}}_{\B} A \oplus B, \Pi, \Theta_{\At}$. Pick any $\C \supseteq \B$ and $\Phi_{\At}, \Psi_{\At}$ such that $\Vdash^{\Phi^{\At}}_{\C} A, \Psi_{\At}$. The induction hypothesis yields $\vdash_{\C} \Phi_{\At} \Rightarrow p^{A},  \Psi_{\At}$, so we can construct the following deduction:

  \begin{prooftree}
      \AxiomC{\vdots}
      \noLine
      \UnaryInfC{$ \Phi_{\At} \Rightarrow p^{A},  \Psi_{\At}$}
      \RightLabel{\tiny{$L \oplus$}}
      \UnaryInfC{$\Phi_{\At} \Rightarrow p^{A \oplus B},  \Psi_{\At}$}
  \end{prooftree}

 So $\vdash_{C} \Phi_{\At} \Rightarrow p^{A \oplus B},  \Psi_{\At}$ and then $\Vdash^{\Phi_{\At}}_{C} p^{A \oplus B},  \Psi_{\At}$, whence we conclude $A \Vdash^{\varnothing}_{\B} p^{A \oplus B}$ by (Inf). A similar argument yields $B \Vdash^{\varnothing}_{\B} p^{A \oplus B}$. Since $\Vdash^{\Omega_{\At}}_{\B} A \oplus B, \Pi, \Theta_{\At}$ and both $A \Vdash^{\varnothing}_{\B} p^{A \oplus B}$ and $B \Vdash^{\varnothing}_{\B} p^{A \oplus B}$, by the semantic clause for additive disjunction we conclude $\Vdash^{\Omega_{\At}}_{\B} p^{A \oplus B}, \Pi, \Theta_{\At}$, hence we can use the induction hypothesis to obtain $\vdash_{\B} \Omega_{\At} \Rightarrow  p^{A \oplus B}, \Pi_{\At}, \Theta_{\At}$

  ($\Leftarrow$): Assume $\vdash_{\B} \Omega_{\At} \Rightarrow p^{A \oplus B},  \Pi_{\At}, \Theta_{\At}$. Pick any $\C \supseteq \B$ and $\Phi_{\At}, \Psi_{\At}$ such that $A \Vdash^{\Phi_{\At}}_{\C} \Psi_{\At}$ and  $B \Vdash^{\Phi_{\At}}_{\C} \Psi_{\At}$. Since $\vdash_{\C} p^{A} \Rightarrow p^{A}$ and $\vdash_{\C} p^{B} \Rightarrow p^{B}$ both hold by Axiom, the induction hypothesis yields $\Vdash^{p^{A}}_{\C} A$ and $\Vdash^{p^{B}}_{\C} B$, hence we can apply (Inf) to obtain $\Vdash^{\Phi_{\At}, p^{A}}_{\C} \Psi_{\At}$ and $\Vdash^{\Phi_{\At}, p^{B}}_{\C} \Psi_{\At}$. Those yield $\vdash_{\C} p^{A}, \Phi_{\At} \Rightarrow \Psi_{\At}$ and $\vdash_{\C} p^{B}, \Phi_{\At} \Rightarrow \Psi_{\At}$. Since  $\vdash_{\B} \Omega_{\At} \Rightarrow p^{A \oplus B},  \Pi_{\At}, \Theta_{\At}$ we also have  $\vdash_{\C} \Omega_{\At} \Rightarrow p^{A \oplus B},  \Pi_{\At}, \Theta_{\At}$, so we can construct the following deduction:
  
\begin{prooftree}
     \AxiomC{$\vdots$}
      \noLine
      \UnaryInfC{$\Omega_{\At} \Rightarrow p^{A \oplus B},  \Pi_{\At}, \Theta_{\At}$}
      \AxiomC{\vdots}
      \noLine
      \UnaryInfC{$p^{A}, \Phi_{\At} \Rightarrow \Psi_{\At}$}
        \AxiomC{\vdots}
      \noLine
      \UnaryInfC{$p^{B}, \Phi_{\At} \Rightarrow \Psi_{\At}$}
         \RightLabel{\tiny{$R \oplus$}}
      \BinaryInfC{$p^{A \oplus B}, \Phi_{\At} \Rightarrow \Psi_{\At}$}
      \RightLabel{\tiny{\Cutc}}
      \BinaryInfC{$\Omega_{\At}, \Phi_{\At} \Rightarrow \Psi_{\At}, \Pi_{\At}, \Theta_{\At}$}
\end{prooftree}

So $\vdash_{\C} \Omega_{\At}, \Phi_{\At} \Rightarrow \Psi_{\At}, \Pi_{\At}, \Theta_{\At}$ and then $\Vdash^{\Omega_{\At}, \Phi_{\At}}_{\C} \Psi_{\At}, \Pi_{\At}, \Theta_{\At}$, whence the induction hypothesis yields $\Vdash^{\Omega_{\At}, \Phi_{\At}}_{\C} \Psi_{\At}, \Pi, \Theta_{\At}$, so by the semantic clause for additive disjunction we conclude $\Vdash^{\Omega_{\At}}_{\B} A \oplus B, \Pi, \Theta_{\At}$.

\item[($C = 0$):] ($\Rightarrow$): Assume $\Vdash^{\Omega_{\At}}_{\B} 0, \Pi, \Theta_{\At}$. Then $\Vdash^{\Omega_{\At} \cup \Phi_{\At}}_{\B} \Pi, \Theta_{\At}, \Psi_{\At}$ for all $\Phi_{\At}$ and $\Psi_{\At}$. By putting $\Phi_{\At} = \varnothing$ and $\Psi_{\At} = \{p^{0}\}$ and applying the induction hypothesis we obtain $\Omega_{\At} \Rightarrow p^{0}, \Pi_{\At}, \Theta_{\At}$.

  ($\Leftarrow$): Assume $\vdash_{\C} \Omega_{\At} \Rightarrow p^{0}, \Pi_{\At}, \Theta_{\At}$, and pick arbitrary $\Phi_{\At}$ and $\Psi_{\At}$. We can construct the following deduction:

 \begin{prooftree}
      \AxiomC{\vdots}
      \noLine
      \UnaryInfC{$ \Omega_{\At} \Rightarrow p^{0}, \Pi_{\At}, \Theta_{\At}$}
      \AxiomC{}
        \RightLabel{\tiny{$L 0$}}
      \UnaryInfC{$\Phi_{\At},  p^{0} \Rightarrow \Psi_{\At}$}
      \RightLabel{\tiny{\Cutc}}
      \BinaryInfC{$ \Omega_{\At}, \Phi_{\At} \Rightarrow  \Pi_{\At}, \Theta_{\At}, \Psi_{\At}$}
  \end{prooftree}

The induction hypothesis yields $\Vdash^{\Omega_{\At} \cup \Phi_{\At}}_{\C} \Pi, \Theta_{\At}, \Psi_{\At}$ for arbitrary $\Phi_{\At}$ and $\Psi_{\At}$, hence $\Vdash^{\Omega_{\At}}_{\C} 0, \Pi, \Theta_{\At}$.
    \end{itemize}
\end{proof}

\subsection{The multiplicative case}\label{app:mult}

\begin{proof}[of Theorem~\ref{lemma:classicalsoundnessadd} -- The multiplicative case]]
  We shall prove soundness with respect to the the sequent system $\MCLp$ for classical logic. The proof proceeds by structural induction on the last rule applied in the derivation of $\Gamma \Rightarrow \Delta$. In what follows, let $\Gamma = \{A^1, \dots, A^n\}$.
   \begin{itemize}
        \item[(Ax):]
        The last rule applied in the derivation is (Ax), so in particular suppose our derivation concludes $A\Rightarrow A$. Pick a classical base $\B$. In any extension $\base{C}$, it holds trivially that if $\supp{C}{}A$ then $\supp{C}{}A$. By (Inf) we thus have that $A\supp{B}{}A$, and as $\base{B}$ was arbitrary, we have that $A\supp{}{}A$ in general.

        \item[($L\otimes$):]
        If the last rule in the derivation is $(L\otimes)$, then the derivation concludes in the following way:
        \begin{prooftree}
            \AxiomC{$A,B,\Gamma\Rightarrow\Delta$}
            \UnaryInfC{$A\otimes B,\Gamma\Rightarrow\Delta$}
        \end{prooftree}
        For the induction, we assume that $A,B, \Gamma\supp{}{}\Delta$. Let $\Gamma = A^1,\dots, A^n$ and let $\base{B}$ be an arbitrary classical base. We want to show that $A \otimes B, \Gamma\supp{B}{\varnothing}\Delta$. By (Inf), this is equivalent to showing that for all $\base{C}\supseteq\base{B}$, for all $i$ from 1 to $n$:
        \begin{align*}
        \text{if: }& \supp{C}{\Sigma^i_{\At}}A^i,\Pi^i_{\At},\text{ and }\supp{C}{\Upsilon\Atomic}A\otimes B, \Lambda\Atomic \\
        \text{then: }& \supp{C}{\left(\bigcup\Sigma\Atomic^i\right)\cup\Upsilon\Atomic}\Delta,\Lambda\Atomic,\Pi^1\Atomic,\dots,\Pi^n\Atomic \quad (\star)
        \end{align*}

        So, choose a $\base{C}\supseteq\base{B}$ such that the above conditions hold. By Definition~\ref{def:addandmultsupport}, we have $A,B,\Gamma\supp{C}{\varnothing}\Delta$. As $\supp{C}{\Sigma^i_{\At}}A^i,\Pi\Atomic^i$, we cut by all elements in $\Gamma$ to give $A,B\supp{C}{\bigcup\Sigma\Atomic^i}\Delta,\Pi\Atomic^1,\dots,\Pi\Atomic^n$.
        
        By assumption $\supp{C}{\Upsilon\Atomic}A\otimes B, \Lambda\Atomic$. Hence by the semantic clause for conjunction we have that $\supp{C}{\Upsilon\Atomic\cup\bigcup\Sigma\Atomic^i}\Delta,\Lambda\Atomic,\Pi\Atomic^1,\dots,\Pi\Atomic^n$, which is $(\star)$, as desired.
        
        Hence, we have shown that $A\otimes B,\Gamma\supp{B}{\varnothing}\Delta$, and as $\base{B}$ was an arbitrary classical base, the result holds without restriction on base or context. Thus, assuming $A,B,\Gamma\supp{}{}\Delta$ we have shown $A\otimes B,\Gamma\supp{}{}\Delta$.
        
        \item[$(R\otimes)$:]
        If the last rule in the derivation is $(R\otimes)$, then the derivation concludes in the following way:
        \begin{prooftree}
            \AxiomC{$\Gamma\Rightarrow\Delta,A$}
            \AxiomC{$\Gamma'\Rightarrow\Delta',B$}
            \BinaryInfC{$\Gamma,\Gamma'\Rightarrow\Delta,\Delta',A\otimes B$}
        \end{prooftree}
        For the induction, we assume that $\Gamma\supp{}{}\Delta,A$ and $\Gamma'\supp{}{}\Delta',B$. Let $\Gamma = A^1,\dots, A^n$, $\Gamma'=B^1,\dots, B^m$, and let $\base{B}$ be an arbitrary classical base. By Definition~\ref{def:addandmultsupport}, we have that $\Gamma\supp{B}{\varnothing}\Delta,A$ and $\Gamma'\supp{B}{\varnothing}\Delta',B$. Pick an extension $\base{C}\supseteq\base{B}$ such that for all $i$ from 1 to $n$ and all $j$ from 1 to $m$, $\supp{C}{\Sigma^i_{\At}}A^i,\Pi^i_{\At}$ and $\supp{C}{\Theta^j_{\At}}B^j,\Omega^j_{\At}$, where $\Pi^i\Atomic,\Sigma^i\Atomic,\Omega^i\Atomic,\Theta^i\Atomic$ all represent arbitrary atomic contexts.  We want to show that in this base:
        \begin{equation}
        \supp{C}{\left(\bigcup\Sigma^i\Atomic\right)\cup\left(\bigcup\Theta^j\Atomic\right)}A\otimes B,\Delta,\Delta',\Pi\Atomic^1,\dots,\Pi\Atomic^n,\Omega\Atomic^1,\dots,\Omega\Atomic^m. \quad (\star)   
        \end{equation}

        For clarity, let:
        \begin{align*}
            \hat\Gamma&\coloneq\left\{ \left(\bigcup^n_{i=1}\Sigma^i\Atomic\right)\cup\left(\bigcup^m_{j=1}\Theta^j\Atomic\right)\right\},\\
            \hat\Delta&\coloneq\{\Delta,\Delta',\Pi\Atomic^1,\dots,\Pi\Atomic^n,\Omega\Atomic^1,\dots,\Omega\Atomic^m\},
        \end{align*}
        so that the consequence $(\star)$ above can be rewritten as $\supp{C}{\hat\Gamma}A\otimes B,\hat\Delta$. By the clause for conjunction, $(\star)$ is thus equivalent to: for all $\base{D}\supseteq\base{C}$, and for arbitrary $\Phi\Atomic,\Psi\Atomic$;
        \[
        \text{if }A,B\supp{D}{\Phi\Atomic} \Psi\Atomic\text{ then }\supp{D}{\hat\Gamma\cup\Phi\Atomic}\hat\Delta,\Psi\Atomic.
        \]

        Now, choose a $\base{D}\supseteq\base{C}$ such that $A,B\supp{D}{\Phi\Atomic}\Psi\Atomic$, for arbitrary atomic contexts $\Phi\Atomic,\Psi\Atomic$. By monotonicity, $\Gamma\supp{D}{\varnothing}\Delta,A$ and $\Gamma'\supp{D}{\varnothing}\Delta',B$, therefore by Cut of $A$ and $B$ (Lemma~\ref{lemma:basecut}), $\Gamma,\Gamma'\supp{D}{\Phi\Atomic}\Delta,\Delta',\Psi\Atomic$. As $\supp{C}{\Sigma^i_{\At}}A^i,\Pi^i_{\At}$ and $\supp{C}{\Theta^j_{\At}}B^j,\Omega^j_{\At}$, then by (Inf) we have that $\supp{D}{\Phi\Atomic\cup(\bigcup\Sigma^i\Atomic)\cup(\bigcup\Theta^i\Atomic)}\Delta,\Psi\Atomic,\Pi^1\Atomic,\dots,\Pi^n\Atomic,\Omega^1\Atomic,\dots,\Omega^m\Atomic$, which is $\supp{D}{\hat\Gamma\cup\Phi\Atomic}\hat\Delta,\Psi\Atomic$, as desired.

        Hence, we have shown: if $\supp{C}{\Sigma^i_{\At}}A^i,\Pi^i_{\At}$ and $\supp{C}{\Theta^j_{\At}}B^j,\Omega^j_{\At}$, then $\supp{C}{\hat\Gamma}A\otimes B,\hat\Delta$. By (Inf), this is equivalent to having shown $\Gamma,\Gamma'\supp{B}{\varnothing}A\otimes B,\Delta,\Delta'$. As $\base{B}$ was arbitrary, we have $\Gamma,\Gamma'\supp{}{}A\otimes B,\Delta,\Delta'$ in general, assuming that $\Gamma\supp{}{}\Delta,A$ and $\Gamma'\supp{}{}\Delta',B$. That concludes the induction for this case.

        \item[$(L\parr)$:]
        If the last rule in the derivation is $(L\parr)$, then the derivation concludes in the following way:
        \begin{prooftree}
            \AxiomC{$A,\Gamma\Rightarrow\Delta$}
            \AxiomC{$B,\Gamma'\Rightarrow\Delta'$}
            \BinaryInfC{$A\parr B,\Gamma,\Gamma'\Rightarrow\Delta,\Delta'$}
        \end{prooftree}
        For the induction, we assume that $A,\Gamma\supp{}{}\Delta$ and $B,\Gamma'\supp{}{}\Delta'$. Let $\Gamma = A^1,\dots, A^n$, $\Gamma'=B^1,\dots, B^m$, and let $\base{B}$ be an arbitrary classical base. We want to show that $A\parr B,\Gamma,\Gamma'\supp{B}{\varnothing}\Delta,\Delta'$. By (Inf), this is equivalent to showing that for all $\base{C}\supseteq\base{B}$, for all $i$ from 1 to $n$ and all $j$ from 1 to $m$:
        \begin{align*}
        \text{if: }& \supp{C}{\Sigma^i_{\At}}A^i,\Pi^i_{\At},\supp{C}{\Theta^j_{\At}}B^j,\Omega^j_{\At},\text{ and }\supp{C}{\Upsilon\Atomic}A\parr B, \Lambda\Atomic \\
        \text{then: }& \supp{C}{\left(\bigcup\Sigma\Atomic^i\right)\cup\left(\bigcup\Theta\Atomic^j\right)\cup\Upsilon\Atomic}\Delta,\Delta',\Lambda\Atomic,\Pi^1\Atomic,\dots,\Pi^n\Atomic,\Omega^1\Atomic,\dots,\Omega^m\Atomic \quad (\star)
        \end{align*}

        So, choose a $\base{C}\supseteq\base{B}$ such that the above conditions hold. By Definition~\ref{def:addandmultsupport}, we have $A,\Gamma\supp{C}{\varnothing}\Delta$ and $B,\Gamma'\supp{C}{\varnothing}\Delta'$. As $\supp{C}{\Sigma^i_{\At}}A^i,\Pi\Atomic^i$, we cut by all elements in $\Gamma$ to give $A\supp{C}{\bigcup\Sigma\Atomic^i}\Delta,\Pi\Atomic^1,\dots,\Pi\Atomic^n$. Similarly, we cut by all elements in $\Gamma'$ to give $B\supp{C}{\bigcup\Theta\Atomic^j}\Delta',\Omega\Atomic^1,\dots,\Omega\Atomic^m$.
        
        By assumption $\supp{C}{\Upsilon\Atomic}A\parr B, \Lambda\Atomic$. Hence by the semantic clause for disjunction we have that $\supp{C}{\left(\bigcup\Sigma\Atomic^i\right)\cup\left(\bigcup\Theta\Atomic^j\right)\cup\Upsilon\Atomic}\Delta,\Delta',\Lambda\Atomic,\Pi^1\Atomic,\dots,\Pi^n\Atomic,\Omega^1\Atomic,\dots,\Omega^m\Atomic$, which is $(\star)$, as desired.
        
        Hence, we have shown that $A\parr B,\Gamma,\Gamma'\supp{B}{\varnothing}\Delta,\Delta'$, and as $\base{B}$ was an arbitrary classical base, the result holds without restriction on base or context. Thus, assuming that $A,\Gamma\supp{}{}\Delta$ and $B,\Gamma'\supp{}{}\Delta'$, we have shown $A\parr B,\Gamma,\Gamma'\supp{}{}\Delta,\Delta'$.

        \item[$(R\parr)$:]
        If the last rule in the derivation is $(R\parr)$, then the derivation concludes in the following way:
        \begin{prooftree}
            \AxiomC{$\Gamma\Rightarrow\Delta,A,B$}
            \UnaryInfC{$\Gamma\Rightarrow\Delta,A\parr B$}
        \end{prooftree}
        For the induction, we assume that $\Gamma\supp{}{}\Delta,A,B$. Let $\Gamma = A^1,\dots, A^n$ and let $\base{B}$ be an arbitrary classical base. We want to show that $\Gamma\supp{B}{\varnothing}\Delta,A\parr B$. By (Inf), this is equivalent to showing that for all $\base{C}\supseteq\base{B}$, for all $i$ from 1 to $n$:
        \begin{align*}
        \text{if: }& \supp{C}{\Sigma^i_{\At}}A^i,\Pi^i_{\At},\\
        \text{then: }& \supp{C}{\bigcup\Sigma\Atomic^i}A\parr B,\Delta,\Pi^1\Atomic,\dots,\Pi^n\Atomic \quad (\star)
        \end{align*}
        For clarity, let:
        \begin{align*}
            \hat\Gamma&\coloneq\left\{\bigcup^n_{i=1}\Sigma^i\Atomic\right\},\\
            \hat\Delta&\coloneq\{\Delta,\Pi\Atomic^1,\dots,\Pi\Atomic^n\},
        \end{align*}
        so that the consequence $(\star)$ above can be rewritten as $\supp{C}{\hat\Gamma}A\parr B,\hat\Delta$. Let $\base{C}$ be a base with the above condition. By the semantic clause for disjunction, $(\star)$ is thus equivalent to: for all $\base{D}\supseteq\base{C}$, and for arbitrary $\Theta\Atomic,\Phi\Atomic,\Psi\Atomic,\Omega\Atomic$;
        \[
        \text{if both }A\supp{D}{\Theta\Atomic}\Lambda\Atomic\text{ and } B\supp{D}{\Phi\Atomic}\Psi\Atomic\text{, then}\supp{D}{\hat\Gamma\cup\Theta\Atomic\cup\Phi\Atomic}\hat\Delta,\Lambda\Atomic\Psi\Atomic.
        \]
        
        Now, choose a $\base{D}\supseteq\base{C}$ such that $A\supp{D}{\Theta\Atomic}\Lambda\Atomic$ and  $B\supp{D}{\Phi\Atomic}\Psi\Atomic$. By monotonicity, $\Gamma\supp{D}{\varnothing}\Delta,A,B$, therefore by applying Cut twice, $\Gamma\supp{D}{\Theta\Atomic\cup\Phi\Atomic}\Delta,\Lambda\Atomic,\Psi\Atomic$. As $\supp{C}{\Sigma^i_{\At}}A^i,\Pi^i_{\At}$, then by (Inf) we have that $\supp{D}{\Theta\Atomic\cup\Phi\Atomic\cup(\bigcup\Sigma^i\Atomic)}\Delta,\Lambda\Atomic,\Psi\Atomic,\Pi^1\Atomic,\dots,\Pi^n\Atomic$, that is, $\supp{D}{\hat\Gamma\cup\Theta\Atomic\cup\Phi\Atomic}\hat\Delta,\Lambda\Atomic,\Psi\Atomic$, as desired.

        Hence, we have shown: if $\supp{C}{\Sigma^i_{\At}}A^i,\Pi^i_{\At}$, then $\supp{C}{\hat\Gamma}A\parr B,\hat\Delta$. As $\base{B}$ was arbitrary, we have $\Gamma\supp{}{}\Delta,A\parr B$ in general, assuming that $\Gamma\supp{}{}\Delta,A,B$. That concludes the induction for this case.

        \item[$(L\multimap)$:] 
        If the last rule in the derivation is $(L\multimap)$, then the derivation concludes in the following way:
        \begin{prooftree}
            \AxiomC{$\Gamma\Rightarrow\Delta,A$}
            \AxiomC{$B,\Gamma'\Rightarrow\Delta'$}
            \BinaryInfC{$A\multimap B,\Gamma,\Gamma'\Rightarrow\Delta,\Delta'$}
        \end{prooftree}
        For the induction, we assume that $\Gamma\supp{}{}\Delta,A$ and $B,\Gamma'\supp{}{}\Delta'$. Let $\Gamma = A^1,\dots, A^n$, $\Gamma'=B^1,\dots, B^m$, and let $\base{B}$ be an arbitrary classical base. We want to show that $A\multimap B,\Gamma,\Gamma'\supp{B}{\varnothing}\Delta,\Delta'$. By (Inf), this is equivalent to showing that for all $\base{C}\supseteq\base{B}$, for all $i$ from 1 to $n$ and all $j$ from 1 to $m$:
        \begin{align*}
        \text{if: }& \supp{C}{\Sigma^i_{\At}}A^i,\Pi^i_{\At},\supp{C}{\Theta^j_{\At}}B^j,\Omega^j_{\At},\text{ and }\supp{C}{\Upsilon\Atomic}A\multimap B, \Lambda\Atomic \\
        \text{then: }& \supp{C}{\left(\bigcup\Sigma\Atomic^i\right)\cup\left(\bigcup\Theta\Atomic^j\right)\cup\Upsilon\Atomic}\Delta,\Delta',\Lambda\Atomic,\Pi^1\Atomic,\dots,\Pi^n\Atomic,\Omega^1\Atomic,\dots,\Omega^m\Atomic \quad (\star)
        \end{align*}

        So, choose a $\base{C}\supseteq\base{B}$ such that the above conditions hold. By Definition~\ref{def:addandmultsupport}, we have $\Gamma\supp{C}{\varnothing}\Delta,A$ and $B,\Gamma'\supp{C}{\varnothing}\Delta'$. As $\supp{C}{\Sigma^i_{\At}}A^i,\Pi\Atomic^i$, we cut by all elements in $\Gamma$ to give $\supp{C}{\bigcup\Sigma\Atomic^i}\Delta,A,\Pi\Atomic^1,\dots,\Pi\Atomic^n$. Similarly, we cut by all elements in $\Gamma'$ to give $B\supp{C}{\bigcup\Theta\Atomic^j}\Delta',\Omega\Atomic^1,\dots,\Omega\Atomic^m$.
        
        By assumption $\supp{C}{\Upsilon\Atomic}A\multimap B, \Lambda\Atomic$. Hence by the semantic clause for implication we have that $\supp{C}{\left(\bigcup\Sigma\Atomic^i\right)\cup\left(\bigcup\Theta\Atomic^j\right)\cup\Upsilon\Atomic}\Delta,\Delta',\Lambda\Atomic,\Pi^1\Atomic,\dots,\Pi^n\Atomic,\Omega^1\Atomic,\dots,\Omega^m\Atomic$, which is $(\star)$, as desired.
        
        Hence, we have shown that $A\multimap B,\Gamma,\Gamma'\supp{B}{\varnothing}\Delta,\Delta'$, and as $\base{B}$ was an arbitrary classical base, the result holds without restriction on base or context. Thus, assuming that $\Gamma\supp{}{}\Delta,A$ and $B,\Gamma'\supp{}{}\Delta'$, we have shown $A\multimap B,\Gamma,\Gamma'\supp{}{}\Delta,\Delta'$.
        
        \item[$(R\multimap)$:]
        If the last rule in the derivation is $(R\multimap)$, then the derivation concludes in the following way:
        \begin{prooftree}
            \AxiomC{$A,\Gamma\Rightarrow\Delta,B$}
            \UnaryInfC{$\Gamma\Rightarrow\Delta,A\multimap B$}
        \end{prooftree}
        For the induction, we assume that $A,\Gamma\supp{}{}\Delta,B$. Let $\Gamma = A^1,\dots, A^n$ and let $\base{B}$ be an arbitrary classical base. We want to show that $\Gamma\supp{B}{\varnothing}\Delta,A\multimap B$. By (Inf), this is equivalent to showing that for all $\base{C}\supseteq\base{B}$, for all $i$ from 1 to $n$:
        \begin{align*}
        \text{if: }& \supp{C}{\Sigma^i_{\At}}A^i,\Pi^i_{\At},\\
        \text{then: }& \supp{C}{\bigcup\Sigma\Atomic^i}A\multimap B,\Delta,\Pi^1\Atomic,\dots,\Pi^n\Atomic \quad (\star)
        \end{align*}
        For clarity, let:
        \begin{align*}
            \hat\Gamma&\coloneq\left\{\bigcup^n_{i=1}\Sigma^i\Atomic\right\},\\
            \hat\Delta&\coloneq\{\Delta,\Pi\Atomic^1,\dots,\Pi\Atomic^n\},
        \end{align*}
        so that the consequence $(\star)$ above can be rewritten as $\supp{C}{\hat\Gamma}A\multimap B,\hat\Delta$. Let $\base{C}$ be a base with the above condition. By the semantic clause for implication, $(\star)$ is thus equivalent to: for all $\base{D}\supseteq\base{C}$, and for arbitrary $\Theta\Atomic,\Phi\Atomic,\Psi\Atomic,\Omega\Atomic$;
        \[
        \text{if both }\supp{D}{\Theta\Atomic}A,\Lambda\Atomic\text{ and } B\supp{D}{\Phi\Atomic}\Psi\Atomic\text{, then}\supp{D}{\hat\Gamma\cup\Theta\Atomic\cup\Phi\Atomic}\hat\Delta,\Lambda\Atomic\Psi\Atomic.
        \]
        
        Now, choose a $\base{D}\supseteq\base{C}$ such that $\supp{D}{\Theta\Atomic}A,\Lambda\Atomic$ and  $B\supp{D}{\Phi\Atomic}\Psi\Atomic$. By Definition~\ref{def:addandmultsupport}, $A,\Gamma\supp{D}{\varnothing}\Delta,B$, therefore by Cut (Lemma~\ref{lemma:basecut}) of $A$, $\Gamma\supp{D}{\Theta\Atomic}\Delta,\Lambda\Atomic,B$. Similarly, Cut by $B$ to yield $\Gamma\supp{D}{\Theta\Atomic\cup\Phi\Atomic}\Delta,\Lambda\Atomic,\Psi\Atomic$. As $\supp{C}{\Sigma^i_{\At}}A^i,\Pi^i_{\At}$, then by (Inf) we have that $\supp{D}{\Theta\Atomic\cup\Phi\Atomic\cup(\bigcup\Sigma^i\Atomic)}\Delta,\Lambda\Atomic,\Psi\Atomic,\Pi^1\Atomic,\dots,\Pi^n\Atomic$, that is, $\supp{D}{\hat\Gamma\cup\Theta\Atomic\cup\Phi\Atomic}\hat\Delta,\Lambda\Atomic,\Psi\Atomic$, as desired.

        Hence, we have shown: if $\supp{C}{\Sigma^i_{\At}}A^i,\Pi^i_{\At}$, then $\supp{C}{\hat\Gamma}A\multimap B,\hat\Delta$. As $\base{B}$ was arbitrary, we have $\Gamma\supp{}{}\Delta,A\multimap B$ in general, assuming that $A,\Gamma\supp{}{}\Delta,B$. That concludes the induction for this case.

  \item[$(L \bot)$] The derivation has the following shape:

            \begin{prooftree}
            \AxiomC{}
            \UnaryInfC{$\bot \Rightarrow \varnothing$}
        \end{prooftree}

    By picking an arbitrary classical base $\B$ and an arbitrary $\C \supseteq \B$ with $\Vdash_{\C}^{\Sigma_{\At}} \bot, \Pi_{\At}$ we obtain $\Vdash_{\C}^{\Sigma_{\At}} \Pi_{\At}$ by $(\bot)$, whence by (Inf) and arbitrariness of $\B$ and $\C$ we have $\bot \Vdash \varnothing$.

        \item[$(R \bot)$] The derivation has the following shape:

            \begin{prooftree}
            \AxiomC{$\Gamma \Rightarrow \Delta$}
            \UnaryInfC{$\Gamma \Rightarrow \Delta, \bot$}
        \end{prooftree}

  Let $\Gamma = \{A^{1}, \ldots, A^{n}\}$.  By picking an arbitrary classical base $\B$ and an arbitrary $\C \supseteq \B$ with $\Vdash_{\C}^{\Sigma^{i}_{\At}} A^{i}, \Pi_{\At}$ for all $1 \leq i \leq m$ we obtain $\Vdash_{\C}^{\bigcup \Sigma^{i}_{\At}} \Delta, \Pi^{1}_{\At}, \ldots, \Pi^{n}_{\At}$. Then we have  $\Vdash_{\C}^{\bigcup \Sigma^{i}_{\At}} \Delta, \bot, \Pi^{1}_{\At}, \ldots, \Pi^{n}_{\At}$ by $(\bot)$, hence by (Inf) and arbitrariness of $\Sigma^{i}_{\At}$, $\Pi^{i}_{\At}$, $\B$ and $\C$ we have $\Gamma \Vdash \Delta, \bot$.

    \end{itemize}
\end{proof}

\begin{proof}[of Lemma~\ref{mainlemmacompletenessadd} -- The multiplicative case]
If the complexity of of $\Sigma \cup \Theta_{\At}$ is $0$ then $\Sigma$ is a set of atoms, so the result holds by the semantic clause for atoms. If the complexity of $\Sigma \cup \Theta_{\At}$ is greater than $0$, then $\Sigma = \Pi \cup C$ for some $C$ with shape $A \otimes B$, $A \parr B$ or $A \multimap B$. 

\begin{itemize}
    \item[($C = A \multimap B$):]
($\Rightarrow$): Assume $\Vdash^{\Omega_{\At}}_{\B} A \multimap B, \Pi, \Theta_{\At}$. By \init\ we have $\vdash_{\B} p^{A} \Rightarrow p^{A}$, so $\Vdash^{p^{A}}_{\B} A$ via the induction hypothesis. Pick any $\C \supseteq \B$ and $\Phi_{\At}, \Psi_{\At}$ such that $\Vdash^{\Phi^{\At}}_{\C} B, \Psi_{\At}$. The induction hypothesis yields $\vdash_{\C} \Phi_{\At} \Rightarrow p^{B}, \Psi_{\At}$  and so $\Vdash^{\Phi^{\At}}_{\C} p^{B}, \Psi_{\At}$, hence $B \Vdash^{\varnothing}_{\B} p^{B}$ by (Inf). Since $\Vdash^{\Omega_{\At}}_{\B} A \multimap B, \Pi, \Theta_{\At}$ the semantic clause for implication yields $\Vdash^{p^{A}, \Omega_{\At}}_{\B} \Pi, \Theta_{\At}, p^{B}$.  The induction hypothesis yields $\vdash_{\B} p^{A}, \Omega_{\At} \Rightarrow \Pi_{\At}, \Theta_{\At}, p^{B}$, so by applying the appropriate instance of $R \multimap$ we conclude $\vdash_{\B} \Omega_{\At} \Rightarrow p^{A \multimap B}, \Pi_{\At}, \Theta_{\At}$.

($\Leftarrow$): Assume $\vdash_{\B} \Omega_{\At} \Rightarrow p^{A \multimap B}, \Pi_{\At}, \Theta_{\At}$. Pick any $\C \supseteq \B$ and $\Phi_{\At}, \Psi_{\At}, \Upsilon_{\At}, \Lambda_{\At}$ such that $\Vdash^{\Phi_{\At}}_{\C} A, \Psi_{\At}$ and $B \Vdash_{\C}^{\Upsilon_{\At}}, \Lambda_{\At}$. By induction hypothesis, $\Vdash^{\Phi_{\At}}_{\C} p^{A}, \Psi_{\At}$, so $\vdash_{\C} \Phi_{\At} \Rightarrow p^{A}, \Psi_{\At}$. By \init\ we have $\vdash_{\C} p ^{B} \Rightarrow p^{B}$, so $\Vdash^{p^{B}}_{\C} p^{B}$, hence by induction hypothesis $\Vdash^{p^{B}}_{\C} B$. Since $\Vdash^{p^{B}}_{\C} B$ and $B \Vdash_{\C}^{\Upsilon_{\At}}, \Lambda_{\At}$ we obtain $\Vdash^{p^{B}, \Upsilon_{\At}}_{\C} \Lambda_{\At}$, so $\vdash_{\C} p^{B}, \Upsilon_{\At} \Rightarrow \Lambda_{\At}$. Since $\vdash_{\B} \Omega_{\At} \Rightarrow p^{A \multimap B}, \Pi_{\At}, \Theta_{\At}$ and $\mathcal{C} \supseteq \mathcal{B}$ we also have $\vdash_{\C} \Omega_{\At} \Rightarrow p^{A \multimap B}, \Pi_{\At}, \Theta_{\At}$, hence we can build the following deduction:

\begin{prooftree}
      \AxiomC{\vdots}
        \noLine
    \UnaryInfC{$\Phi_{\At }\Rightarrow p^{A}, \Psi_{\At}$}
    \AxiomC{$\Xi$}
    \noLine
    \UnaryInfC{$\Omega_{\At}, p^{A} \Rightarrow p^{B},  \Pi_{\At}, \Theta_{\At}$}
    \RightLabel{\scriptsize{\Cutc}}
\BinaryInfC{$\Phi_{\At }, \Omega_{\At}\Rightarrow p^{B},  \Pi_{\At},\Theta_{\At},  \Psi_{\At}$}
\AxiomC{\vdots}
        \noLine
    \UnaryInfC{$p^{B}, \Upsilon_{\At }\Rightarrow \Lambda_{\At}$}
    \RightLabel{\scriptsize{\Cutc}}
    \BinaryInfC{$\Phi_{\At }, \Upsilon_{\At}, \Omega_{\At} \Rightarrow  \Pi_{\At}, \Theta_{\At}, \Psi_{\At},  \Lambda_{\At}$}
\end{prooftree}
where $\Xi$ is the derivation
\begin{prooftree}
    \AxiomC{\vdots}
        \noLine
    \UnaryInfC{$\Omega_{\At}\Rightarrow p^{A \multimap B},  \Pi_{\At}, \Theta_{\At}$}
    \AxiomC{}
    \RightLabel{\scriptsize{\init}}
    \UnaryInfC{$p^{A} \Rightarrow p^{A}$}
    \AxiomC{}
    \RightLabel{\scriptsize{\init}}
    \UnaryInfC{$p^{B} \Rightarrow p^{B}$}
       \RightLabel{\scriptsize{$L \multimap$}}
    \BinaryInfC{$p^{A}, p^{A \multimap B} \Rightarrow p^{B}$}
    \RightLabel{\scriptsize{\Cutc}}
    \BinaryInfC{$\Omega_{\At}, p^{A} \Rightarrow p^{B},  \Pi_{\At}, \Theta_{\At}$}
\end{prooftree}

So $\vdash_{\C} \Phi_{\At }, \Upsilon_{\At}, \Omega_{\At} \Rightarrow  \Pi_{\At}, \Theta_{\At}, \Psi_{\At},  \Lambda_{\At}$ and $\Vdash^{\Phi_{\At}, \Lambda_{\At} , \Omega_{\At}}_{\C} \Pi_{\At}, \Theta_{\At}, \Psi_{\At}, \Lambda_{\At}$. By induction hypothesis, $\Vdash^{\Phi_{\At}, \Lambda_{\At} , \Omega_{\At}}_{\C} \Pi, \Theta_{\At}, \Psi_{\At}, \Lambda_{\At}$  Since $\C$ was an arbitrary extension of $\B$ with $\Vdash^{\Phi_{\At}}_{\C} A, \Psi_{\At}$ and $B \Vdash_{\C}^{\Upsilon_{\At}}, \Lambda_{\At}$, we conclude $\Vdash_{\B}^{\Omega_{\At}} A \multimap B, \Pi, \Theta_{\At}$ by the semantic clause for implication.

  \item[($C = A \otimes B$):] ($\Rightarrow$): Assume $\Vdash^{\Omega_{\At}}_{\B} A \otimes B, \Pi, \Theta_{\At}$. Pick any $\C \supseteq \B$ and $\Phi_{\At}, \Psi_{\At}, \Upsilon_{\At}, \Lambda_{\At}$ such that $\Vdash^{\Phi^{\At}}_{\C} A, \Psi_{\At}$ and $\Vdash^{\Upsilon^{\At}}_{\C} B, \Lambda_{\At}$. The induction hypothesis yields $\vdash_{\C} \Phi_{\At} \Rightarrow p^{A}, \Psi_{\At}$ and $\vdash_{\C} \Upsilon_{\At} \Rightarrow p^{B}, \Lambda_{\At}$. By applying the appropriate instance of $R \otimes$ we obtain $\vdash_{\C} \Phi_{\At}, \Upsilon_{\At} \Rightarrow \Psi_{\At}, \Lambda_{\At}, p^{A \otimes B}$, so $\Vdash^{\Phi_{\At}, \Upsilon_{\At}}_{\C} p^{A \otimes B}, \Psi_{\At}, \Lambda_{\At}$. We conclude $A, B \Vdash^{\varnothing}_{\B} p^{A \otimes B}$, hence since $\Vdash^{\Omega_{\At}}_{\B} A \otimes B, \Pi, \Theta_{\At}$ the semantic clause for conjunction yields $\Vdash^{\Omega_{\At}}_{\B} p^{A \otimes B}, \Pi, \Theta_{\At}$, so by the induction hypothesis we conclude $\vdash_{\B} \Omega_{\At} \Rightarrow  p^{A \otimes B}, \Pi_{\At}, \Theta_{\At}$.

  ($\Leftarrow$): Assume $\vdash_{\B} \Omega_{\At} \Rightarrow p^{A \otimes B},  \Pi_{\At}, \Theta_{\At}$. Pick any $\C \supseteq \B$ and any $\Phi_{\At}, \Psi_{\At}$ such that $A, B \Vdash^{\Phi_{\At}}_{\C} \Psi_{\At}$. By \init\ we have $\vdash_{\C} p^{A} \Rightarrow p^{A}$ and $\vdash_{\C} p^{B} \Rightarrow p^{B}$, so the induction hypothesis yields $\Vdash^{p^{A}}_{\C} A$ and $\Vdash^{p^{B}}_{\C} B$, hence by applying (Inf) we get $\Vdash^{p^{A}, p^{B},  \Phi_{\At}}_{\C} \Psi_{\At}$ and so $\vdash_{\C} p^{A}, p^{B},  \Phi_{\At} \Rightarrow \Psi_{\At}$. The we build the following deduction:

      \begin{prooftree}
      \AxiomC{$\vdots$}
      \noLine
      \UnaryInfC{$\Omega_{\At} \Rightarrow p^{A \otimes B},  \Pi_{\At}, \Theta_{\At}$}
    \AxiomC{$\vdots$}
     \noLine
   \UnaryInfC{$p^{A}, p^{B}, \Phi_{\At} \Rightarrow \Psi_{\At}$}
    \UnaryInfC{$p^{A \otimes B}, \Phi_{\At} \Rightarrow \Psi_{\At}$}
    \RightLabel{\scriptsize{\Cutc}}
    \BinaryInfC{$\Omega_{\At}, \Phi_{\At}\Rightarrow   \Pi_{\At}, \Theta_{\At}, \Psi_{\At}$}
      \end{prooftree}

So $\vdash_{\C} \Omega_{\At}, \Phi_{\At}\Rightarrow   \Pi_{\At}, \Theta_{\At}, \Psi_{\At}$, hence by the induction hypothesis $\Vdash^{\Omega_{\At}, \Phi_{\At}}_{\C}   \Pi, \Theta_{\At}, \Psi_{\At}$, so the semantic clause for conjunction yields $\Vdash^{\Omega_{\At}}_{\B} A \otimes B, \Pi, \Theta_{\At}$.

  \item[($C = A \parr B$):] ($\Rightarrow$): Assume $\Vdash^{\Omega_{\At}}_{\B} A \parr B, \Pi, \Theta_{\At}$. Pick any $\C \supseteq \B$ and $\Phi_{\At}, \Psi_{\At}$ such that $\Vdash^{\Phi^{\At}}_{\C} A, \Psi_{\At}$. The induction hypothesis yields $\vdash_{\C} \Phi_{\At} \Rightarrow p^{A},  \Psi_{\At}$ and then $\Vdash^{\Phi_{\At}}_{\C} p^{A}$, so we conclude $A \Vdash^{\varnothing}_{\B} p^{A}$. A similar argument yields $B \Vdash^{\varnothing}_{\B} p^{B}$, hence since $\Vdash^{\Omega_{\At}}_{\B} A \parr B, \Pi, \Theta_{\At}$ we conclude $\Vdash^{\Omega_{\At}}_{\B} p^{A}, p^{B}, \Pi,\Theta_{\At}$. The induction hypothesis yields $\vdash_{\B} \Omega_{\At} \Rightarrow p^{A}, p^{B}, \Pi,\Theta_{\At}$, so by applying the appropriate instance of $R \parr$ we conclude $\vdash_{\B} \Omega_{\At} \Rightarrow p^{A \parr B},  \Pi,\Theta_{\At}$

  ($\Leftarrow$): Assume $\vdash_{\B} \Omega_{\At} \Rightarrow p^{A \parr B},  \Pi_{\At}, \Theta_{\At}$. Pick any $\C \supseteq \B$ and $\Phi_{\At}, \Psi_{\At}, \Upsilon_{\At}, \Lambda_{\At}$ such that $A \Vdash^{\Phi_{\At}}_{\C} \Psi_{\At}$ and  $B \Vdash^{\Upsilon_{\At}}_{\C} \Lambda_{\At}$. Since $\vdash_{\C} p^{A} \Rightarrow p^{A}$ and $\vdash_{\C} p^{B} \Rightarrow p^{B}$ both hold by Axiom, the induction hypothesis yields $\Vdash^{p^{A}}_{\C} A$ and $\Vdash^{p^{B}}_{\C} B$, hence we can apply (Inf) to obtain $\Vdash^{\Phi_{\At}, p^{A}}_{\C} \Psi_{\At}$ and $\Vdash^{\Upsilon_{\At}, p^{B}}_{\C} \Lambda_{\At}$. Those yield $\vdash_{\C} p^{A}, \Phi_{\At} \Rightarrow \Psi_{\At}$ and $\vdash_{\C} p^{B}, \Upsilon_{\At} \Rightarrow \Lambda_{\At}$. Since  $\vdash_{\B} \Omega_{\At} \Rightarrow p^{A \parr B},  \Pi_{\At}, \Theta_{\At}$ we also have  $\vdash_{\C} \Omega_{\At} \Rightarrow p^{A \parr B},  \Pi_{\At}, \Theta_{\At}$, so we can construct the following deduction:

\begin{prooftree}
     \AxiomC{$\vdots$}
      \noLine
      \UnaryInfC{$\Omega_{\At} \Rightarrow p^{A \parr B},  \Pi_{\At}, \Theta_{\At}$}
      \AxiomC{\vdots}
      \noLine
      \UnaryInfC{$p^{A}, \Phi_{\At} \Rightarrow \Psi_{\At}$}
        \AxiomC{\vdots}
      \noLine
      \UnaryInfC{$p^{B}, \Upsilon_{\At} \Rightarrow \Lambda_{\At}$}
         \RightLabel{\tiny{$R \parr$}}
      \BinaryInfC{$p^{A \parr B}, \Phi_{\At} \Rightarrow \Psi_{\At}$}
      \RightLabel{\tiny{\Cutc}}
      \BinaryInfC{$\Omega_{\At}, \Phi_{\At}, \Upsilon_{\At} \Rightarrow \Psi_{\At}, \Pi_{\At}, \Theta_{\At}, \Lambda_{\At}$}
\end{prooftree}

So $\vdash_{\C} \Omega_{\At}, \Phi_{\At} , \Upsilon_{\At}  \Rightarrow \Psi_{\At}, \Pi_{\At}, \Theta_{\At}, \Lambda_{\At}$ and then $\Vdash^{\Omega_{\At}, \Phi_{\At}, \Upsilon_{\At}}_{\C} \Psi_{\At}, \Pi_{\At}, \Theta_{\At}, \Lambda_{\At}$, whence the induction hypothesis yields $\Vdash^{\Omega_{\At}, \Phi_{\At}, \Upsilon_{\At}}_{\C} \Psi_{\At}, \Pi, \Theta_{\At}, \Lambda_{\At}$, so by the semantic clause for disjunction we conclude $\Vdash^{\Omega_{\At}}_{\B} A \parr B, \Pi, \Theta_{\At}$.

\item[($C = \bot$):] ($\Rightarrow$): Assume $\Vdash^{\Omega_{\At}}_{\B} \bot, \Pi, \Theta_{\At}$. Then $\Vdash^{\Omega_{\At}}_{\B} \Pi, \Theta_{\At}$. The induction hypothesis yields $\vdash_{\C} \Omega_{\At} \Rightarrow \Pi_{\At},  \Theta_{\At}$, so we can construct the following deduction:

  \begin{prooftree}
      \AxiomC{\vdots}
      \noLine
      \UnaryInfC{$ \Omega_{\At} \Rightarrow \Pi_{\At}, \Theta_{\At}$}
      \RightLabel{\tiny{$R \bot$}}
      \UnaryInfC{$ \Omega_{\At} \Rightarrow p^{\bot}, \Pi_{\At}, \Theta_{\At}$}
  \end{prooftree}

Which shows the the desired result.

  ($\Leftarrow$): Assume $\vdash_{\C} \Omega_{\At} \Rightarrow p^{\bot}, \Pi_{\At}, \Theta_{\At}$. We can construct the following deduction:

 \begin{prooftree}
      \AxiomC{\vdots}
      \noLine
      \UnaryInfC{$ \Omega_{\At} \Rightarrow p^{\bot}, \Pi_{\At}, \Theta_{\At}$}
      \AxiomC{}
        \RightLabel{\tiny{$L \bot$}}
      \UnaryInfC{$p^{\bot} \Rightarrow \varnothing$}
      \RightLabel{\tiny{\Cutc}}
      \BinaryInfC{$ \Omega_{\At} \Rightarrow  \Pi_{\At}, \Theta_{\At}$}
  \end{prooftree}

The induction hypothesis yields $\Vdash^{\Omega_{\At}}_{\C} \Pi, \Theta_{\At}$, hence $\Vdash^{\Omega_{\At}}_{\C} \bot, \Pi, \Theta_{\At}$.
\end{itemize}
\end{proof}

\section{Proof of Lemma~\ref{lemma:intmod}}\label{app:single}
\begin{proof}
We start by proving the case $\Sigma = \varnothing$ by induction on the degree $\Delta$.

If $\Delta$ is empty the result already holds. If $\Delta$ is a non-empty set of atoms $\Delta_{\At}$, then $\vdash_{\mathcal{B}} \Gamma_{\At} \Rightarrow \Delta_{\At}$ by $(\At)$. The definition of $\B$ guarantees that all its rules have at most one atom on their consequents, including the last rule applied on the derivation proving $\vdash_{\mathcal{B}} \Gamma_{\At} \Rightarrow \Delta_{\At}$. Therefore, $\Delta_{\At} = \{ p\}$ for some $p \in \At$. 

For non-empty $\Delta$ with at least one non-atomic formula $C$, the proof is separated into cases. We only prove the result for $A \oplus B$, as the others are analogous.

Let $\Delta = \{A \oplus B\} \cup \Theta$. By assumption, $\Vdash^{\Gamma_{\At}}_{\B} A \oplus B, \Theta$. First, pick an arbitrary $p \in \At$, and let $\C$ be obtained by adding to $\B$ one rule with no premises and with conclusion $\Omega_{\At} \Rightarrow p$ for every $\Omega_{\At}$. Now pick any $\D \supseteq \C$ and arbitrary $\Phi_{\At}, \Psi_{\At}$ such that $\Vdash^{\Phi_{\At}}_{\D} A, \Psi_{\At}$. By the induction hypothesis, $\Psi_{\At}$ must be empty. Since $\vdash_{\D} \Phi_{\At} \Rightarrow p$ holds by the definition of $\D$, $\Vdash^{\Phi_{\At}}_{\D} p$ also holds, hence by arbitrariness of $\Phi_{\At}$ and $\Psi_{\At}$ we obtain $A \Vdash^{\varnothing}_{\C} p$. A similar argument shows $B \Vdash^{\varnothing}_{\C} p$. Since $\Vdash^{\Gamma_{\At}}_{\B} A \oplus B, \Theta$, $A \Vdash^{\varnothing}_{\C} p$ and $B \Vdash^{\varnothing}_{\C} p$ we conclude $\Vdash^{\Gamma_{\At}}_{\C} p, \Theta$, and by the induction hypothesis we conclude that $\Theta$ must be empty, proving the result.

To conclude the proof, let $\Sigma = \{ C^{1}, \ldots, C^{n}\}$. Pick any $\C \supseteq \B$ such that $\Vdash^{\Phi^{i}_{\At}}_{\C} C^{i}, \Psi^{i}_{\At}$ for all $1 \leq n $. By the result for $\Sigma = \varnothing$, $\Psi^{i}$ is empty for all $1 \leq n$. Since $\Sigma \Vdash^{\Gamma_{\At}}_{\B} \Delta$, we conclude $\Vdash^{\bigcup \Psi^{i}_{\At} \cup \Gamma_{\At} }_{\B} \Delta$, whence the result for empty antecedent yields that $\Delta$ has at most one formula.

Since the semantic clauses are still the same, so are the proofs using them. Therefore, if we repeat the proofs of soundness and completeness but now using simulation bases containing at most one formula on its right.

\section{Proofs of Section~\ref{sec:jankov}}\label{app:jankov}
\begin{proof}[of Lemma~\ref{lemma:derivJankov}]
    For the left to right direction, let $\Pi$ be a derivation in $\JK$. Now let $\Pi'$ be the deduction obtained by replacing all occurrences of $\sf{AxJK}$ with conclusion $\Rightarrow \neg p\oplus \neg \neg p$ by occurrences of \init \  with conclusion $\neg p \oplus \neg \neg p \Rightarrow \neg p \oplus \neg \neg p$, as well as adding the formula $\neg p \oplus \neg \neg p$ to the left context of all sequents appearing below it. Since all new formulas to the left of end sequent can be elements of an atomic Jankov set, it follows that this is a deduction in $\Int$ of the desired shape.

    For the converse, pick derivation $\Pi$ in $\Int$ of a sequent $\Theta, \Gamma \Rightarrow \Delta$, where $\Theta$ is an atomic Jankov set. Given an arbitrary formula $\neg p \oplus \neg \neg p$ of $\Theta$, we can produce the following derivation $\Pi'$, where $\Omega = \Theta - \{ \neg p \oplus \neg \neg p \}$:

    \begin{prooftree}
    \AxiomC{}
    \RightLabel{$\sf{AxJK}$}
    \UnaryInfC{$\Rightarrow \neg p \oplus \neg \neg p$}
        \AxiomC{$\Pi$}
        \noLine
        \UnaryInfC{$\neg p \oplus \neg \neg p,   \Omega, \Gamma  \Rightarrow \Delta$}
        \RightLabel{\Cuti}
        \BinaryInfC{$\Omega, \Gamma \Rightarrow \Delta$}
    \end{prooftree}

    Since $\Theta$ is finite, by repeating the procedure finitely many times we obtain the desired derivation showing $\Gamma \Rightarrow A$ in $\JK$.

    \end{proof}

\begin{proof}[of Lemma~\ref{lemma:jankovsound}]

    For the base case of weakly Jankov bases, pick any $\C \supseteq \B$ such that $\neg p \Vdash^{\Gamma_{\At}}_{\C} \Delta_{\At}$ and $\neg \neg p \Vdash^{\Gamma_{\At}}_{\C} \Delta_{\At}$ for arbitrary $\Gamma_{\At}$ and $\Delta_{\At}$. Pick a $\D \supseteq \C^{\neg}$ such that, for arbitrary $\Theta_{\At}$,  $\Sigma_{\At}$ and $\Omega_{\At}$, we have $\Vdash^{\Theta_{\At}}_{\D} p$ and $\bot \Vdash^{\Sigma_{\At}}_{\D}, \Omega_{\At}$. From $\Vdash^{\Theta_{\At}}_{\D} p$ we conclude $\vdash_{\D} \Theta_{\At}\Rightarrow p$. By applying \Cuti \ with $[/p \Rightarrow \varnothing]$ we conclude $\vdash_{\D} \Theta_{\At }\Rightarrow \varnothing$ and $\Vdash^{\Theta_{\At}}_{\D} \varnothing$. By $(\bot)$ we conclude $\Vdash^{\Theta_{\At}}_{\D} \bot$, so by Lemma \ref{lemma:basecut} and $\bot \Vdash^{\Sigma_{\At}}_{\D} \Omega_{\At}$ we conclude $\Vdash^{\Theta_{\At} \cup \Sigma_{\At}}_{\D} \Omega_{\At}$, so  $\Vdash^{\varnothing}_{\C^{\neg}} \neg p$ by $(\multimap)$. Since $\neg p \Vdash^{\Gamma_{\At}}_{\C^{\neg}} \Delta_{\At}$ holds by monotonicity, by Lemma \ref{lemma:basecut} we conclude $\Vdash^{\Gamma_{\At}}_{\C^{\neg} } \Delta_{\At}$ hence $\vdash_{\C^{\neg}} \Gamma_{\At} \Rightarrow \Delta_{\At}$.
    
    Now pick a $\D \supseteq \C^{\neg \neg}$  such that, for arbitrary $\Theta_{\At}$,  $\Sigma_{\At}$ and $\Omega_{\At}$, we have $\Vdash^{\Theta_{\At}}_{\D} \neg p$ and $\bot \Vdash^{\Sigma_{\At}}_{\D}, \Omega_{\At}$. Since $\Vdash^{\{p\}}_{\D} p$ and $\bot \Vdash^{\varnothing}_{\D} \varnothing$ can be proved trivially, from $\Vdash^{\Theta_{\At}}_{\D} \neg p$ we conclude $\Vdash^{\Theta_{\At}, p }_{\D} \varnothing$, hence also $\vdash_{\D} \Theta_{\At}, p \Rightarrow \varnothing$. Since $\D$ contains $[\Theta_{\At}, p \Rightarrow \varnothing / \Theta_{\At} \Rightarrow \varnothing]$, we can obtain $\vdash_{\D} \Theta_{\At} \Rightarrow \varnothing$, then $\Vdash^{\Theta_{\At}}_{\D} \varnothing$ and $\Vdash^{\Theta_{\At}}_{\D} \bot$. Since $\bot \Vdash^{\Sigma_{\At}}_{\D}, \Omega_{\At}$, by Lemma \ref{lemma:basecut} we can conclude $\Vdash^{\Theta_{\At} \cup \Sigma_{\At}}_{\D} \Omega_{\At}$, which yields $\Vdash^{\varnothing}_{\C^{\neg \neg}} \neg \neg p$. Since $\neg \neg p \Vdash^{\Gamma_{\At}}_{\C^{\neg \neg}} \Delta_{\At}$ holds by monotonicity, by Lemma \ref{lemma:basecut} we conclude $\Vdash^{\Gamma_{\At}}_{\C^{\neg \neg} } \Delta_{\At}$ hence $\vdash_{\C^{\neg \neg}} \Gamma_{\At} \Rightarrow \Delta_{\At}$.

    Since $\vdash_{\C^{\neg}} \Gamma_{\At} \Rightarrow \Delta_{\At}$, $\vdash_{\C^{\neg \neg}} \Gamma_{\At} \Rightarrow \Delta_{\At}$ and $\C$ is a weakly Jankov base, we conclude $\vdash_{\C} \Gamma_{\At} \Rightarrow \Delta_{\At}$ hence $\Vdash^{\Gamma_{\At}}_{\C} \Delta_{\At}$. Since $\C$ is an arbitrary extension of $\B$ with $\neg p \Vdash^{\Gamma_{\At}}_{\C} \Delta_{\At}$ and $\neg \neg p \Vdash^{\Gamma_{\At}}_{\C} \Delta_{\At}$, we finally conclude $\Vdash^{\varnothing}_{\B} \neg p \oplus \neg \neg p$.

    Now for strong Jankov bases. Pick any $\C \supseteq \B$ such that $\neg p \Vdash^{\Gamma_{\At}}_{\C} \Delta_{\At}$ and $\neg \neg p \Vdash^{\Gamma_{\At}}_{\C} \Delta_{\At}$ for arbitrary $\Gamma_{\At}$ and $\Delta_{\At}$.  Assume $\vdash_{\B} p \Rightarrow \varnothing$, and pick any $\D \supseteq \C$ with $\Vdash^{\Theta_{\At}}_{\D} p$ and $\bot \Vdash^{\Sigma_{\At}}_{\D} \Omega_{\At}$. From $\vdash_{\B} p \Rightarrow \varnothing$ follows $\vdash_{\D} p \Rightarrow \varnothing$, so using the same reasoning as in the case for $\C^{\neg}$ in weakly Jankov bases we obtain $\Vdash^{\varnothing}_{\C} \neg p$, hence since $\neg p \Vdash^{\Gamma_{\At}}_{\C} \Delta_{\At}$ we conclude $ \Vdash^{\Gamma_{\At}}_{\C} \Delta_{\At}$ by Lemma \ref{lemma:basecut}, whence we already have $\Vdash^{\varnothing}_{\C} \neg p \oplus \neg \neg p$. Now assume that, for all $\D \supseteq \B$ and all $\Theta_{\At}$, if $\vdash_{\D} \Theta_{\At}, p \Rightarrow \varnothing$ then $\vdash_{\D} \Theta_{\At} \Rightarrow \varnothing$. Pick any $\E \supseteq \C$ such that $\Vdash^{\Sigma_{\At}}_{\E} \neg p$ and $\bot \Vdash^{\Omega_{\At}}_{\E} \Lambda_{\At}$. Using the same reasoning as in the case for $\C^{\neg \neg}$ we obtain $\Vdash^{\varnothing}_{\C} \neg \neg p$, so since $\neg \neg p \Vdash^{\Gamma_{\At}}_{\C} \Delta_{\At}$ by Lemma \ref{lemma:basecut} we obtain $\Vdash^{\Gamma_{\At}}_{\C} \Delta_{\At}$, so $\Vdash^{\varnothing}_{\B} \neg p \oplus \neg \neg p$. Since all strongly Jankov bases must satisfy one of those conditions, the result follows.

\end{proof}

\begin{proof}[of Lemma~\ref{lemma:weak}]
  Cumulativeness follows immediately from the construction of $\E$, where all rules are now allowed to have arbitrary extra contexts. Notice that we still have a simulation base for intuitionistic logic, as our definition of $\Int$ \ allows arbitrary contexts in $\init$. To prove that they are also weakly Jankov, pick any atom $p$. Since $\alpha$ is a bijection, there is some $B \in \mathbb{F}(\Gamma^{S}_{\At})$ such that $\alpha(B) = p$, hence $p = p^{B}$. Since $B \in \mathbb{F}(\Gamma^{S}_{\At})$ and $\Theta$ is the set of all Jankov formulas in $\mathbb{F}(\Gamma^{S}_{\At})$, we have $(\neg B \oplus \neg \neg B) \in \Theta$ so there must be some atom such that $\alpha(\neg B \oplus \neg \neg B) = p^{\neg B \oplus \neg \neg B}$, and also $\E$ contains a rule concluding $\Rightarrow p^{\neg B \oplus \neg \neg B}$. Now pick any cumulative $\D \supseteq \E$ and any $q$ such that there is a deduction $\Pi$ showing $\vdash_{\D^{\neg}} \Gamma_{\At} \Rightarrow q$ and a deduction $\Pi'$ showing $\vdash_{\D^{\neg \neg}} \Gamma_{\At} \Rightarrow \Delta_{\At}$. If either $\Pi$ or $\Pi'$ do not use the new rules added to $\D$ then they already are deductions in $\D$, so $\vdash_{\D} \Gamma_{\At} \Rightarrow \Delta_{\At}$. Now assume that $\Pi$ and $\Pi'$ do use instances of the new rules. Then we conjoin the deductions and use the following transformation to systematically remove all occurrences of the new rules of $\D^{\neg}$ and $\D^{\neg \neg}$:

\begin{prooftree}
\AxiomC{}
\RightLabel{$[\Omega_{\At}, p^{B} \Rightarrow \varnothing]$}
    \UnaryInfC{$\Omega_{\At}, p^{B} \Rightarrow \varnothing$}
    \noLine
    \UnaryInfC{$\Sigma$}
    \noLine
    \UnaryInfC{$\Gamma_{\At} \Rightarrow \Delta_{\At}$}
    \DisplayProof
    \qquad 
      \AxiomC{$\Sigma'$}
    \noLine
    \UnaryInfC{$\Theta_{\At}, p^{B} \Rightarrow\varnothing$}
    \RightLabel{$[\Theta_{\At}, p^{B} \Rightarrow \varnothing / \Theta_{\At} \Rightarrow \varnothing]$}
    \UnaryInfC{$\Theta_{\At} \Rightarrow \varnothing$}
    \noLine
    \UnaryInfC{$\Sigma''$}
\end{prooftree}

are transformed into

\begin{center}
\makebox[0pt][c]{%
\AxiomC{}
\RightLabel{$[/ \varnothing \Rightarrow p^{\neg B \oplus \neg \neg B }]$}
\UnaryInfC{$\Rightarrow p^{\neg B \lor \neg \neg B}$}
    \AxiomC{}
    \RightLabel{\init}
    \UnaryInfC{$\Omega_{\At}, p^{B} \Rightarrow p^{B}$}
    \AxiomC{}
    \RightLabel{L$ \bot$}
    \UnaryInfC{$p^{\bot} \Rightarrow \varnothing$}
    \RightLabel{L $\multimap$}
    \BinaryInfC{$ \Omega_{\At},  p^{B},  p^{\neg B} \Rightarrow \varnothing$}
    \noLine
    \UnaryInfC{$\Sigma_{*}$}
    \noLine
    \UnaryInfC{$p^{\neg B}, \ldots, p^{\neg B}, \Gamma_{\At} \Rightarrow \Delta_{\At} $}
       \RightLabel{LC}
    \doubleLine
    \UnaryInfC{$p^{\neg B}, \Gamma_{\At} \Rightarrow \Delta_{\At} $}
     \AxiomC{$\Sigma_{*}'$}
    \noLine
    \UnaryInfC{$\Theta_{\At}, p^{B} \Rightarrow  \varnothing$}
    \RightLabel{R $\bot$}
    \UnaryInfC{$\Theta_{\At}, p^{B} \Rightarrow \bot$}
    \RightLabel{R $\multimap$}
    \UnaryInfC{$\Theta_{\At} \Rightarrow p^{\neg B}$}
    \AxiomC{}
    \RightLabel{L $\bot$}
    \UnaryInfC{$\bot  \Rightarrow \varnothing$}
    \RightLabel{L $\multimap$}
    \BinaryInfC{$\Theta_{\At}, p^{\neg \neg B} \Rightarrow \varnothing$}
    \noLine
    \UnaryInfC{$\Sigma_{*}''$}
    \noLine
    \UnaryInfC{$ p^{\neg \neg B}, \ldots, p^{\neg \neg B}, \Gamma_{\At} \Rightarrow \Delta_{\At} $}
       \RightLabel{LC}
    \doubleLine
    \UnaryInfC{$ p^{\neg\neg  B} ,\Gamma_{\At} \Rightarrow \Delta_{\At} $}
    \RightLabel{L $\oplus$}
    \BinaryInfC{$p^{\neg B \oplus \neg \neg B}, \Gamma_{\At} \Rightarrow \Delta_{\At}$}
    \RightLabel{\Cuti}
    \BinaryInfC{$\Gamma_{\At} \Rightarrow \Delta_{\At}$}
\DisplayProof}
\end{center}

Where double lines represent successive applications of LC and the $\Sigma_{*}$ are obtained by replacing all other occurrences of those rules in the same way and adding $p^{\neg B}$ and $p^{\neg \neg B}$ to contexts, which can be done because $\D$ is cumulative by definition. Since we have chosen arbitrary $p = p^{B}$ and $\D$ and shown $\vdash_{\D} \Gamma_{\At} \Rightarrow \Delta$  we conclude that $\E$ is a weakly Jankov base

\end{proof}

\begin{proof}[of Theorem~\ref{thm:weakcompleteness}]
    Let $\U$ be a simulation base defined for $(\Gamma \cup \Delta)^{S}_{\At}$, and let $\E$ be its extension as defined in the lemma above. Since by Lemma \ref{mainlemmacompletenessadd} we have that $\Vdash^{p^{A}}_{\E} A$ holds for all $A \in \Gamma$ we obtain $\Vdash^{\Gamma_{\At}}_{\E} \Delta$ and then $\vdash_{\E} \Gamma_{\At} \Rightarrow \Delta_{\At}$, where $\Gamma_{\At}$ and $\Delta_{\At}$ are obtained by substituting $A$ by $p^{A}$ in $\Gamma$ and $\Delta$. Notice that $\E$ only contain new rules concluding $\Sigma_{\At} \Rightarrow p^{B}$, where $B$ is either a formula in $\Gamma$ or a Jankov formula, so we can add $p^B$ to all left contexts to turn them into applications of \init \ and the deduction into a deduction in $\Int$, possibly also removing occurrences later via contraction.  By substituting all atoms with shape $p^{C}$ by $C$ in the proofs we obtain a derivation of $\Theta, \Gamma \Rightarrow A$ in a calculus for intuitionistic logic (with arbitrary left contexts). Since $\Theta$ is a Jankov set, by Lemma \ref{lemma:derivJankov} and Corollary \ref{cor:Jankovcor} conclude that $\Gamma \Rightarrow A$ is derivable in Jankov logic.
\end{proof}

\begin{proof}[of Lemma~\ref{mainlemmaJankovstrong}]
  
Since from Lemma \ref{mainlemmacompletenessadd} it follows that $\vdash_{\E} \Rightarrow p^{C}$ implies $\Vdash^{\varnothing}_{\E} C$, the right to left direction follows from the definition of $\E$. As for the other direction, assume $\Vdash^{\varnothing}_{\E} C$. The completeness lemma yields $\vdash_{\E} \Rightarrow p^{C}$. If $\Pi$ is the derivation of $\vdash_{\E} \Rightarrow p^{C}$, let $\Pi'$ be the derivation obtained by substituting every rule added to $\U$ with conclusion $\Rightarrow p^{B}$ by an instance $p^{B} \Rightarrow p^{B}$ and adding $p^{B}$ to the antecedent of all sequent below it, which is allowed because $\U$ contains one version of each rule for each possible extra left context. Then $\Pi'$ is a deduction showing $\vdash_{U} \Omega_{\At} \Rightarrow p^{C}$, where $\Omega_{\At}$ is a finite set of atoms of shape $p^{A}$ for $A \in \Gamma'$. Replacing every $p^{D}$ in the deduction by $D$ we obtain a proof of the sequent $\Omega \Rightarrow C$ in intuitionistic logic, where $\Omega$ is a finite set with $\Omega \subset \Gamma'$. Then we have $\Gamma' \vdash_{\Int} C$, hence since prime theories are deductively closed we have $C \in \Gamma'$.

\end{proof}

\begin{proof}[of Lemma~\ref{lemma:isstronglyJankov}]
    Since $\Gamma'$ is a prime theory in the language $\mathbb{F}(\Gamma \cup \{A\})^{S}_{\At}$ and contains a formula of shape $\neg B \oplus \neg \neg B$ for all $B \in \mathbb{F}(\Gamma \cup \{A\})^{S}_{\At}$, for all $B \in \mathbb{F}(\Gamma \cup \{A\})^{S}_{\At}$ we have either $\neg B \in \Gamma'$ or $\neg \neg B \in \Gamma'$. Since $\alpha$ is a bijection we also have that for all atoms $p$ we have $p = p^{B}$ for some $B \in \mathbb{F}(\Gamma \cup \{A\})^{S}_{\At}$. If $\neg B \in \Gamma'$, by construction of $\E$ we have $\vdash_{\E} \Rightarrow p^{\neg B}$. Then we can do the following:

    \begin{prooftree}
        \AxiomC{}
        \RightLabel{$[/\Rightarrow p^{\neg B}]$}
        \UnaryInfC{$\Rightarrow p^{\neg B}$}
        \AxiomC{}
          \RightLabel{\init}
        \UnaryInfC{$p^{B} \Rightarrow p^{B}$}
        \AxiomC{}
          \RightLabel{L $\bot$}
        \UnaryInfC{$\bot \Rightarrow \varnothing$}
        \RightLabel{L $\multimap$}
        \BinaryInfC{$p,  p^{\neg B} \Rightarrow \varnothing$}
        \RightLabel{\Cuti}
        \BinaryInfC{$p^{B} \Rightarrow \varnothing$}
    \end{prooftree}
    
    So $\vdash_{\E} p^{B} \Rightarrow \varnothing$. Otherwise, $\neg \neg B \in \Gamma'$ and so $\vdash_{\E} \Rightarrow p^{\neg \neg B}$. In this case, pick any $\C \supseteq \E$ such that $\vdash_{\D} \Theta_{\At}, p^{B} \Rightarrow \varnothing$ for some $\Theta_{\At}$. Then we do the following:

    \begin{prooftree}
    \AxiomC{}
    \RightLabel{$[/ \Rightarrow p^{\neg \neg B}]$}
    \UnaryInfC{$\Rightarrow p^{\neg \neg B}$}
        \AxiomC{$\Sigma$}
        \noLine
        \UnaryInfC{$\Theta_{\At}, p^{B} \Rightarrow \varnothing$}
        \RightLabel{$R \bot$}
       \UnaryInfC{$\Theta_{\At}, p^{B} \Rightarrow \bot$}
       \RightLabel{L $\multimap$}
       \UnaryInfC{$\Theta_{\At} \Rightarrow p^{\neg B}$}
       \AxiomC{}
       \RightLabel{L $\bot$}
       \UnaryInfC{$\bot \Rightarrow \varnothing$}
       \RightLabel{L $\multimap$}
       \BinaryInfC{$\Theta_{\At}, p^{\neg \neg B} \Rightarrow \varnothing$}
       \RightLabel{\Cuti}
       \BinaryInfC{$\Theta_{\At} \Rightarrow \varnothing$}
    \end{prooftree}

So $\vdash_{\D} \Theta_{\At} \Rightarrow \varnothing$. Since the result was shown for arbitrary  $p = p^{B}$ we conclude that $\E$ is a strongly Jankov base.
\end{proof}

\begin{proof}[of Theorem~\ref{thm:strongcomplete}]
 Assume $\Gamma \nvdash_{\JK} A$. By Lemma \ref{lemma:derivJankov} we have $\Gamma, \Theta \nvdash_{\Int} A$, where $\Theta$ is the Jankov set for $\mathbb{F}((\Gamma \cup \{A\})^{S}_{\At})$ so there is a prime theory $\Gamma'$ with $(\Gamma \cup \Theta) \subseteq \Gamma'$ such that $\Gamma' \nvdash_{\Int} A$. By Lemma \ref{mainlemmaJankovstrong}, we can create a simulation base $\U$ and a base $\E \supseteq \U$ such that $\Vdash^{\varnothing}_{\E} C$ iff $C \in \Gamma'$. Since $\Gamma \subset \Gamma'$ and $A \notin \Gamma'$ (as $A \in \Gamma$ would imply $\Gamma' \vdash_{\Int} A$), we have $\Vdash^{\varnothing}_{\E} B$ for all $B \in \Gamma$ but $\nVdash^{\varnothing}_{\E} A$, so $\Gamma \nVdash^{\varnothing}_{\E} A$. Since $\E$ is a strongly Jankov base, $\Gamma \nVdash_{\SJK} A$, hence the result follows by contraposition.
\end{proof}

\end{proof}
\end{document}